%% file: main.tex
\documentclass[journal]{IEEEtran}
\IEEEoverridecommandlockouts

\input{./definitions.tex}

\input{./tikzStyles.tex}

\usepackage[final]{pdfpages}    

\newcommand{\lineref}[1]{(\tikzexternaldisable\ref{#1}\tikzexternalenable)}	


\usepackage[hidelinks]{hyperref}
\usepackage{cite}
\usepackage{amsmath,amssymb,amsfonts}
\usepackage{algorithmic}
\usepackage{graphicx}
\usepackage{textcomp}
\usepackage{xcolor}
\usepackage{cite}
\usepackage{amsmath,amssymb,amsfonts}
\usepackage{algorithmic}
\usepackage{graphicx}
\usepackage{textcomp}
\usepackage{xcolor}
\usepackage{orcidlink}
\usepackage{cuted}          
\usepackage{mathtools}      
\usepackage{enumitem}       
\usepackage{comment}

\usepackage{array, booktabs, xltabular}
\usepackage{colortbl}         
\usepackage{multirow}	      
\usepackage{longtable}	      
\usepackage{nicefrac}	      
\usepackage{adjustbox}        
\usepackage{upgreek}          
\usepackage[per-mode=symbol]{siunitx}    		
\DeclareSIUnit{\dBm}{dBm}	
\usepackage{balance}          
\usepackage{graphicx}         
\usepackage{placeins}         
\usepackage{multirow}

\usepackage[
    indexonlyfirst, 
    section=section,
    nonumberlist,numberedsection
]{glossaries}
\makenoidxglossaries%
\loadglsentries{abbr}%
\usepackage{ifthen}
\newcommand{\exportFigures}{false}       

\usepackage{etoolbox}
\usepackage{bibunits}
\defaultbibliographystyle{IEEEtran}

\newcommand{\herm}{\mathsf{H}}

\newcommand{\trp}{\mathsf{T}}

\newcommand{\realset}[2]{ \mathbb{R}^{#1 \times #2}  }
\newcommand{\realsetone}[1]{ \mathbb{R}^{#1}  }

\newcommand{\complexsetone}[1]{ \mathbb{C}^{#1}  }

\DeclareMathOperator{\tr}{tr}
\DeclareMathOperator{\arctantwo}{arctan2}
\newcommand{\condIndep}{\mathrel{\perp\mspace{-10mu}\perp}} 
\newcommand{\eye}[1]{\mathbf{I}_{\scriptscriptstyle#1}}       

\newcommand{\CN}{{\mathcal{CN}}}

\input{math-notation}   

\ifthenelse{\equal{\exportFigures}{false}}
{
   \newcommand{\tikzexternaldisable}{}  
   \newcommand{\tikzexternalenable}{}   
  }
{
  \usepgfplotslibrary{external}
  \tikzexternalize[prefix=tikz/]
  \usepackage{shellesc}
}

\newcounter{MYtempeqncnt}
\usepackage{amsthm}
\newtheorem{proposition}{Proposition}

\newlength\figureheight
\newlength\figurewidth 
\newlength\verticalSpace

\newsavebox{\foobox}
\newcommand{\slantbox}[2][0]{\mbox{%
        \sbox{\foobox}{#2}%
        \hskip\wd\foobox
        \pdfsave
        \pdfsetmatrix{1 0 #1 1}%
        \llap{\usebox{\foobox}}%
        \pdfrestore
}}
\newcommand\unslant[2][-.25]{\slantbox[#1]{$#2$}}

\usepackage{algorithm2e}		
\RestyleAlgo{ruled}			
\SetKw{KwBy}{by}			%
\SetKwComment{Comment}{\% }{}	

\usepackage[hang,flushmargin]{footmisc}
\newcommand{\algorithmfootnote}[2][\footnotesize]{%
  \let\old@algocf@finish\@algocf@finish
  \def\@algocf@finish{\old@algocf@finish
    \leavevmode\rlap{\begin{minipage}{\linewidth}
    #1#2
    \end{minipage}}%
    \vspace{-0.3cm} 
  }%
}

\newcolumntype{C}{@{\hskip 0.075cm}c@{\hskip 0.075cm}}

\def\BibTeX{{\rm B\kern-.05em{\sc i\kern-.025em b}\kern-.08em
    T\kern-.1667em\lower.7ex\hbox{E}\kern-.125emX}}

\begin{document}
\allowdisplaybreaks[4] 

\title{
Coherent Direct D-MIMO Localization:\\
Analysis of Coherence Levels
\thanks{
\IEEEauthorrefmark{1}The AMBIENT-6G project has received funding from the Smart Networks and Services Joint Undertaking 
under the European Union's Horizon Europe research and innovation programme (Grant Agreement No. 101192113).}
}


\allowdisplaybreaks
\frenchspacing
	\author{\IEEEauthorblockN{Benjamin J.\,B. Deutschmann, Lukas D'Angelo, Erik Leitinger, Klaus Witrisal}
	
    \IEEEauthorblockA{
	Institute of Comm. Networks and Satellite Comms., Graz University of Technology, Graz, Austria \\ 
    }
	\vspace*{-6mm}}

\maketitle

\begin{abstract}
Distributed multiple-input multiple-output (D-MIMO) offers favorable geometry for localization and sensing, with its greatest potential arising from joint coherent processing across panels.
However, stringent frequency-synchronization and phase-calibration requirements, together with multimodal likelihood functions, complicate the estimation problem.
Consequently, most existing methods process the panels noncoherently, potentially sacrificing localization accuracy.
We present a unified family of Bayesian state filters based on concentrated Type-I and marginal Type-II likelihoods for wideband near-field D-MIMO systems.
The Type-I filters realize (i) noncoherent, (ii) coherent, and (iii) carrier-phase-based processing.
We show that a zero-mean Type-II model is inherently noncoherent under distributed processing, whereas observation stacking restores coherence.
A recently proposed nonzero-mean Type-II model adapts to the coherence available in the data, a property termed ``soft coherence.'' 
We derive posterior Cram\'er--Rao lower bounds (PCRLBs) for all three coherence levels and show that each level is fundamentally tied to the number of phase parameters used for positioning or treated as nuisance parameters.
Numerical results show that the filters closely approach their respective PCRLBs and that coherent processing substantially outperforms noncoherent processing.
Our particle-based implementations parallelize over particles and panels, scale linearly with the observed data, and achieve runtimes of tens of milliseconds per time step under GPU acceleration.
\end{abstract}

\glsresetall        

\section{Introduction}\label{sec:intro}
Accurate wireless localization is becoming a key capability of future radio infrastructures, supporting autonomous navigation, industrial automation, and location-aware services as part of 6G \gls{isac}~\cite{GonFurKalValDarSheSheBayWymProcIEEE2024}.
This development coincides with the emergence of \gls{xlmimo} infrastructures~\cite{XuLiuZhaMinCai:TSP2024,XuLiuZhaCaiWu:Arxiv2025,Hua:EUSIPCO2025,WuQuiSunWeiZhaEld:TSP2026,Deutschmann26WCM} and \gls{dmimo} systems~\cite{HuRusEdf:TWC2018,TenWymKesDeySveJCS2026,Deutschmann24SPAWC}, in which large numbers of antenna elements are distributed over large spatial regions.
Spatially separated antenna panels can be processed as a very large jointly coherent synthetic array whose aperture spans the entire infrastructure~\cite{Fascista23RadioStripesICC,Fascista25RadioStripes}.
This panel-based model likewise applies to \gls{xlmimo} arrays partitioned into subarrays.
While potentially yielding substantial gains in localization accuracy, such coherent processing requires accurate channel models accounting for near-field spherical-wavefront propagation~\cite{GueGuiDarDju:TSP2021}, likelihood models that permit coherent data fusion across distributed panels, and estimators that can cope with both spatial nonstationarity~\cite{LiaLeiMey:TSP2025,Deutschmann26TSP} and the likelihood multimodality (i.e., grating lobes) inherent to the spatially aliased synthetic array.
These accuracy gains can be realized if the infrastructure fulfills the necessary frequency-synchronization and phase-calibration requirements~\cite{Osorio25distributedISAC}.
At the same time, scaling up wireless infrastructures places stringent demands on algorithmic scalability.
As antenna deployments become increasingly massive~\cite{Bjornson25GiganticMIMO} and signal bandwidths continue to grow~\cite{Sarajlic23subTHz}, the amount of observed data is ever increasing.
Localization algorithms must therefore scale favorably with the data and allow sufficient parallelization on modern computing architectures~\cite{Iancu25distributedMessagePassing} to meet the stringent runtime  requirements of real-time operation.

\subsection{State of the Art}

Model-based wireless localization methods can broadly be divided into two families. 
\emph{Two-step} approaches~\cite{richter2005estimation,ChuJSTSP2019,GreLeiWitFle:TWC2024} first extract parametric channel estimates---such as delays, angles, Doppler shifts, and complex amplitudes---from the received signals and subsequently perform sequential Bayesian filtering~\cite{Meyer16DistrTracking,Iancu25distributedMessagePassing,Venus24biasTracking,LeitMeyHlaWitTufWin:TWC2019}. 
However, the data compression inherent to the preceding channel estimation and detection stage can lead to a loss of position-relevant information, particularly at low \gls{snr}~\cite{Weiss04DPD,Mingchao23TBT}.

In contrast, \emph{direct} approaches \cite{Hadaschik15coherentDMIMO,Garcia17directLocalization,ZhaStaJosWanGenDamWymHoeTAES2020, LiaMey:Asilomar2024, Deutschmann24SPAWC, Fascista25RadioStripes} operate directly on the raw received signals and embed the physical channel model within the statistical inference engine. 
By jointly processing the received signals within a position-dependent probabilistic model, these methods avoid the information loss associated with a preliminary channel-estimation and detection stage and enable principled exploitation of low-\gls{snr} signal components.
We hereafter refer to the spatially separated antenna panels as \glspl{pa}.

\paragraph*{Coherence Levels}
We distinguish three coherence levels.
\Gls{nc} processing assumes no phase calibration among \glspl{pa} and retains one nuisance phase per \gls{pa} and per time step,
\gls{c} processing assumes 
phase calibration among \glspl{pa} and retains one common nuisance phase,
and \gls{cp} processing additionally assumes phase calibration w.r.t. the \gls{mt} clock and retains no nuisance phase.

Direct approaches provide a natural foundation for coherent Bayesian fusion 
at the raw-signal level.
While noncoherent direct wireless localization methods \cite{Oispuu10coherentNoncoherent,Hadaschik15coherentDMIMO,Garcia17directLocalization,Kamil17distributedTracking} are the most common, coherent methods~\cite{Hadaschik15coherentDMIMO,Vokomirovic19carrierPhasePos,Fascista23RadioStripesICC,Fascista25RadioStripes,Oispuu10coherentNoncoherent,Godrich10coherentDMIMOradar} and even carrier-phase-based (albeit two-step) approaches ~\cite{Talvitie23CarrierPhasePCRLB} have been explored.
\Glspl{crlb} for snapshot-based estimators and \glspl{pcrlb} for Bayesian state filters have been derived for noncoherent methods \cite{Oispuu10coherentNoncoherent,Kamil17distributedTracking,YuanShen10EFIM,ShenWin10widebandArrayCRLB,YanjunHan16EFIM}, coherent \gls{dmimo} approaches~\cite{Godrich10coherentDMIMOradar,Shourezari26CoherentPos,Wymeersch23BoundsCP,Vukomirovic21DirectCohrentDMIMO,Oispuu10coherentNoncoherent} and near-field \gls{xlmimo}~\cite{GueGuiDarDju:TSP2021} approaches.
A carrier-phase-based snapshot \gls{crlb} was derived in~\cite{HuRusEdf:TWC2018}, and a \gls{pcrlb} for two-step filtering was derived in~\cite{Talvitie23CarrierPhasePCRLB}, while \glspl{pcrlb} for direct carrier-phase-based filtering remain largely unexplored.

\ifthenelse{\equal{\IEEEversion}{true}}
{%
}%
{%
The authors of~\cite{Hadaschik15coherentDMIMO} compared noncoherent and coherent \gls{dmimo} snapshot estimators and found that the impairment caused by likelihood multimodality can be mitigated through particle filtering.
The authors of~\cite{HuRusEdf:TWC2018} discovered a large performance gain in \gls{xlmimo} systems when transitioning from a coherent setup to a carrier-phase-based setup.
In~\cite{Fascista23RadioStripesICC,Fascista25RadioStripes}, we discovered another large performance gap in \gls{dmimo} systems when transitioning from a coherent to a noncoherent setup.
}

Addressing the remaining challenges requires
(i) a framework of estimators and \glspl{pcrlb} that unifies the defined coherence levels, 
(ii) likelihood models that preserve and exploit the phase coherence available in the data while allowing \gls{pa}-local processing, and
(iii) scalable algorithms capable of processing wide bandwidths and large antenna arrays in real time.

\subsection{Contributions}
In this paper, we propose a unified family of Bayesian state filters spanning the defined coherence levels and derive the corresponding \glspl{pcrlb}. Joint coherent processing is key to exploiting the full distributed aperture for \gls{dmimo} positioning and is the central objective of this paper.
We show that the coherence levels are fundamentally tied to the number of nuisance phases in the complex signal amplitudes.
Coherent \gls{dmimo} processing therefore requires retaining the phase relations between distributed \glspl{pa}, making the treatment of complex amplitudes central.
Nuisance amplitudes are traditionally modeled as deterministic or stochastic unknowns and eliminated from the likelihood by \textit{concentration} or \textit{marginalization}~\cite{KrimViberg96ASP}.
For a general \gls{dmimo} positioning and tracking scenario, we show how the amplitudes must be treated to achieve the desired coherence level using either (i) a concentrated \textit{Type-I} likelihood or (ii) a \textit{Type-II} marginal likelihood. We implement the resulting Bayesian state filters using particle-based \gls{bp}, which is naturally parallel and can be accelerated using a \gls{gpu}.
The main contributions are summarized as follows.
\begin{itemize}
\item We present seven direct estimators spanning different coherence levels for wideband, near-field localization with distributed array-equipped \glspl{pa}, applicable to \gls{dmimo} and panelized \gls{xlmimo}.
\item We derive a snapshot \gls{crlb} and three \glspl{pcrlb} and show that coherence and Fisher information are fundamentally tied to the number of phase parameters used for positioning or treated as nuisance parameters.
\item We introduce three Bayesian state filters that are based on concentrated Type-I likelihoods and closely approach their respective noncoherent, coherent, and carrier-phase-based \glspl{pcrlb}.
\item We show that the commonly used zero-mean Type-II likelihood is inherently noncoherent under distributed processing, whereas a nonzero-mean model retaining shared phase information enables coherence. 
\item We demonstrate that this nonzero-mean Type-II model automatically performs noncoherent processing if the data are noncoherent---a property that we term \textit{soft coherence}---where other coherent estimators diverge. 
\item Our particle-based \gls{bp} method scales linearly with the observed data length and the numbers of \glspl{pa} and particles, while its \gls{gpu} implementation runs in tens of milliseconds per snapshot.
\end{itemize}

\subsection{Scope and Paper Organization}
This paper discusses and compares the coherence of direct \gls{dmimo} estimators based on several commonly used Type-I and Type-II likelihoods.
For an accessible presentation tightly tailored to this scope, we consider simple \gls{los} channels in \gls{awgn}.
In~\cite{Deutschmann26TSP}, we demonstrated the applicability of the modeling and inference framework presented here to specular multipath channels by extending it to~\gls{slam}.
The presented methods likewise extend to multiobject tracking~\cite{LiaMey:Asilomar2024} or Bernoulli filters that accommodate spatial nonstationarity.

Throughout this paper, we use \glspl{fg}~\cite{Loeliger04IntroFG} both to represent the statistical models underlying our Bayesian state filters and to derive the corresponding \gls{bp} messages. 
\Gls{bp} message passing provides a principled way to compute marginal posterior \glspl{pdf} required for Bayesian state estimation.
We then present numerically feasible implementations using particle-based \gls{bp}.

The remainder of the paper is organized as follows.  
Section~\ref{sec:System-Model} introduces the geometric, signal, and statistical models underlying the \gls{dmimo} localization and state-filtering problem.
Section~\ref{sec:problem-formulation} describes the proposed \gls{bp}-based inference algorithm, whose particle-based evaluation is presented in Section~\ref{sec:particle-based-implementation}.
Sections~\ref{sec:Type-I} and~\ref{sec:Type-II} introduce likelihood models and estimators for varying coherence levels.
Section~\ref{sec:PCRLB} derives the \glspl{pcrlb} bounding the coherence-level-dependent estimation performance of our proposed state filters.  
Section~\ref{sec:Results} presents numerical experiments and performance evaluations, which are discussed in Section~\ref{sec:discussion}. 
Finally, Section~\ref{sec:conclusion} concludes the paper.

\textit{Notation:} 
Scalars are denoted by lowercase letters $\mathrm{x}$, 
column vectors by bold lowercase letters $\mathbf{x}$, and matrices by bold uppercase letters $\mathbf{X}$. 
Regular upright font is used for deterministic constants.
\Glspl{rv} are typeset in upright sans-serif font, e.g., $\rv{x}$ and $\RV{x}$, and their realizations in italic serif font, e.g., $x$ and $\V{x}$. 
$f(\V{x})$, shorthand for $f_{\RV{x}}(\V{x})$, denotes the \gls{pdf} of a continuous \gls{rv} $\RV{x}$.
$f(\V{x}|\V{y})$, shorthand for $f_{\RV{x}|\RV{y}}(\V{x}|\V{y})$, denotes the conditional \gls{pdf} of $\RV{x}$ given $\RV{y}$.
We use $\bm{x}^\trp$ and $\bm{x}^\herm$ to denote the transpose and Hermitian transpose of $\bm{x}$, respectively.
The Hadamard product is denoted by $\odot$ and the Kronecker product is denoted by $\otimes$. 
The Dirac distribution is $\delta(\cdot)$.  
The $N\!\times\!N$ identity matrix is denoted by $\eye{N}$.

\begin{figure}
        \vskip 0pt	
        \begin{center}
            \newcommand{\LWrays}{0.05pt}
            \def\datapath{.}
            \setlength{\figurewidth}{0.8\linewidth}
            \tikzexternaldisable    
	           \input{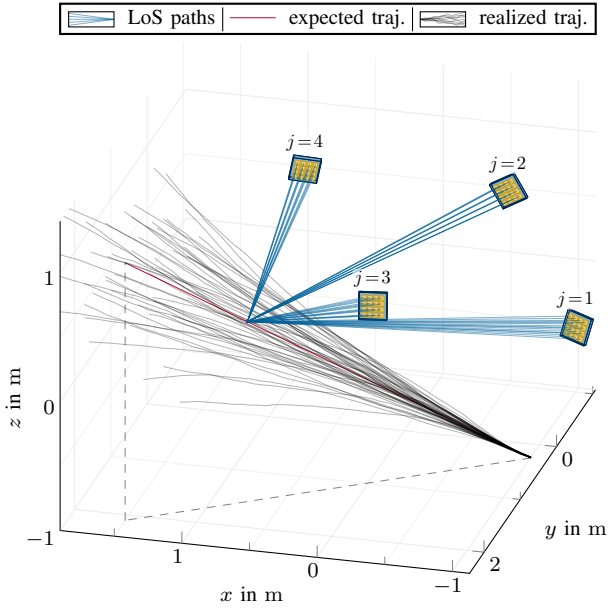}
            \tikzexternalenable
        \end{center}
        \vspace{-2mm}
        \captionof{figure}{
        Scenario for synthetic data generation with
        $J\!=\!4$ \glspl{pa} 
        numbered from right to left and 
        random realizations of \gls{mt} trajectories.
        }
        \label{fig:scenario}
    \vspace{-1mm}
\end{figure}

\section{System Model}\label{sec:System-Model}  

We consider a single-antenna \gls{mt} moving on a trajectory of unknown positions $\pos{n}\!\in\!\realsetone{3}$. 
At each time step $n$, the \gls{mt} transmits an uplink pilot with bandwidth $B$ at carrier frequency $\fc$, which is received by a set $\setAnchors\!\coloneqq\!\{1\hdots J\}$ of identical \glspl{pa}, each equipped with an $\Nantennas$-antenna \gls{ura}.
The \glspl{pa} have known positions $\posPA{j}\!\in\!\realsetone{3}$ and orientations modeled by rotation matrices 
$\rotM{j}\!\in\!SO(3)$. 
For simplicity, we assume that the \glspl{pa} and the \gls{mt} are jointly synchronized in time and frequency.

\subsection{State Vectors and Observation Model}\label{sec:observation-model}
For sequential Bayesian filtering, we are interested in the \gls{mt} state $\RVstate{n} \!\coloneqq\! [\RVpos{n}^\trp \iist \RVvel{n}^\trp]^\trp \!\in\! \realsetone{6}$, comprising the \gls{mt} position $\RVpos{n}$ and velocity $\RVvel{n}\!\in\!\realsetone{3}$, both of which are modeled as \glspl{rv}. 
The sequence of \gls{mt} states up to time $n$ is denoted by $\RVstate{0:n} \!\coloneqq\! [\RVstate{0}^\trp \iist \cdots \iist \RVstate{n}^\trp]^\trp$.
We consider an \gls{los} channel\footnote{The system model, methods, and \glspl{pcrlb} readily extend to multipath channels with spatial nonstationarity~\cite{Deutschmann26TSP}.} observation
\begin{align}\label{eq:observation}
    \observation{n}{j} = 
    \amplitude{s}{n}{j} 
    \steerVec{j}(\pos{n})
    +
    \noise{n}{j}
    \quad
    \in \complexsetone{\Nz}
\end{align}
at \gls{pa} $j$, modeled as a steering vector $\steerVec{j}(\pos{n})\!\in\!\complexsetone{\Nz}$ scaled by an amplitude $\amplitude{s}{n}{j}\!\in\!\complexsetone{}$ and corrupted by noise $\noise{n}{j}\!\in\!\complexsetone{\Nz}$.
Here, $\Nfrequency$ frequency-domain samples from each of the $\Nantennas$ antenna elements are stacked, leading to an observation length $\Nz\!=\!\Nfrequency\Nantennas$.
The steering vectors $\steerVec{j}(\pos{n})\!\in\!\complexsetone{\Nz}$ are obtained by using the \gls{mt} positions $\pos{n}$ to parameterize the array response, which we introduce below. 
The noise $\noise{n}{j}$ can represent different noise sources, potentially including a \gls{dmc}~\cite{Fascista25RadioStripes}. 
For simplicity, however, we assume temporally and spatially uncorrelated, circularly symmetric, complex \gls{awgn}, i.e., $\RVnoise{n}{j}|\RVetan{n}^{\scriptscriptstyle(j)} \!\sim\! \CN(\mathbf{0},\RVetan{n}^{\scriptscriptstyle(j)} \eye{\Nz})$ with noise variance $\RVetan{n}^{\scriptscriptstyle(j)}$, where the noise vectors are mutually independent across \glspl{pa} when conditioned on their noise variances. 
Consequently, the observations are conditionally independent, i.e., $\RVobservation{n}{\grave{\jmath} } \condIndep \RVobservation{n}{\acute{\jmath}}
\,\big|\,
\RVstate{n},\{\RVamplitude{s}{n}{j},\RVetan{n}^{\scriptscriptstyle(j)}\}_{j\in\setAnchors}$ for $\grave{\jmath}\neq \acute{\jmath}$.
The vector of stacked observations is $\RVobservationn{n} \!\coloneqq\! [{\RVobservation{n}{1}}^\trp \ist \hdots \ist {\RVobservation{n}{J}}^\trp]^\trp$ and the sequence of observations up to time $n$ is $\RVobservationn{1:n} \!\coloneqq\! [\RVobservationn{1}^\trp \iist \cdots \iist \RVobservationn{n}^\trp]^\trp$.

\paragraph*{Type-I Likelihood}
Treating the amplitudes $\amplitude{s}{n}{j}$ as deterministic latent variables leads to the Type-I likelihood function
\begin{align}\label{eq:likelihood-Type-I}
    f_{\scriptscriptstyle\mathrm{I}}\big(\observation{n}{j};\state{n},\amplitude{s}{n}{j},\etan{n}^{\scriptscriptstyle(j)}\big) 
    &= 
    \CN\big(\observation{n}{j};
    \amplitude{s}{n}{j}
    \steerVecx{s}{n}{j},
    \etan{n}^{\scriptscriptstyle(j)} \eye{\Nz}\big)
\end{align}
of the observations $\observation{n}{j}$, parameterized by the amplitudes $\amplitude{s}{n}{j}$ and noise variance $\etan{n}^{\scriptscriptstyle(j)}$, where $\steerVecx{s}{n}{j}\!\coloneqq\!\steerVec{j}(\pos{n})$ is a shorthand for a steering vector parameterized by the \gls{mt} state $\state{n}$.

\paragraph*{Type-II Marginal Likelihood}
Treating the amplitudes $\amplitude{s}{n}{j}$ as stochastic latent variables---often modeled using a complex Gaussian prior \gls{pdf} $\RVamplitude{s}{n}{j}|\RVPFmu{s}{n}{j},\RVPFgamma{s}{n}{j} \!\sim\! \CN(\RVPFmu{s}{n}{j},\RVPFgamma{s}{n}{j})$ with prior mean $\RVPFmu{s}{n}{j}\!\in\!\complexsetone{}$ and variance $\RVPFgamma{s}{n}{j}\!\in\!\realsetone{}_{\scriptscriptstyle\geq 0}$---leads to the Type-II marginal likelihood function
\begin{align}\label{eq:likelihood-Type-II}
    f_{\scriptscriptstyle\mathrm{II}}\big(\observation{n}{j}|\state{n},\PFphi{s}{n}{j},\etan{n}^{\scriptscriptstyle(j)}\big) 
    &= \CN\big(\observation{n}{j};
    \bm{\mu}_{\scriptscriptstyle n}^{\scriptscriptstyle (j)},
    \bm{C}_{\scriptscriptstyle n}^{\scriptscriptstyle (j)}\big)
\end{align}
with mean $\bm{\mu}_{\scriptscriptstyle n}^{\scriptscriptstyle (j)}\!=\!
\PFmu{s}{n}{j}
\steerVecx{s}{n}{j}$
and covariance matrix $\bm{C}_{\scriptscriptstyle n}^{\scriptscriptstyle (j)}\!=\!
\PFgamma{s}{n}{j}
\steerVecx{s}{n}{j}
{\steerVecx{s}{n}{j}}^\herm
+
\etan{n}^{\scriptscriptstyle(j)} \eye{\Nz}$
upon analytical marginalization over the amplitudes $\RVamplitude{s}{n}{j}$.
Conditioned on their prior parameters, the amplitudes $\RVamplitude{s}{n}{j}$ are modeled as independent across time $n$ and i.i.d. across \glspl{pa} $j$.
Together with the conditional independence of the noise vectors, this implies that the Type-II likelihood of the stacked observation factorizes over \glspl{pa}.
For brevity, we define the amplitude state as
$\RVPFphi{s}{n} \!\coloneqq\! [ \RVPFgamma{s}{n}{j} \iist \RVPFmu{s}{n}{j}]^\trp \!\in\! \realsetone{}_{\scriptscriptstyle\geq 0} \!\times\! \complexsetone{}$.
While a nonzero-mean prior is necessary for coherent \gls{dmimo} data fusion in the likelihood domain~\cite{Deutschmann26TSP}, temporal filtering of these prior means leverages slow-time Doppler information. 

The key to coherent \gls{dmimo} processing---hence aperture gain---lies in retaining a common phase relation across distributed \glspl{pa} that takes the distance-dependent (i.e., near-field) phase shifts to the \gls{mt} into account.

\subsection{Array Response}\label{sec:array-response}

Let $\acute{\V{r}}\!\coloneqq\![\rx\iist\ry\iist\rz]^\trp$ be a vector pointing from the \gls{pa} phase center to the \gls{mt} position in local Cartesian \gls{pa} coordinates, i.e., the \gls{pa}-local frame of reference. 
For data fusion based on the amplitude moduli, we define a path-loss-compensated array response
\begin{align}\label{eq:array-response}
    \V{\psi}(\acute{\V{r}})
    &\coloneqq 
    \frac{\lambda}{\sqrt{4 \pi}}
    \frac{1}{\sqrt{4 \pi} \lVert \acute{\V{r}} \rVert}
    \widetilde{\V{\psi}}(\acute{\V{r}})
\end{align}
with wavelength $\lambda=\frac{\lightspeed}{\fc}$ and
propagation velocity $\lightspeed$.
For data fusion based on the amplitude phases, the carrier-phase-based unit-modulus array response is given by
\begin{align}\label{eq:unit-modulus-array-response}
    \widetilde{\V{\psi}}(\acute{\V{r}})
    &\!\coloneqq\! \Big(\bm{b}\big(\delay(\acute{\V{r}})\big)  
        \otimes 
        \bm{a}\big(\elevation(\acute{\V{r}}),\azimuth(\acute{\V{r}})\big)\Big)
        \exp \!\Big(\!\!-\!\mathrm{j}\frac{2\pi}{\lightspeed}\fc 
        \lVert \acute{\V{r}}\rVert 
        \Big)
        \in \mathbb{T}^{\Nz}
\end{align}
where $\mathbb{T}^{\Nz}\!\coloneqq\!(\mathbb{S}^1)^{\Nz}$ denotes the $\Nz$-torus and $\mathbb{S}^1\!\coloneqq\!\{x\!\in\!\mathbb{C}\!:\!|x|\!=\!1\}$ the unit circle.
The array response in \eqref{eq:unit-modulus-array-response} is computed from the Kronecker product of the temporal array response
\begin{align}\label{eq:delay-array-response}
    \bm{b}(\delay) \coloneqq 
    \exp\big(\!-\!\mathrm{j} 2 \pi \mathbf{f} \delay\big) \quad \in \mathbb{T}^{\Nfrequency}\,,
\end{align}
in delay $\delay(\acute{\V{r}}) = \lVert \acute{\V{r}} \rVert / \lightspeed$
and the spatial array response in \gls{aoa}~\cite{richter2005estimation,DeutschmannCISA2025}, $\bm{a}(\elevation,\azimuth)\!\coloneqq\! \bm{a}_y(\elevation,\azimuth)\!\otimes\!\bm{a}_z(\elevation)$ in elevation $\elevation(\acute{\V{r}}) = \arccos (\rz/ \lVert \acute{\V{r}} \rVert)$ and azimuth angle $\azimuth(\acute{\V{r}}) = \arctantwo (\ry,\rx)$ in \gls{pa}-local spherical coordinates.
With \gls{ura}-equipped \glspl{pa}, the latter is itself modeled as a Kronecker product of the horizontal spatial array response
\begin{align}\label{eq:spatial-response-y}
    \bm{a}_y(\elevation,\azimuth) \coloneqq \exp\Big(\mathrm{j} \frac{2 \pi}{\lambda} \mathbf{p}_y \sin(\elevation) \sin(\azimuth) \Big) \quad \in \mathbb{T}^{\Nantennasy}
\end{align}
and the vertical spatial array response
\begin{align}\label{eq:spatial-response-z}
    \bm{a}_z(\elevation) \coloneqq \exp\Big(\mathrm{j} \frac{2 \pi}{\lambda} \mathbf{p}_z \cos(\elevation) \Big) \quad \in \mathbb{T}^{\Nantennasz}\,.
\end{align}
Let $\mathbf{f}\!\in\!\realsetone{\Nfrequency}$ be a vector of $\Nfrequency$ baseband frequencies equally spaced across bandwidth $B$, $\mathbf{p}_y\!\in\! \realsetone{\Nantennasy}$ a vector of $\Nantennasy$ horizontal position offsets, and $\mathbf{p}_z\!\in\!\realsetone{\Nantennasz}$ a vector of $\Nantennasz$ vertical offsets of a ``template'' \gls{ura}. 
This leads to an observation length $\Nz\!=\!\Nfrequency\Nantennas$ where $\Nantennas\!=\!\Nantennasy\Nantennasz$.
We require each of the support vectors $\{\mathbf{f}, \mathbf{p}_y, \mathbf{p}_z\}$ to be symmetric around $0$.
While assuming plane-wave propagation at the \gls{pa} level, \eqref{eq:array-response} accommodates spherical-wavefront near-field processing at the infrastructure level.
See~\cite[Sec.\,S-II]{Deutschmann26TSP} for a detailed derivation of this channel model.
Steering vectors $\steerVec{j}(\pos{n}) \!\coloneqq \V{\psi}(\rangep{k}{n}{j})$ are elements of the array manifold $\mathcal{M}\!\coloneqq\!\big\{\V{\psi}(\acute{\V{r}})|\acute{\V{r}}\!\in\! \mathbb{R}^3\setminus\{\mathbf{0}\} \big\}\!\subset\!\complexsetone{\Nz}$, obtained by parameterizing the array response in~\eqref{eq:array-response} with vectors 
\begin{align}\label{eq:rangep}
    \rangep{k}{n}{j}(\pos{n},\posPA{j},\rotM{j}) 
    \!\coloneqq\! \rotM{j}^{-1}\big(\pos{n}\! -\! \posPA{j}\big)  \in \realsetone{3}\,
\end{align}
pointing from \gls{pa} $j$ to the \gls{mt} in \textit{local} Cartesian \gls{pa} coordinates, i.e., the local frame of reference of \gls{pa} $j$.
Absorbing the \textit{known} \gls{pa} position $\posPA{j}$ and orientation $\rotM{j}$, the \gls{pa}-dependent steering vectors $\steerVec{j}(\pos{n})$ depend only on the \gls{mt} position $\pos{n}$.

\subsection{State-Transition Models}\label{sec:state-transition-model}
All states evolve independently.
We assume first-order Markovity and conditional independence of the state from past measurements, i.e.,
$\RVstate{n}\!\condIndep\!
\big(\RVstate{0:n\!-\!2},\RVobservationn{1:n\!-\!1}\big)
\mid\RVstate{n\!-\!1}$.
Hence the \gls{mt} state evolves according to
$f(\state{n}|\state{0:n-1},\observationn{1:n-1}) =f(\state{n}|\state{n-1})$.
Under the same assumptions, the amplitude state and noise states evolve according to state-transition \glspl{pdf} $f(\PFphi{s}{n}{1}|\PFphi{s}{n-1}{J})$ and $f(\etan{n}^{\scriptscriptstyle(j)}|\etan{n-1}^{\scriptscriptstyle(j)})$, respectively.

\section{Problem Formulation}\label{sec:problem-formulation}  
In this section, we consider the nonzero-mean Type-II likelihood model, which is the most comprehensive model, while peculiarities of the Type-I likelihood model are discussed in Section~\ref{sec:Type-I}.
We aim to jointly estimate the \gls{mt} state $\RVstate{n}$, the amplitude state $\RVPFphi{s}{n}$, and the noise variances $\RVetan{n}^{\scriptscriptstyle(j)}$ for all \glspl{pa} $j \in \{1\iist\dots\iist J\}$ based on the observed 
measurements $\observationn{1:n}$.

\subsection{State Estimation}\label{sec:estimation}
In Bayesian inference, estimates can be computed from marginal posterior \glspl{pdf} $f(\state{n}|\observationn{1:n})$,
$f(\PFphi{s}{n}{j}|\observationn{1:n})$, and 
$f(\etan{n}^{\scriptscriptstyle(j)}|\observationn{1:n})$, by means of the \gls{mmse} estimator, i.e., as posterior means~\cite{kay1993estimation}
\begin{align}
    \stateHat{n}^{\text{\tiny MMSE}} &= \E({\RVstate{n}|\RVobservationn{1:n}\!=\!\observationn{1:n}})\! = 
    \!\int\!\! \state{n} f(\state{n}|\observationn{1:n}) \mathrm{d}\state{n}\,, \ist
    \label{eq:stateHat}\\
    {\PFphiHat{s}{n}{j}}^{\text{\tiny MMSE}} &= \E({\RVPFphi{s}{n}|\RVobservationn{1:n}\!=\!\observationn{1:n}})\! = 
    \!\int\!\! \PFphi{s}{n}{j} f(\PFphi{s}{n}{j}|\observationn{1:n}) \mathrm{d}\PFphi{s}{n}{j}\,, \ist
    \label{eq:PFphiHat}\\
    \etanHat{n}^{\,\text{\tiny MMSE}\scriptscriptstyle(j)} &= 
    \E({\RVetan{n}^{\scriptscriptstyle(j)}|\RVobservationn{1:n}\!=\!\observationn{1:n}})\! = \!\int\!\! \etan{n}^{\scriptscriptstyle(j)} f(\etan{n}^{\scriptscriptstyle(j)}|\observationn{1:n}) \mathrm{d}\etan{n}^{\scriptscriptstyle(j)} \ist.
    \label{eq:etanHat} 
\end{align}

\begin{figure}[h]%
    \centering%
    \includegraphics[width = \linewidth]{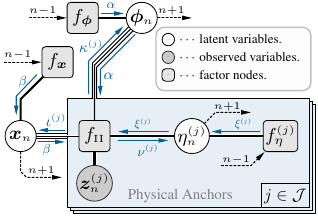}%
    \vspace{-0.2cm}\caption{
    FG representing the joint posterior \gls{pdf} from \eqref{eq:joint-post-PDF} and
    corresponding to VII) the NZM filter with Type-II likelihood $f_{\scriptscriptstyle\mathrm{II}}\!:=\!f\big(\observation{n}{j}|\state{n},\PFphi{s}{n}{j},\etan{n}^{\scriptscriptstyle(j)}\big)$, 
    state-transition \glspl{pdf}
    $f_{\bm{x}}\!\coloneqq\!f(\state{n}|\state{n-1})$,
    $f_{\bm{\phi}}\!\coloneqq\!f(\PFphi{}{n}{}|\PFphi{}{n-1}{})$,
    $f_{\eta}^{\scriptscriptstyle(j)}\!\coloneqq\!f(\etan{n}^{\scriptscriptstyle(j)}|\etan{n-1}^{\scriptscriptstyle(j)})$. 
    Prediction messages are 
    $\alpha\!\coloneqq\!\alpha(\PFphi{s}{n}{j})$,
    $\beta\!\coloneqq\!\beta(\state{n})$, and
    $\Mxi^{\scriptscriptstyle(j)}\!\coloneqq\!\xi(\etan{n}^{\scriptscriptstyle(j)})$. 
    Update messages are 
    $\kappa^{\scriptscriptstyle(j)}\!\coloneqq\!\kappa(\PFphi{s}{n}{j};\observation{n}{j})$,
    $\iota^{\scriptscriptstyle(j)}\!\coloneqq\!\iota(\state{n};\observation{n}{j})$, and
    $\nu^{\scriptscriptstyle(j)}\!\coloneqq\!\nu(\etan{n}^{\scriptscriptstyle(j)};\observation{n}{j})$.
    }%
    \label{fig:factor-graph}%
\end{figure}%

\subsection{The Factor Graph}\label{sec:factor-graph}
Following the statistical model and assumptions in Section~\ref{sec:System-Model}, the joint posterior \gls{pdf} of $\RVstate{0:n}$, $\RVPFphi{s}{0:n}$, and $\RVetann{0:n}$ conditioned on the measurements $\observationn{1:n}$ can be factorized as shown in~\eqref{eq:joint-post-PDF}\addtocounter{equation}{1}. 
A single time step of the corresponding \gls{fg}~\cite{Kschischang01factorGraphs,Loeliger04IntroFG} is depicted in Fig.\,\ref{fig:factor-graph}. 
This factorization enables the development of an efficient method for computing approximate marginal posterior \glspl{pdf}, referred to as beliefs, i.e., $\belief(\state{n}) \approx f(\state{n}|\observationn{1:n})$, $\belief(\etan{n}^{\scriptscriptstyle(j)}) \approx f(\etan{n}^{\scriptscriptstyle(j)}|\observationn{1:n})$, and $\belief(\PFphi{s}{n}{J}) \approx f(\PFphi{s}{n}{J}|\observationn{1:n})$ as described below. 

\begin{figure*}[!t]
\normalsize
\setcounter{MYtempeqncnt}{\value{equation}}
\setcounter{equation}{12}       
\begin{align}\label{eq:joint-post-PDF}
    &f(\state{0:n},\PFphi{s}{0:n}{j},\etann{0:n}|\observationn{1:n})
    \propto 
    \nonumber \\[-1mm]
    &~~~~~
    \underbrace{
        f(\state{0})
        \Big(
            \prod_{j \in \setAnchors}
            f(\etan{0}^{\scriptscriptstyle(j)})
        \Big)
        f(\PFphi{s}{0}{J})
    }_{\text{Initial prior PDFs}} 
    \Bigg(
    \underbrace{
    \prod_{n'=1}^{n}
        \Big(\prod_{j \in \setAnchors}f(\etan{n'}^{\scriptscriptstyle(j)}|\etan{n'\!-1}^{\scriptscriptstyle(j)})\Big)
        f(\state{n'}|\state{n'\!-1})
        f(\PFphi{s}{n'}{1}|\PFphi{s}{n'\!-1}{J})
    }_{\text{State-transition PDFs}}
    \underbrace{
        \prod_{j \in \setAnchors}
        f_{\scriptscriptstyle\mathrm{II}}(\observation{n'}{j}|\state{n'}\!,\PFphi{s}{n'}{j}\!,\etan{n'}^{\scriptscriptstyle(j)})
    }_{\text{Likelihood functions}}
    \Bigg)
\end{align}
\setcounter{equation}{\value{MYtempeqncnt}}
\hrulefill
\vspace{-5mm}%
\end{figure*}

\subsection{BP Method}\label{sec:BP-method}
We compute beliefs through \gls{bp} message passing on the \gls{fg} in Fig.\,\ref{fig:factor-graph} using the
\gls{spa} rules~\cite{Kschischang01factorGraphs,Loeliger04IntroFG}.
On tree-structured \glspl{fg}, \gls{bp} message passing computes exact marginal posterior \glspl{pdf}.
However, the \gls{fg} in Fig.\,\ref{fig:factor-graph} has loops, i.e., cycles, between the \gls{mt} state $\state{n}$ and amplitude state $\PFphi{s}{n}{j}$ variable nodes via the edges connecting to likelihood factor nodes of different \glspl{pa}, and across time steps $n$ and $n+1$.
Iteratively applying \gls{bp} message passing on loopy graphs results in loopy \gls{bp}~\cite{Frey97loopyBP}, where beliefs approximate marginal posterior \glspl{pdf}.
On loopy graphs, \gls{bp} does not prescribe a unique message-passing schedule~\cite[Sec.\,V-A]{Kschischang01factorGraphs}.
We specify the following message schedule:

Messages are sent only forward in time from step $n\!-\!1$ to $n$, i.e., filtering rather than smoothing, and every message is computed only once, i.e., a single iteration.
Within this schedule, we define the following message-passing \textit{phases}:
\begin{enumerate}[label=(\roman*)]
    \item State-transition factor-node-to-variable-node prediction messages ($\Mbeta$, $\Mxi^{\scriptscriptstyle(j)}$, $\alpha$) from time $n-1$
    \item Variable-node-to-likelihood-factor-node prior messages ($\Mbeta$, $\Mxi^{\scriptscriptstyle(j)}$, $\alpha$)
    \item Likelihood factor-node-to-variable-node update messages ($\Miota{j}$, $\Mnu{j}$, $\Mkappa{s}{j}$)
\end{enumerate}
The phases are executed in \textit{series} according to the order specified above. 
Within each phase, all pending messages are transmitted in parallel.

Once all messages at time step $n$ have been computed, marginal posterior \glspl{pdf} are approximated by beliefs. 

\subsection{Belief Calculation}\label{sec:belief-calculation}
Beliefs of states are computed as the product of all messages received at the respective variable nodes, given by
\begin{align}
    \belief(\state{n})&\propto
    \beta(\state{n})\prod_{\scriptscriptstyle j\in\setAnchors} \iota(\state{n};\observation{n}{j})
    \label{eq:belief-state}
    \\
    \belief(\PFphi{s}{n}{j})
    &\propto 
    \alpha(\PFphi{s}{n}{j}) 
    \prod_{\scriptscriptstyle j\in \setAnchors} 
    \kappa(\PFphi{s}{n}{j};\observation{n}{j}) 
    \label{eq:belief-PFphi}
    \\
    \belief(\etan{n}^{\scriptscriptstyle(j)}) &\propto \xi(\etan{n}^{\scriptscriptstyle(j)}) 
    \nu(\etan{n}^{\scriptscriptstyle(j)};\observation{n}{j}) 
    \label{eq:belief-eta}
\end{align}
The beliefs are nonnegative and we require them to be appropriately normalized to integrate to one, making them \glspl{pdf}.
Consequently, prediction messages will likewise be \glspl{pdf}.
Approximating marginal posterior \glspl{pdf}, these beliefs can be used for state estimation according to Section~\ref{sec:estimation}.
First, we derive prediction messages in Section~\ref{sec:prediction-messages} and then we derive update messages in Section~\ref{sec:update-messages}.

\subsection{Prediction Messages}\label{sec:prediction-messages}

Prediction messages are factor-node-to-variable-node messages sent from state-transition (prior) factors to the respective state variables~\cite[eq.\,(6)]{Kschischang01factorGraphs}.
Inserting the state-transition \glspl{pdf} from Section~\ref{sec:state-transition-model}, the prediction \glspl{pdf} become
\begin{align}
    \!\!f(\state{n}|\observationn{1:n-1}) 
    &=
    \int \! f(\state{n}|\state{n-1}) f(\state{n-1}|\observationn{1:n-1}) \mathrm{d}\state{n-1}  \,,
    \label{eq:Mbeta-exact}
    \\
    f(\PFphi{s}{n}{j}|\observationn{1:n-1}) 
    &=
    \int \! 
    f(\PFphi{s}{n}{1}|\PFphi{s}{n-1}{j}) f(\PFphi{s}{n-1}{j}|\observationn{1:n\!-1}) \mathrm{d}\PFphi{s}{n-1}{j}
    \,,
    \label{eq:Malpha-exact}
    \\
    f(\etan{n}^{\scriptscriptstyle(j)}|\observationn{1:n-1}) &= \int \! f(\etan{n}^{\scriptscriptstyle(j)}|\etan{n-1}^{\scriptscriptstyle(j)}) f(\etan{n-1}^{\scriptscriptstyle(j)}|\observationn{1:n-1}) \mathrm{d}\etan{n-1}^{\scriptscriptstyle(j)}
    \label{eq:Mxi-exact}
    \,.
\end{align}
Approximating the true marginal posterior \glspl{pdf} through their respective beliefs leads to the prediction messages 
$\beta(\state{n})\!=\!\int \! f(\state{n}|\state{n-1}) \belief(\state{n-1}) \mathrm{d}\state{n-1}$,
$\alpha(\PFphi{s}{n}{j})\!=\!\int \! 
f(\PFphi{s}{n}{1}|\PFphi{s}{n-1}{j}) \belief(\PFphi{s}{n-1}{j}) \mathrm{d}\PFphi{s}{n-1}{j}$,
and
$\xi(\etan{n}^{\scriptscriptstyle(j)})\!=\! \int \! f(\etan{n}^{\scriptscriptstyle(j)}|\etan{n-1}^{\scriptscriptstyle(j)}) \belief(\etan{n-1}^{\scriptscriptstyle(j)}) \mathrm{d}\etan{n-1}^{\scriptscriptstyle(j)}$.

\subsection{Update Messages}\label{sec:update-messages}
With each new channel observation $\observation{n}{j}$, beliefs are updated through update messages.
In \gls{bp} message passing, update messages are described by factor-node-to-variable-node messages sent from likelihood factors to the respective state variables. 
Messages sent from a factor node to a neighboring variable node are computed by multiplying the factor with all incoming messages from the \textit{other} neighboring variable nodes and marginalizing over those variables~\cite{Kschischang01factorGraphs}. 

Update messages sent from likelihood factor nodes 
$f_{\scriptscriptstyle\mathrm{II}}\big(\observation{n}{j}|\state{n},\PFphi{s}{n}{j},\etan{n}^{\scriptscriptstyle(j)}\big)$
to the \gls{mt} state variable node $\state{n}$ are
\begin{align}\label{eq:Miota-exact}
    \iota\big(\state{n};\observation{n}{j}\big) 
    &\!= 
    \!\! \iint \!
    f_{\scriptscriptstyle\mathrm{II}}\big(\observation{n}{j}|\state{n},\PFphi{s}{n}{j},\etan{n}^{\scriptscriptstyle(j)}\big)
    \alpha(\PFphi{s}{n}{j}) 
    \Mxi(\etan{n}^{\scriptscriptstyle(j)})
    \mathrm{d}\PFphi{s}{n}{j}
    \mathrm{d}\etan{n}^{\scriptscriptstyle(j)}
    \,.
\end{align}

Update messages sent from likelihood factor nodes 
to the amplitude state variable node $\PFphi{s}{n}{j}$ are
\begin{align}\label{eq:Mkappa-exact}
    \kappa\big(\PFphi{s}{n}{j};\observation{n}{j}\big) 
    &\!= 
    \!\! \iint \!
    f_{\scriptscriptstyle\mathrm{II}}\big(\observation{n}{j}|\state{n},\PFphi{s}{n}{j},\etan{n}^{\scriptscriptstyle(j)}\big)
    \beta(\state{n})
    \Mxi(\etan{n}^{\scriptscriptstyle(j)})
    \mathrm{d}\state{n}
    \mathrm{d}\etan{n}^{\scriptscriptstyle(j)}
    \,,
\end{align}
and 
update messages sent from likelihood factor nodes 
to the noise variance variable nodes $\etan{n}^{\scriptscriptstyle(j)}$ are
\begin{align}\label{eq:Mnu-exact}
    \nu\big(\etan{n}^{\scriptscriptstyle(j)};\observation{n}{j}\big) 
    &\!= 
    \!\! \iint \!
    f_{\scriptscriptstyle\mathrm{II}}\big(\observation{n}{j}|\state{n},\PFphi{s}{n}{j},\etan{n}^{\scriptscriptstyle(j)}\big)
    \beta(\state{n})
    \alpha(\PFphi{s}{n}{j}) 
    \mathrm{d}\state{n}
    \mathrm{d}\PFphi{s}{n}{j}
    \,.
\end{align}
These update messages were derived for the \gls{fg} in Fig.\,\ref{fig:factor-graph} under the Type-II marginal likelihood model in \eqref{eq:likelihood-Type-II}, where amplitudes $\amplitude{s}{n}{j}$ are marginalized out analytically, introducing their prior parameters, i.e., the amplitude state $\PFphi{s}{n}{j}$, into the \gls{fg}.

\subsection{Approximate Moment-Matched Update Messages}\label{sec:moment-matching}

Beliefs \eqref{eq:belief-state}--\eqref{eq:belief-eta} and the resulting prediction messages \eqref{eq:Mbeta-exact}--\eqref{eq:Mxi-exact} are general non-Gaussian \glspl{pdf}.
Consequently, the exact update messages are likewise non-Gaussian and the marginalization integrals in \eqref{eq:Miota-exact}--\eqref{eq:Mnu-exact} do not generally admit a closed-form evaluation.
In Section~\ref{sec:particle-based-implementation}, we introduce \glspl{pr} of beliefs and prediction messages.
If each prediction message is represented by a separate set of particles, the Monte Carlo integrations approximating \eqref{eq:Miota-exact}--\eqref{eq:Mnu-exact} would involve sums over different particle sets, with computational cost scaling cubically in the number of particles $P$ and becoming prohibitive for large $P$.
The stacking trick~\cite{Meyer16DistrTracking}---stacking particles of different prior messages---could retain linear scaling in $P$, but the resulting implementation would operate on a large stacked state and could suffer from the curse of dimensionality~\cite{Bengtsson08curseOfDimensionality}, a common problem of sample-based estimators operating with high-dimensional states~\cite{WieLeiMeyWit:TSIPN2023}. 
For that reason, we follow the \textit{moment-matching} approach in~\cite{Davies22MomentMatching,Mingchao23TBT,LiaMey:Asilomar2024,LiaLeiMey:TSP2025,Deutschmann26TSP}.
That is, the exact update messages in \eqref{eq:Miota-exact}--\eqref{eq:Mnu-exact} are interpreted as \glspl{pdf} in $\observation{n}{j}$ conditioned on the respective state being updated and approximated by complex Gaussian\footnote{The update messages are non-Gaussian in the respective states being updated.} \glspl{pdf} in $\observation{n}{j}$ with matching conditional means and covariance matrices.
Specifically, we formulate the approximate messages 
$\widetilde{\iota}(\state{n};\observation{n}{j}) \!\coloneqq\! \mathcal{CN}\big(\observation{n}{j};\muiota(\state{n}),\Ciota(\state{n})\big)$,~
$\widetilde{\nu}(\etan{n}^{\scriptscriptstyle(j)};\observation{n}{j})\!\coloneqq\! \mathcal{CN}\big(\observation{n}{j};\munu,\Cnu(\etan{n}^{\scriptscriptstyle(j)})\big)$, and
$\widetilde{\kappa}(\PFphi{s}{n}{j};\observation{n}{j})\!\coloneqq\! \mathcal{CN}\big(\observation{n}{j};\mukappa(\PFphi{s}{n}{j}),\Ckappa(\PFphi{s}{n}{j})\big)$.

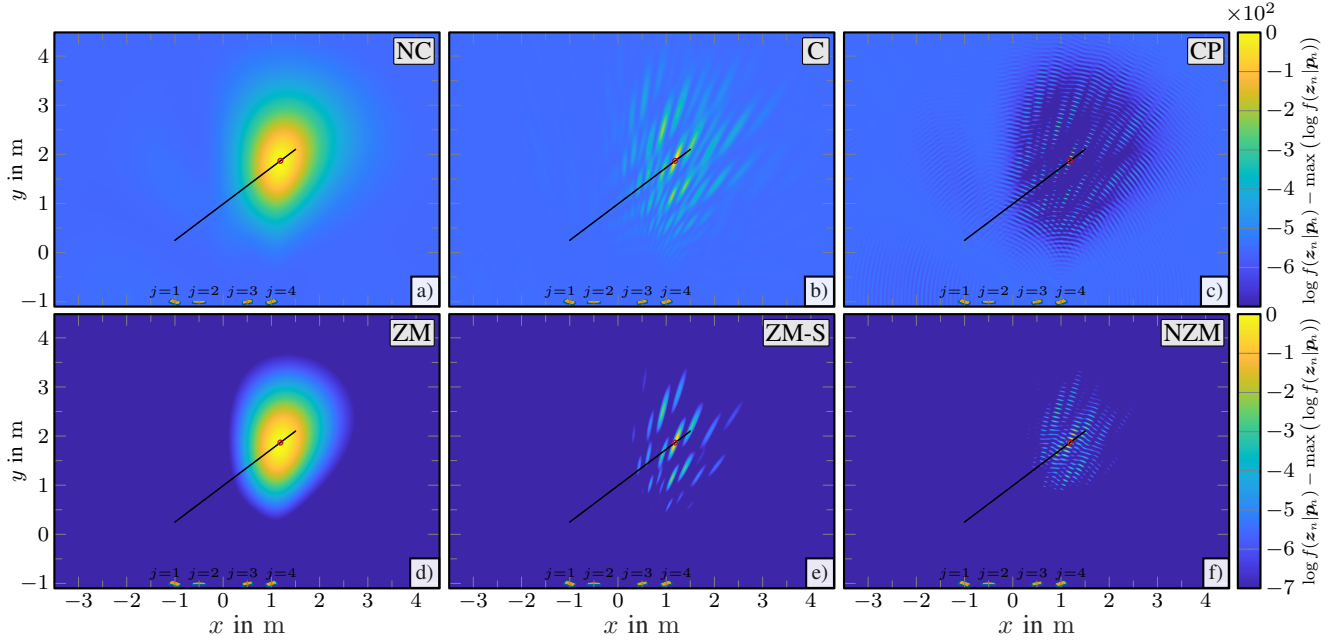
\begin{figure*}[t]
    \centering%
    \newcommand{\datapath}{figures/LHFs}
    \setlength{\figurewidth}{0.32\linewidth}
    \input{figures/LHFs/LHFs.tex}%
    \vspace{-0.1cm}\caption{%
    Type-I likelihood functions of filters 
    a) noncoherent (NC),
    b) coherent (C),
    c) carrier-phase-based (CP), 
    and Type-II likelihood functions of filters 
    d) zero-mean (ZM),
    e) stacked zero-mean (ZM-S),
    f) nonzero-mean (NZM)
    evaluated at time step $n\!=\!175$ on the expected \gls{mt} trajectory, with system parameters detailed in Section~\ref{sec:Results}.
    The likelihood functions are evaluated on a plane parallel to the $xy$-plane at the height of the \gls{mt}, whose position is marked by a red circle.
    The subplots show the relative concentrated (Type-I) and marginal (Type-II) log likelihoods, with their maxima subtracted. 
    The underlying likelihoods are $\prod_{j\in\setAnchors}\widetilde{\iota}(\state{n};\observation{n}{j})$ for \gls{pa}-local (i.e., distributed) message computation and $\widetilde{\iota}(\state{n};\observationn{n})$ for stacked (i.e., centralized) processing.
    }%
    \label{fig:LHFs}%
    \vspace{-0.2cm}
\end{figure*}

\begingroup             
\allowdisplaybreaks[4] 
\begin{proposition}
The first raw moments of the messages are
\begin{align}
    \muiota(\state{n}) 
    &\coloneqq 
    \E_{\iota} \!
    \left(
        \RVobservation{n}{j}|\RVstate{n}
    \right)
    =
    \musnj{1}(\state{n})
    \label{eq:muiota}
    \\
    \munu
    &\coloneqq 
    \E_{\nu} \!
    \left(
        \RVobservation{n}{j}|\RVetan{n}^{\scriptscriptstyle(j)}
    \right)
    =
    \musnj{3}
    \\
    \mukappa(\PFphi{s}{n}{j})
    &\coloneqq 
    \E_{\kappa} \!
    \left(
        \RVobservation{n}{j}|\RVPFphi{s}{n}
    \right)
    = 
    \musnj{2}(\PFphi{s}{n}{j}) 
\end{align}
and the second central moments of the messages are
\begin{align}
    &\Ciota(\state{n}) 
    \!\coloneqq\!
    \E_{\iota} 
    \!
    \Big(
        \RVobservation{n}{j}{\RVobservation{n}{j}}^\herm|\RVstate{n}
    \Big)
    \!-\!
    \E_{\iota} \!
    \left(
        \RVobservation{n}{j}|\RVstate{n}
    \right) 
    \E_{\iota} \!
    \left(
        \RVobservation{n}{j}|\RVstate{n}
    \right)^\herm
    \label{eq:Ciota}
    \\
    &=
    \eye{\Nz} \etaXi{n} \!+\!
     \Csnj{1}(\state{n}) 
    \!+\!
    \Kiota\!(\state{n}) \!-\! \muiota{\muiota}^{\!\herm}\!(\state{n}) 
    \nn \\ 
    &\Cnu(\etan{n}^{\scriptscriptstyle(j)}) 
    \!\coloneqq \!
    \E_{\nu} 
    \!
    \Big(
        \!\RVobservation{n}{j}{\RVobservation{n}{j}}^\herm|\RVetan{n}^{\scriptscriptstyle(j)}\!
    \Big)
    \!-\!
    \E_{\nu} \!
    \left(
        \RVobservation{n}{j}|\RVetan{n}^{\scriptscriptstyle(j)}\!
    \right) 
    \!\E_{\nu} \!
    \left(
        \RVobservation{n}{j}|\RVetan{n}^{\scriptscriptstyle(j)}
    \right) ^\herm
    \nn \\
    &=
    \eye{\Nz} \etan{n}^{\scriptscriptstyle(j)} +
     \Csnj{3}
    \!+\!
    \Knu  \!-\! \munu{\munu}^{\!\herm}
    \\ 
    &\Ckappa(\PFphi{s}{n}{j})
    \!\coloneqq\!
    \E_{\kappa} 
    \!
    \Big(
        \! \RVobservation{n}{j}{\RVobservation{n}{j}}^\herm|\RVPFphi{s}{n} \!
    \Big)
    \!-\!
    \E_{\kappa} \!
    \left(
        \RVobservation{n}{j}|\RVPFphi{s}{n}
    \right) 
    \!
    \E_{\kappa} \!
    \left(
        \RVobservation{n}{j}|\RVPFphi{s}{n}
    \right)^{\!\herm}
    \nn 
    \\
    &=
    \eye{\Nz} \etaXi{n}
    +\!
     \Csnjp{2}(\PFphi{s}{n}{j}) 
    +
    \Kkappa(\PFphi{s}{n}{j}) \!-\! \mukappa{\mukappa}^{\!\herm}(\PFphi{s}{n}{j}) 
    \label{eq:Ckappa}
\end{align}
with the noise variance prior mean $\etaXi{n}\!\coloneqq\!\int\etan{n}^{\scriptscriptstyle(j)}\Mxi(\etan{n}^{\scriptscriptstyle(j)})\mathrm{d}\etan{n}^{\scriptscriptstyle(j)}$ and abbreviated moment terms
\begin{align}
    \musnj{1}(\state{n}) 
    &\coloneqq \int \PFmu{\!\overline{s},\underline{s}\!}{n}{j} 
    \steerVecx{s}{n}{j}
    \alpha(\PFphi{s}{n}{j}) \mathrm{d}\PFphi{s}{n}{j} 
    \label{eq:musnj1}
    \\
    \musnj{2}(\PFphi{s}{n}{j}) &\coloneqq  \!\!
    \int \! 
    \PFmu{\!\overline{s},\underline{s}\!}{n}{j}
    \steerVecx{s}{n}{j} 
    \beta(\state{n}) \mathrm{d}\state{n}
    \label{eq:musnj2}
    \\
    \musnj{3} &\coloneqq \iint
    \PFmu{\!\overline{s},\underline{s}\!}{n}{j} 
    \steerVecx{s}{n}{j}
    \beta(\state{n})
    \mathrm{d}\state{n}
    \alpha(\PFphi{s}{n}{j})  
    \mathrm{d}\PFphi{s}{n}{j}
    \label{eq:musnj3}
    \\
    \Csnj{1}(\state{n}) &\coloneqq
    \int \PFgamma{\!\overline{s},\underline{s}\!}{n}{j} \steerVecx{s}{n}{j} {\steerVecx{s}{n}{j}}^\herm 
    \alpha(\PFphi{s}{n}{j}) \mathrm{d}\PFphi{s}{n}{j} 
    \label{eq:Csnj1}
    \\
    \Csnj{2}(\PFphi{s}{n}{j}) &\coloneqq
    \int \PFgamma{\!\overline{s},\underline{s}\!}{n}{j}
    \steerVecx{s}{n}{j} {\steerVecx{s}{n}{j}}^\herm 
    \beta(\state{n}) \mathrm{d}\state{n}
    \label{eq:Csnj2}
    \\
    \Csnj{3} 
    &\coloneqq
    \iint
    \! \PFgamma{\!\overline{s},\underline{s}\!}{n}{j} 
    \steerVecx{s}{n}{j} {\steerVecx{s}{n}{j}}^\herm
    \!\beta(\state{n})\mathrm{d}\state{n}
    \alpha(\PFphi{s}{n}{j})  
    \mathrm{d}\PFphi{s}{n}{j}
    \label{eq:Csnj3}
    \\
    \Kiota(\state{n}) 
    &\coloneqq
    \int 
    \bm{\mu}_{\scriptscriptstyle n}^{\scriptscriptstyle (j)}
    {\bm{\mu}_{\scriptscriptstyle n}^{\scriptscriptstyle (j)}}^\herm 
    \alpha(\PFphi{s}{n}{j}) \mathrm{d}\PFphi{s}{n}{j} 
    \label{eq:Kiota}
    \\
    \Kkappa(\PFphi{s}{n}{j}) 
    &\coloneqq
    \int 
    \bm{\mu}_{\scriptscriptstyle n}^{\scriptscriptstyle (j)}
    {\bm{\mu}_{\scriptscriptstyle n}^{\scriptscriptstyle (j)}}^\herm 
    \beta(\state{n}) \mathrm{d}\state{n}
    \label{eq:Kkappa}
    \\
    \Knu
    &\coloneqq
    \iint
    \bm{\mu}_{\scriptscriptstyle n}^{\scriptscriptstyle (j)}
    {\bm{\mu}_{\scriptscriptstyle n}^{\scriptscriptstyle (j)}}^\herm 
    \beta(\state{n})\mathrm{d}\state{n}
    \alpha(\PFphi{s}{n}{j})  
    \mathrm{d}\PFphi{s}{n}{j}
    \label{eq:Knu}
\end{align}
\end{proposition}
\endgroup
\begin{proof}
See our Extended Derivations in~\cite{ThisPaperDerivations}.
\end{proof} %
We use the approximations $\Kiota\!(\state{n}) \!\approx\! \muiota{\muiota}^{\!\herm}\!(\state{n})$, $\Knu  \!\approx\! \munu{\munu}^{\!\herm}$, and $\Kkappa(\PFphi{s}{n}{j}) \!\approx\! \mukappa{\mukappa}^{\!\herm}(\PFphi{s}{n}{j})$, allowing us to parameterize the approximate Gaussian update messages solely 
through means and covariance matrices 
\eqref{eq:musnj1}--\eqref{eq:Csnj3}.

\Gls{bp} messages would follow analogously under a Type-I likelihood model. 
However, we choose a different treatment of nuisance parameters.

\section{Type-I Estimators}\label{sec:Type-I}  
The amplitude moduli $|\amplitude{s}{n}{j}|$ and noise variance $\etan{n}^{\scriptscriptstyle(j)}$ contain negligible information about the parameters of interest, i.e., the \gls{mt} state $\state{n}$, hence they are regarded as nuisance parameters.
Two pathways for their treatment 
are particularly convenient:
profiling, i.e., concentration at \gls{ml} estimates, or closed-form marginalization.
For the Type-I estimators, we concentrate the likelihood at the
\gls{ml} estimates of its nuisance parameters, which therefore do not
appear as variable nodes in the corresponding \glspl{fg}.

\subsection{Snapshot-Based NC Type-I Estimator}
\vspace{-3mm}
\begin{figure}[h]%
    \centering%
    \includegraphics[width = \linewidth]{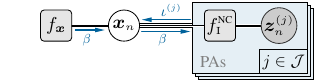}%
    \vspace{-0.2cm}\caption{FG of I) the snapshot-based NC estimator with $f_{\scriptscriptstyle\mathrm{I}}^{\text{\tiny{NC}}}\!\coloneqq\!f_{\scriptscriptstyle\mathrm{I}}^{\text{\tiny{NC}}}(\observation{n}{j};\state{n})$
    and a fixed prior \gls{pdf} $f_{\bm{x}}\!\coloneqq\!f(\state{n})$.
    }%
    \label{fig:main-I-NC-snapshot}%
\end{figure}%
We compute the \textit{profile} likelihood, i.e., we \textit{concentrate} w.r.t. the nuisance parameters using \gls{ml} estimates~\cite{KrimViberg96ASP,Deutschmann24SPAWC,DeutschmannCISA2025}
\begin{align}
    \amplitudeHat{n}{j}|\state{n} &= 
    \frac{ 
        {\steerVecx{s}{n}{j}}^\herm \observation{n}{j} 
    }{
        \lVert {\steerVecx{s}{n}{j}} \rVert^2
    } \,,
\\
    \etaHat{n}{j}|\state{n} &= 
    \frac{ 
        \lVert \observation{n}{j} \rVert^2
    }{
        \Nz
    }
    -
    \frac{ 
        \big|
            {\steerVecx{s}{n}{j}}^\herm \observation{n}{j} 
        \big|^2
    }{
        \Nz \lVert \steerVecx{s}{n}{j} \rVert^2
    }
\end{align}
conditional on the state $\state{n}$.
Reinsertion of $\amplitudeHat{n}{j}|\state{n}$ and $\etaHat{n}{j}|\state{n}$ into~\eqref{eq:likelihood-Type-I} yields the noncoherent Type-I profile likelihood function
\begin{align}\label{eq:likelihood-Type-I-NC}
    f_{\scriptscriptstyle\mathrm{I}}^{\text{\tiny{NC}}}\big(\observation{n}{j};\state{n}\big) 
    &= 
    \CN\big(\observation{n}{j};
    \amplitudeHat{n}{j}
    \steerVec{j}(\state{n}),
    \etaHat{n}{j}\eye{\Nz}\big)\,.
\end{align}
Fig.\,\ref{fig:main-I-NC-snapshot} shows the \gls{fg} of our noncoherent snapshot-based Type-I estimator \estLabel{S}.
Here, the update message $\iota(\state{n};\observation{n}{j})\!=\!f_{\scriptscriptstyle\mathrm{I}}^{\text{\tiny{NC}}}\big(\observation{n}{j};\state{n}\big)$ corresponds to the per-\gls{pa} likelihood function.
Due to conditionally independent observations, the joint (over \glspl{pa}) profile likelihood factorizes into a product over per-\gls{pa} likelihoods $\prod_{j\in\setAnchors}f_{\scriptscriptstyle\mathrm{I}}^{\text{\tiny{NC}}}\big(\observation{n}{j};\state{n}\big)$.
We find that \textit{coherence} is fundamentally tied to the number of \textit{nuisance phases}---contained in complex amplitudes---that need to be estimated:
Concentration w.r.t. one nuisance phase per \gls{pa} $j$ yields the \textit{noncoherent} profile likelihood function depicted in Fig.\,\ref{fig:LHFs}\,a), which is unimodal, yet has limited curvature around the true \gls{mt} position.
The belief is computed according to \eqref{eq:belief-state} with $\beta(\state{n})\!=\!f(\state{n})$ acting as a fixed prior message given by, e.g., a uniform \gls{pdf} 
that is not sequentially updated.

\subsection{NC Type-I Filter}
\vspace{-3mm}
\begin{figure}[h]%
    \centering%
    \includegraphics[width = \linewidth]{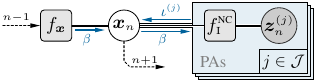}%
    \vspace{-0.2cm}\caption{FG of II) the NC filter with likelihood $f_{\scriptscriptstyle\mathrm{I}}^{\text{\tiny{NC}}}:=f_{\scriptscriptstyle\mathrm{I}}^{\text{\tiny{NC}}}(\observation{n}{j};\state{n})$ and state-transition \gls{pdf} $f_{\bm{x}}\!\coloneqq\!f(\state{n}|\state{n-1})$.}%
    \label{fig:main-I-NC}%
\end{figure}%
In contrast, the noncoherent Type-I filter \estLabel{NC} sequentially computes prediction messages $\beta(\state{n})$ through \eqref{eq:Mbeta-exact}.
In the corresponding \gls{fg} in Fig.\,\ref{fig:main-I-NC}, the prior factor is the state-transition \gls{pdf} $f(\state{n}|\state{n-1})$. 
It uses the same update messages $\iota(\state{n};\observation{n}{j})\!=\!f_{\scriptscriptstyle\mathrm{I}}^{\text{\tiny{NC}}}\big(\observation{n}{j};\state{n}\big)$ as its snapshot-based counterpart but gains additional information by propagating the belief $\belief(\state{n-1})$ over the state-transition \gls{pdf} from one time step to the next, i.e., by state filtering.

\subsection{C Type-I Filter}
\vspace{-3mm}
\begin{figure}[h]%
    \centering%
    \includegraphics[width = \linewidth]{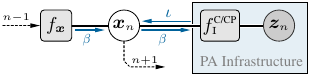}%
    \vspace{-0.2cm}\caption{FG of 
    III) the C filter with $f_{\scriptscriptstyle\mathrm{I}}^{\text{\tiny{C}}}:=f_{\scriptscriptstyle\mathrm{I}}^{\text{\tiny{C}}}(\observationn{n};\state{n})$ or 
    IV) the CP filter with $f_{\scriptscriptstyle\mathrm{I}}^{\text{\tiny{CP}}}:=f_{\scriptscriptstyle\mathrm{I}}^{\text{\tiny{CP}}}(\observationn{n};\state{n})$, 
    and state-transition \gls{pdf} $f_{\bm{x}}\!\coloneqq\!f(\state{n}|\state{n-1})$.}%
    \label{fig:main-I-C}%
\end{figure}%
We present two options for \textit{coherent} \gls{dmimo} data fusion:
i) in the data domain and ii) in the likelihood domain.
In the case of the former, we \textit{stack} observations, i.e., $\observationn{n}$, and
steering vectors 
$\steerVecxBar{n}\!=\!
\steerVecBar(\state{n})\!\coloneqq\!
\big[ 
    {\steerVecx{s}{n}{1}}^\trp
    \ist \hdots \ist
    {\steerVecx{s}{n}{J}}^\trp
\big]^\trp\in\complexsetone{J \Nz}$.
The coherent Type-I model constrains \gls{pa}-local amplitudes to a common 
amplitude,
$\amplitude{s}{n}{j}\!=\!\amplitudeBar{n}$ for all \glspl{pa} $j\!\in\!\setAnchors$.
For simplicity, we also constrain \gls{pa}-local noise variances to one common noise variance
$\etan{n}^{\scriptscriptstyle(j)}\!=\!\etan{n}$, yielding the (stacked) Type-I likelihood
\begin{align}\label{eq:likelihood-Type-I-stacked}
f\big(\observationn{n};\state{n},\amplitudeBar{n},\etan{n}\big) 
= 
\CN\big(\observationn{n};
\amplitudeBar{n}
\steerVecxBar{n},
\etan{n} \eye{J\Nz}\big) \,.
\end{align}
Again, we compute \gls{ml} estimates
\begin{align}\label{eq:amplitudeBarHat}
    \amplitudeBarHat{n}\big|\state{n} &= 
    \frac{ 
        {\steerVecxBar{n}}^\herm \observationn{n} 
    }{
        \lVert {\steerVecxBar{n}} \rVert^2
    } \,,
\\
    \etaHatBar{n}\big|\state{n} &= 
    \frac{ 
        \lVert \observationn{n} \rVert^2
    }{
        J \Nz
    }
    -
    \frac{ 
        \big|
            {\steerVecxBar{n}}^\herm \observationn{n}
        \big|^2
    }{
        J \Nz \lVert \steerVecxBar{n} \rVert^2
    }
\end{align}
conditional on the state $\state{n}$, reinsert them in \eqref{eq:likelihood-Type-I-stacked}, and obtain the \textit{coherent} Type-I profile likelihood function
\begin{align}\label{eq:likelihood-Type-I-C}
    f_{\scriptscriptstyle\mathrm{I}}^{\text{\tiny{C}}}\big(\observationn{n};\state{n}\big) 
    &= 
    \CN\big(\observationn{n};
    \amplitudeBarHat{n}
    \steerVecBar(\state{n}),
    \etaHatBar{n}\eye{J\Nz}\big)\,.
\end{align}
Fig.\,\ref{fig:main-I-C} shows the \gls{fg} of our coherent Type-I filter \estLabel{C}.
Here, the single update message $\iota(\state{n};\observationn{n})\!=\!f_{\scriptscriptstyle\mathrm{I}}^{\text{\tiny{C}}}\big(\observationn{n};\state{n}\big)$ corresponds to the Type-I likelihood function.
Prediction messages $\beta(\state{n})$ are again computed through \eqref{eq:Mbeta-exact} and the belief through \eqref{eq:belief-state}.

Concentration w.r.t. one common complex amplitude---and hence one common nuisance phase---across the entire \gls{pa} infrastructure yields the coherent profile likelihood function depicted in Fig.\,\ref{fig:LHFs}\,b).
It is multimodal due to grating lobes induced by spatial aliasing but exhibits greater curvature around the true \gls{mt} position than the noncoherent likelihood due to the inherent aperture gain.

\subsection{CP Type-I Filter}
In contrast, the carrier-phase-based Type-I filter \estLabel{CP} assumes that the phase relations between all \glspl{pa} and the \gls{mt} are perfectly known (i.e., calibrated) 
so that no phase needs to be estimated and the modulus 
$\amplitudeBarCP{n}\!\in\!\realsetone{\,}_{\scriptscriptstyle\geq0}$ is the only remaining amplitude nuisance parameter.

\begingroup             
\allowdisplaybreaks[4] 
\begin{proposition}
The \textit{carrier-phase-based} Type-I (stacked) profile likelihood function is
\begin{align}\label{eq:likelihood-Type-I-CP}
    f_{\scriptscriptstyle\mathrm{I}}^{\text{\tiny{CP}}}\big(\observationn{n};\state{n}\big) 
    &= 
    \CN\big(\observationn{n};
    \amplitudeBarHatCP{n}
    \steerVecBar(\state{n}),
    \etaHatBarCP{n}\eye{J\Nz}\big)
\end{align}
obtained by concentration using the \gls{ml} estimators
\begin{align}
    \amplitudeBarHatCP{n}
    &\!=\!
    \frac{
    \max \Big(
    0,
    \Re\big( \observationn{n}^\herm\steerVecBar_n \big)
    \Big)
    }{\big\| 
    \steerVecBar_n 
    \big\|^2}\,
    \\
    \etaHatBarCP{n} 
    &\!=\! 
    \frac{
    \big\| 
    \observationn{n}
    \big\|^2
    }{J \Nz}
    -
    \frac{
        (\amplitudeBarHatCP{n})^2
        \big\| 
        \steerVecBar_n 
        \big\|^2
    }
    {J \Nz}\,.
\end{align}
\end{proposition}
\endgroup
\begin{proof}
See Supplementary Material, Sec.\,S-I.
\end{proof} %

Apart from the likelihood and update message $\iota(\state{n};\observationn{n})\!=\!f_{\scriptscriptstyle\mathrm{I}}^{\text{\tiny{CP}}}\big(\observationn{n};\state{n}\big)$, the \gls{fg} in Fig.\,\ref{fig:main-I-C} remains the same as for the coherent Type-I filter.
The carrier-phase-based profile likelihood function depicted in Fig.\,\ref{fig:LHFs}\,c) is highly multimodal due to the distance dependence of the carrier phase but has a very high curvature around the true \gls{mt} position.

Notably, for the presented Type-I filters, our \gls{bp} method reduces to sequential Bayesian filtering, and its particle-based implementation corresponds to a particle filter.

\section{Type-II Estimators}\label{sec:Type-II}  
In the case of the Type-II estimators, the amplitudes have been marginalized out analytically, i.e., $f_{\scriptscriptstyle\mathrm{II}}\big(\observation{n}{j}|\state{n},\PFphi{s}{n}{j},\etan{n}^{\scriptscriptstyle(j)}\big)\!=\!\int f_{\scriptscriptstyle\mathrm{I}}\big(\observation{n}{j}|\state{n},\amplitude{s}{n}{j},\etan{n}^{\scriptscriptstyle(j)}\big) f\big(\amplitude{s}{n}{j}|\PFmu{s}{n}{j},\PFgamma{s}{n}{j}\big) \mathrm{d}\amplitude{s}{n}{j}$.
However, the amplitude prior parameters $\PFphi{s}{n}{j}$ still need to be tracked together with the noise variance $\etan{n}^{\scriptscriptstyle(j)}$.
This is done by the \gls{bp} approach described in Section~\ref{sec:problem-formulation}, using the prediction messages in \eqref{eq:Mbeta-exact}--\eqref{eq:Mxi-exact} and the moment-matched approximations from Section~\ref{sec:moment-matching} of the update messages in \eqref{eq:Miota-exact}--\eqref{eq:Mnu-exact}.
We present three Type-II filters for the \gls{dmimo} state-filtering problem.

\subsection{ZM Type-II Filter}\label{sec:ZM-Type-II-Filter}
The first version corresponds to the model described in Section~\ref{sec:problem-formulation} assuming a zero-mean amplitude prior, i.e., $\RVamplitude{s}{n}{j}|\RVPFgamma{s}{n}{j} \!\overset{\scriptscriptstyle\mathrm{i.i.d.}}{\sim}\! \CN(0,\RVPFgamma{s}{n}{j})$, and $\PFphi{s}{n}{j}\!=\!\PFgamma{s}{n}{j}$.
This leads to our zero-mean Type-II filter \estLabel{ZM}.
Its Type-II marginal likelihood function is
\begin{align}\label{eq:likelihood-Type-II-ZM}
    f_{\scriptscriptstyle\mathrm{II}}^{\text{\tiny{ZM}}}\big(\observation{n}{j}|\state{n},\PFphi{s}{n}{j},\etan{n}^{\scriptscriptstyle(j)}\big) 
    &= \CN\big(\observation{n}{j};
    \bm{0},
    \bm{C}_{\scriptscriptstyle n}^{\scriptscriptstyle (j)}\big)\,.
\end{align}
While it shares a common amplitude prior \textit{variance} across \glspl{pa}, it does not share a common complex amplitude realization and therefore does not preserve the phase relation (hence coherence) across \glspl{pa}, cf.\,\cite[Remark~III.1]{Reniers26MML}. 
Its likelihood function in Fig.\,\ref{fig:LHFs}\,d) resembles that of \estLabel{NC} in Fig.\,\ref{fig:LHFs}\,a).

\subsection{Stacked ZM Type-II Filter}
The second version differs from \estLabel{ZM} by assuming that all \glspl{pa} share the same \textit{realization} $\amplitude{s}{n}{j}=\amplitudeBar{n} ~\forall j\!\in\! \setAnchors$ of a zero-mean complex amplitude $\RVamplitudeBar{n}|\RVPFgamma{s}{n}{j} \!\sim\! \CN({0},\RVPFgamma{s}{n}{j})$ and one common noise variance $\etan{n}$, rather than having independent \gls{pa}-local amplitude realizations and noise variances.
It performs coherent data fusion in the data domain by stacking the observations and steering vectors, as in \estLabel{C}. 
The shared amplitude realization enforces a rigid relative phase relation (and coherence) across \glspl{pa}.
This leads to our \textit{stacked} zero-mean Type-II filter \estLabel{ZM-S}.
Its Type-II marginal likelihood function is
\begin{align}\label{eq:likelihood-Type-II-ZM-S}
    f_{\scriptscriptstyle\mathrm{II}}^{\text{\tiny{ZM-S}}}\big(\observationn{n}|\state{n},\PFphi{s}{n}{j},\etan{n}\big) 
    &= \CN\big(\observationn{n};
    \bm{0},
    \bm{C}_{\scriptscriptstyle n}\big)\,
\end{align}
where $\bm{C}_{\scriptscriptstyle n}\!=\!
\PFgamma{s}{n}{j}
{\steerVecxBar{n}}
{\steerVecxBar{n}}^\herm
+
\etan{n} \eye{J\Nz}$.
This likelihood is depicted in Fig.\,\ref{fig:LHFs}\,e) and resembles that of \estLabel{C} in Fig.\,\ref{fig:LHFs}\,b).

\subsection{NZM Type-II Filter}
The third version corresponds exactly to the model described in Section~\ref{sec:problem-formulation} with a nonzero-mean amplitude prior  
$\RVamplitude{s}{n}{j}|\RVPFmu{s}{n}{j},\RVPFgamma{s}{n}{j} \!\overset{\scriptscriptstyle\mathrm{i.i.d.}}{\sim}\! \CN(\RVPFmu{s}{n}{j},\RVPFgamma{s}{n}{j})$, yielding our nonzero-mean Type-II filter \estLabel{NZM}.
Its Type-II marginal likelihood function is
\begin{align}\label{eq:likelihood-Type-II-NZM}
    f_{\scriptscriptstyle\mathrm{II}}^{\text{\tiny{NZM}}}\big(\observation{n}{j}|\state{n},\PFphi{s}{n}{j},\etan{n}^{\scriptscriptstyle(j)}\big) 
    &= \CN\big(\observation{n}{j};
    \bm{\mu}_{\scriptscriptstyle n}^{\scriptscriptstyle (j)},
    \bm{C}_{\scriptscriptstyle n}^{\scriptscriptstyle (j)}\big)\,.
\end{align}
It performs soft-coherent data fusion in the likelihood domain.
Despite allowing different \gls{pa}-local amplitude realizations, the common nonzero prior mean $\PFmu{s}{n}{j}$ promotes, without rigidly enforcing, a common phase relation among \glspl{pa}.
Its likelihood function in Fig.\,\ref{fig:LHFs}\,f) resembles that of \estLabel{CP} in Fig.\,\ref{fig:LHFs}\,c) if most signal power is captured by the prior means $\PFmu{s}{n}{j}$, and it resembles that of \estLabel{NC} in Fig.\,\ref{fig:LHFs}\,a) if most signal power is captured by the prior variance $\PFgamma{s}{n}{j}$~(cf.\,\cite[Fig.\,5]{Deutschmann26TSP}).
By filtering the phase contained in the prior mean across time, it leverages information analogous to the pulse-to-pulse phase changes underlying slow-time Doppler processing~\cite{Deutschmann24SPAWC}. 

\section{Particle-Based Implementation}\label{sec:particle-based-implementation}  

Since the particle-based implementation of our \gls{bp} method using the Type-I likelihood functions from Section~\ref{sec:Type-I} corresponds to a particle filter~\cite{Arulampalam02PFtutorial} on the \gls{mt} state, we present only the particle-based implementation of \estLabel{ZM} and \estLabel{NZM}.
The particle-based implementation of \estLabel{ZM-S} follows analogously by stacking the observations and steering vectors.
Leveraging centralized fusion (i.e., message multiplication), \estLabel{S}, \estLabel{NC}, \estLabel{ZM}, and \estLabel{NZM} admit distributed \gls{pa}-local processing (i.e., message computation), whereas \estLabel{C}, \estLabel{CP}, and \estLabel{ZM-S} rely on centralized processing.

Because the beliefs are general non-Gaussian \glspl{pdf}, the integrals required by our \gls{bp} method do not admit analytical closed-form solutions. 
We introduce a numerically feasible particle-based implementation of the proposed \gls{bp} method, in which \glspl{pdf} are approximated by \glspl{pr}~\cite{LeitMeyHlaWitTufWin:TWC2019,LiaLeiMey:TSP2025}, which are sets of $P$ weighted particles.
That is, we approximate the beliefs in~\eqref{eq:belief-state}--\eqref{eq:belief-eta} using their \glspl{pr}
$\belief(\state{n}) \approxcolon 
\sum_{\scriptscriptstyle p=1}^{\scriptscriptstyle P} \weight{\bm{x},n}{p} \delta(\state{n}\!- \particle{\state{n}}{p})$,
$\belief(\PFphi{s}{n}{j}) \approxcolon 
\sum_{\scriptscriptstyle p=1}^{\scriptscriptstyle P} \weight{\bm{\phi},n}{p} \delta(\PFphi{s}{n}{j}\!- \particle{\PFphi{s}{n}{j}}{p})$
and
$\belief(\etan{n}^{\scriptscriptstyle(j)}) \approxcolon 
\sum_{\scriptscriptstyle p=1}^{\scriptscriptstyle P}
\weight{\eta,n}{j,p} \delta(\etan{n}^{\scriptscriptstyle(j)}\!- \particle{\etan{n}}{j,p})$.
For these \glspl{pr} to represent \glspl{pdf}, we require their weights to be nonnegative and sum to one, i.e., 
$\sum_{\scriptscriptstyle p=1}^{\scriptscriptstyle P}\weight{\bm{x},n}{p}\!=\!1$,
$\sum_{\scriptscriptstyle p=1}^{\scriptscriptstyle P}\weight{\bm{\phi},n}{p}\!=\!1$,
and 
$\sum_{\scriptscriptstyle p=1}^{\scriptscriptstyle P}\weight{\eta,n}{j,p}\!=\!1$.

\subsection{Prediction Messages}
By approximating marginal posterior \glspl{pdf} at time $n\!-\!1$ with the \glspl{pr} of their beliefs, and marginalizing
\begin{align}
    f(\state{n}|\observationn{1:n-1}) 
    \!&= \!\!
    \int \!\! f(\state{n}|\state{n-1}) f(\state{n-1}|\observationn{1:n-1}) \mathrm{d}\state{n-1}
    \nn \\[-2pt]
    &\approx\!\!
    \int \!\! f(\state{n}|\state{n-1}) 
    \!
    \sum_{\scriptscriptstyle p=1}^{\scriptscriptstyle P} \weight{\bm{x},n-1}{p} \delta(\state{n-1}\!- \particle{\state{n-1}}{p})
    \mathrm{d}\state{n-1}
    \nn \\[-2pt]
    &=
    \sum\nolimits_{\scriptscriptstyle p=1}^{\scriptscriptstyle P} 
    \weight{\bm{x},n-1}{p} f(\state{n}|\particle{\state{n-1}}{p}) \,,\label{eq:PR-prediction-PDF}
\end{align}
we obtain a prediction \gls{pdf} parameterized by particles. 
As is done in \gls{sir}~\cite{Arulampalam02PFtutorial},
for each particle $\particle{\state{n-1}}{p}$ we now draw one particle 
$\particle{\state{n}}{p}$ 
from the state-transition \gls{pdf} $f(\state{n}|\particle{\state{n-1}}{p})$ used as the proposal \gls{pdf} 
to obtain a \gls{pr} for the prediction \gls{pdf}
$f(\state{n}|\observationn{1:n-1}) \!\approx\! \beta(\state{n}) \!\approx \!\sum\nolimits_{\scriptscriptstyle p=1}^{\scriptscriptstyle P}
\weight{\beta,n}{p} \delta\big( \state{n}-\particle{\state{n}}{p} \big)$,
where $\weight{\beta,n}{p}\!=\!\weight{\bm{x},n-1}{p}$ if we can directly sample from $f(\state{n}|\particle{\state{n-1}}{p})$.
The prediction message for the amplitude state is
$\alpha(\PFphi{s}{n}{j})\approx \sum\nolimits_{\scriptscriptstyle p=1}^{\scriptscriptstyle P}
\weight{\alpha,n}{p} \delta(\PFphi{s}{n}{j}\!-\! \particle{\PFphi{s}{n}{j,p}}{p}) $
with $\weight{\alpha,n}{p}\!=\!\weight{\bm{\phi},n-1}{p}$ 
and $\particle{\PFphi{s}{n}{1,p}}{p}$ sampled from $f(\PFphi{s}{n}{1}|\particle{\PFphi{s}{n-1}{J,p}}{p})$.
We likewise compute the noise variance prediction messages 
$\xi(\etan{n}^{\scriptscriptstyle(j)})\approx \sum\nolimits_{\scriptscriptstyle p=1}^{\scriptscriptstyle P}
\weight{\xi,n}{j,p} \delta\big( \etan{n}^{\scriptscriptstyle(j)}-\particle{\etan{n}}{j,p} \big)$
with $\weight{\xi,n}{j,p}\!=\!\weight{\eta,n-1}{j,p}$ 
and $\particle{\etan{n}}{j,p}$ sampled from $f(\etan{n}^{\scriptscriptstyle(j)}|\particle{\etan{n-1}}{j,p})$.

\subsection{Update Messages}\label{sec:PR-update-messages}
Now that \glspl{pr} of prediction messages have been introduced, we can compute the approximate update messages.
In particular, the moments in \eqref{eq:muiota}--\eqref{eq:Ckappa} are approximated using the terms in \eqref{eq:musnj1}--\eqref{eq:Csnj3},
where we insert \glspl{pr} of prediction messages to evaluate the marginalization integrals therein.
Using the current \gls{pa} observations $\observation{n}{j}$, we evaluate 
$\widetilde{\iota}(\particle{\state{n}}{p};\observation{n}{j}) \!=\! \mathcal{CN}\big(\observation{n}{j};\muiota(\particle{\state{n}}{p}),\Ciota(\particle{\state{n}}{p})\big)$,
$\widetilde{\kappa}(\particle{\PFphi{s}{n}{j,p}}{p};\observation{n}{j})\!=\! \mathcal{CN}\big(\observation{n}{j};\mukappa(\particle{\PFphi{s}{n}{j,p}}{p}),\Ckappa(\particle{\PFphi{s}{n}{j,p}}{p})\big)$, and
$\widetilde{\nu}(\particle{\etan{n}}{j,p};\observation{n}{j})\!=\! \mathcal{CN}\big(\observation{n}{j};\munu,\Cnu(\particle{\etan{n}}{j,p})\big)$, 
with \glspl{pr} of mean and covariance terms
from Supplementary Material, Sec.\,S-III. 

\subsection{Beliefs}
With \glspl{pr} of prediction 
and update messages, \glspl{pr} of the beliefs~\eqref{eq:belief-state}--\eqref{eq:belief-eta} are 
$\belief(\state{n}) \!\approx \!
\sum_{\scriptscriptstyle p=1}^{\scriptscriptstyle P} \weight{\bm{x},n}{p} \delta(\state{n}\!- \particle{\state{n}}{p})$
for the \gls{mt} state,
$\belief(\PFphi{s}{n}{j}) \!\approx\!
\sum_{\scriptscriptstyle p=1}^{\scriptscriptstyle P}
\weight{\bm{\phi},n}{p} \delta(\PFphi{s}{n}{j}\!-\! \particle{\PFphi{s}{n}{j,p}}{p})$ for the amplitude state,
and 
$\belief(\etan{n}^{\scriptscriptstyle(j)}) \!\approx\!
\sum_{\scriptscriptstyle p=1}^{\scriptscriptstyle P}
\weight{\eta,n}{j,p} \delta(\etan{n}^{\scriptscriptstyle(j)}\!- \particle{\etan{n}}{j,p})$
for the noise variance with weights 
$\weight{\bm{x},n}{p}\!=\!\frac{\weightt{\bm{x},n}{p}}{\normConst{\bm{x},n}{}}$, 
$\weight{\bm{\phi},n}{p}\!=\! 
\frac{\weightt{\bm{\phi},n}{p}}{\normConst{\bm{\phi},n}{}}$, and
$\weight{\eta,n}{j,p}\!=\!\frac{\weightt{\eta,n}{j,p}}{\normConst{\eta,n}{(j)}}$
where
\begin{align}
    \weightt{\bm{x},n}{p} 
    &\coloneqq 
    \weight{\beta,n}{p}
    \prod_{\scriptscriptstyle j\in \setAnchors}
    \mathcal{CN}\big(\observation{n}{j};\muiota(\particle{\state{n}}{p}),\Ciota(\particle{\state{n}}{p})\big)
    \\
    \weightt{\bm{\phi},n}{p}
    &\coloneqq 
    \weight{\alpha,n}{p}
    \prod_{\scriptscriptstyle j\in \setAnchors}
    \mathcal{CN}\big(\observation{n}{j};\mukappa(\particle{\PFphi{s}{n}{j,p}}{p}),\Ckappa(\particle{\PFphi{s}{n}{j,p}}{p})\big)
    \label{eq:weightt-phi}
    \\
    \weightt{\eta,n}{j,p}
    &\coloneqq 
    \weight{\xi,n}{j,p}
    \mathcal{CN}\big(\observation{n}{j};\munu,\Cnu(\particle{\etan{n}}{j,p})\big)
\end{align}
derived 
in Supplementary Material, Sec.\,S-II, and normalized with constants
$\normConst{\bm{x},n}{}\!=\!{\sum_{\scriptscriptstyle p=1}^{\scriptscriptstyle P} \weightt{\bm{x},n}{p}}$, 
$\normConst{\eta,n}{(j)}\!=\!{\sum_{\scriptscriptstyle p=1}^{\scriptscriptstyle P}
\weightt{\eta,n}{j,p}}$, and
$\normConst{\bm{\phi},n}{}\!=\!\sum_{\scriptscriptstyle p=1}^{\scriptscriptstyle P}
\weightt{\bm{\phi},n}{p}$.
Inner products in the update messages derived in our Extended Derivations~\cite{ThisPaperDerivations} and summarized in Supplementary Material, Sec.\,S-III, dominate the total computational complexity (i.e., \textit{work}) of our algorithm, which scales as $\mathcal{O}(P\,J\,\Nz)$ per time step.
Parallelization can greatly reduce runtime in a \gls{gpu}-accelerated implementation, cf.\,Section~\ref{sec:runtime}.
Following the work--depth model~\cite{Blelloch96parallelComputing}, we distinguish the total computational \textit{work} from the \textit{depth}, defined as the longest chain of sequentially dependent operations.
Under ideal execution, parallel reductions yield a 
\textit{depth} per time step of 
$\mathcal{O}(\log P + \log J + \log \Nz)$, dominated by
(i) sums over $P$ particles during state estimation,
(ii) products of update messages from $J$ \glspl{pa} for data fusion, implemented as sums in the log domain, and
(iii) sums over $\Nz$ observation bins in inner products of \gls{bp} update messages.

After computing the belief of each state, we perform systematic resampling~\cite[Alg.\,2]{Arulampalam02PFtutorial}, which reduces particle degeneracy and implies equal weights of resampled particles.
To counteract particle impoverishment, the resampled \gls{mt} state \gls{pr} is convolved with a zero-mean Gaussian regularization kernel whose covariance matrix is moment-matched to the second central moment of the belief
and scaled by the square of the optimal kernel bandwidth $\mathrm{h}_{\text{\tiny opt}}$~\cite[p.\,253]{Musso2001PFCH12}.

\section{Posterior Cram\'er--Rao Lower Bound}\label{sec:PCRLB}  
We derive one classic \gls{crlb} and three \glspl{pcrlb} for the \gls{dmimo} positioning problem.
First, we decompose amplitudes%
\footnote{For deriving the \gls{pcrlb}, 
the amplitudes $\widetilde{\unslant[-.25]{\varrho}}_{\!\scriptscriptstyle n}^{\scriptscriptstyle(j)}$ absorb the path loss as well as the carrier-phase term $\exp \!(\!-\mathrm{j}\frac{2\pi}{\lightspeed}\fc \lVert \acute{\V{r}}\rVert )$ from~\eqref{eq:unit-modulus-array-response}. 
For details, see Supplementary Material, Sec.\,S-IV.} 
into moduli $\rv{a}_{\!\scriptscriptstyle n}^{\scriptscriptstyle(j)}\!=\!|\widetilde{\unslant[-.25]{\varrho}}_{\!\scriptscriptstyle n}^{\scriptscriptstyle(j)}|
\!\in\!\realsetone{\,}_{\scriptscriptstyle\geq0}$ and nuisance phases $\unslant[-.25]{\varphi}_{\!\scriptscriptstyle n}^{\scriptscriptstyle(j)}\!=\!\angle\RVamplitude{s}{n}{j}$ 
and stack them into vectors
$\RVmodulivec{n}{}\!\coloneqq\!
\big[
{\rv{a}_{\!\scriptscriptstyle n}^{\scriptscriptstyle(1)}}\!
\ist\hdots\ist
{\rv{a}_{\!\scriptscriptstyle n}^{\scriptscriptstyle(J)}}
\big]^{\!\trp}\!\in\!\realsetone{J}_{\scriptscriptstyle\geq0}$
and 
$\RVphasevec{n}{}
\!\in\!\realsetone{\dimPhase}$.
We find that the level of coherence and effective aperture are fundamentally tied to the number $\dimPhase$ of nuisance phases to be estimated, leading to three different \glspl{pcrlb}:

(i) The \textit{noncoherent} \gls{pcrlb} is obtained by treating each component phase as a separate \gls{rv} per \gls{pa}, i.e., 
$\RVphasevec{n}{}\!=\!
\big[
\unslant[-.25]{\varphi}_{\!\scriptscriptstyle n}^{\scriptscriptstyle(1)}
\ist\hdots\ist
\unslant[-.25]{\varphi}_{\!\scriptscriptstyle n}^{\scriptscriptstyle(J)}
\big]^{\!\trp}$ and $\dimPhase\!=\!J$.
(ii) The \textit{coherent} \gls{pcrlb} is obtained by treating the phase 
$\unslant[-.25]{\varphi}_{\!\scriptscriptstyle n}^{\scriptscriptstyle(j)}
\!\eqqcolon\!
\unslant[-.25]{\varphi}_{\!\scriptscriptstyle n}\,\forall j\!\in\!\setAnchors$
as a single \gls{rv} common to the distributed \glspl{pa}, i.e., 
$\RVphasevec{n}{}\!=\!
\unslant[-.25]{\varphi}_{\!\scriptscriptstyle n}$
and $\dimPhase\!=1$.
(iii) The \textit{carrier-phase-based} \gls{pcrlb} is obtained by treating all phases as known, hence $\RVphasevec{n}{}\!=\![~]$ is an empty vector and $\dimPhase\!=\!0$. 
We stack all parameters into a joint state vector
$\RVetaglobal{n}\!\coloneqq\!
\big[ \RVstate{n}^\trp \iist {\RVphasevec{n}{}}^\trp \iist {\RVmodulivec{n}{}}^{\trp} \iist \RVetan{n}\big]^\trp \! \in \!\realsetone{\dimGlobal}$
of dimension $\dimGlobal\!=\!6+\dimPhase\!+\!J\!+\!1$.
For this \gls{pcrlb} derivation and the following experiments, we assume a common noise variance across all \glspl{pa}, i.e., $\etan{n}^{\scriptscriptstyle(j)}\!\eqqcolon\!\etan{n}\, \forall j \!\in\! \setAnchors$.

We are ultimately interested in obtaining the global~\gls{pcrlb}
\begin{align}\label{eq:posteriorCRLB}
	\PCRLB = 
	\big(
		\FIMglobal{n} + \FIMstep{n}{n\!-\!1}
	\big)^{-1} %
	\quad \in  \realset{\dimGlobal}{\dimGlobal}
\end{align}
that is a lower bound on the \gls{mse} matrix~\cite[eq.\,(29)]{VanTrees2007PCRLB} 
of any estimator%
\footnote{The expectation is to be taken under the joint \gls{pdf} $f(\etaglobal{0:n}, \observationn{1:n})$. 
The notation $\M{X} \succeq \M{0}$ is to be interpreted as $\M{X}$ being positive semidefinite~\cite{kay1993estimation}.} 
$\mathbb{E}\big((\RVetaglobalHat{n}\!-\RVetaglobal{n} ) (\RVetaglobalHat{n}\!-\RVetaglobal{n} )^\trp\big) \succeq \PCRLB$.
The \gls{pcrlb} matrix $\PCRLB$ is the inverse of the posterior information matrix 
$\FIMstep{n}{n}\!\coloneqq\!\FIMglobal{n} \!+\! \FIMstep{n}{n\!-\!1}$
that is computed through the information fusion of the information matrix $\FIMglobal{n}$ about the global parameter vector $\RVetaglobal{n}$ obtained from a snapshot of observations $\observationn{n}$ at the current time step $n$ with the predicted information matrix 
$\FIMstep{n}{n\!-\!1}$. 
Under a linear Gaussian state-transition \gls{pdf} 
$f(\etaglobal{n}|\etaglobal{n-1})\!=\!\mathcal{N}(\etaglobal{n};\transitionmatrix \etaglobal{n-1},\processNoiseCov)$,
the predicted information matrix is~\cite[eq.\,(16)]{Hernandez02PCRLB}
\begin{align}\label{eq:priorFIM}
	\FIMstep{n}{n\!-\!1} = 
	\left(
		\transitionmatrix \, \FIMstep{n\!-\!1}{n\!-\!1}^{-1} \, \transitionmatrix^\trp + \processNoiseCov
	\right)^{-1} \,
\end{align}
with state-transition matrix $\transitionmatrix$ and process noise covariance matrix $\processNoiseCov$.
Using the \gls{pcrlb} matrix in~\eqref{eq:posteriorCRLB}, we define the \textit{Bayesian} (i.e., posterior) \gls{peb} as $\PEB\!\coloneqq\!\sqrt{\tr\big(\big[\PCRLB\big]_{\scriptscriptstyle1:3,1:3}\big)}$. 

\subsection{Global Snapshot FIM}\label{sec:FIMg}  
The Bayesian snapshot information matrix $\FIMglobal{n} \!=\! \mathbb{E}_{\RVetaglobalSmall{n}}\!(\FIMclassic{n})$ is computed as the expectation under the marginal prior \gls{pdf} $f(\etaglobal{n})$ of the \textit{classic} snapshot \gls{fim}~\cite{Tichavsky98PCRLB}. 
Assuming that each anchor $j$ contributes independent information about $\RVetaglobal{n}$, i.e., assuming conditionally independent observations $\observation{n}{j}\!$, the classic snapshot \gls{fim}~\cite{Fascista23RadioStripesICC}
\begin{align}\label{eq:FIMclassic}
    \FIMclassic{n} = 
    \sum\nolimits_{\scriptscriptstyle j=1}^{\scriptscriptstyle J} \jacobgn{j} \FIMch{n}{j} {\jacobgn{j}}^\trp
    \quad \in  \realset{\dimGlobal}{\dimGlobal}
\end{align}
is computed as the sum of the local channel \gls{fim} $\FIMch{n}{j} \!\in\! \realset{\dimLocal}{\dimLocal}$ contributions from all $J$ \glspl{pa}, propagated via the Jacobian matrices $\jacobgn{j}\!\coloneqq\! \nicefrac{\partial {\etach{n}{j}}^{\!\!\trp}}{\partial \etaglobal{n}} \!\in\!  \realset{\dimGlobal}{\dimLocal}$ from local channel parameters to global parameters.
The Jacobian matrices in~\eqref{eq:FIMclassic} are derived in Supplementary Material, Sec.\,S-IV-B.
The \textit{classic} snapshot \gls{crlb} $\classicCRLB$ is the inverse of the submatrix $\big[\FIMclassic{n}\big]_{\bm{i}_{\text{\tiny{S}}},{\bm{i}_{\text{\tiny{S}}}}}$ of the snapshot \gls{fim} in \eqref{eq:FIMclassic} with ${\bm{i}_{\text{\tiny{S}}}}\!=\![1\!:\!3,7\!:\!\dimGlobal]$, evaluated at the true parameters $\etaglobal{n}$ that generated the observed data.
Because no prior information (i.e., $\FIMstep{n}{n\!-\!1}$) is sequentially injected, it lower-bounds the estimation error of unbiased estimators attainable from a snapshot of data observed at time $n$.
The \textit{classic} snapshot \gls{peb} is then defined as
$\PEB^{\text{\tiny F}}\!\coloneqq\!\sqrt{\tr\big(\big[\classicCRLB\big]_{\scriptscriptstyle1:3,1:3}\big)}$.

\begin{figure*}
    \vskip 0pt	
    \begin{center}
        \def\datapath{.}
        \setlength{\figurewidth}{0.97\linewidth}
        \setlength{\figureheight}{0.3\linewidth}
        \tikzexternaldisable    
           \input{figures/results/PEB}
        \tikzexternalenable
    \end{center}
    \vspace{-4mm}
    \captionof{figure}{Position RMSEs 
    vs. \glspl{peb} evaluated on synthetic data (left). 
    Cumulative frequency of the position errors over all runs and time steps (right).
    The \gls{mc} analysis used \num{5000} runs per estimator.
    The derived (P)\glspl{crlb} reveal three stages (i)--(iii) of information gain.
    }
    \label{fig:PEB}
\end{figure*}
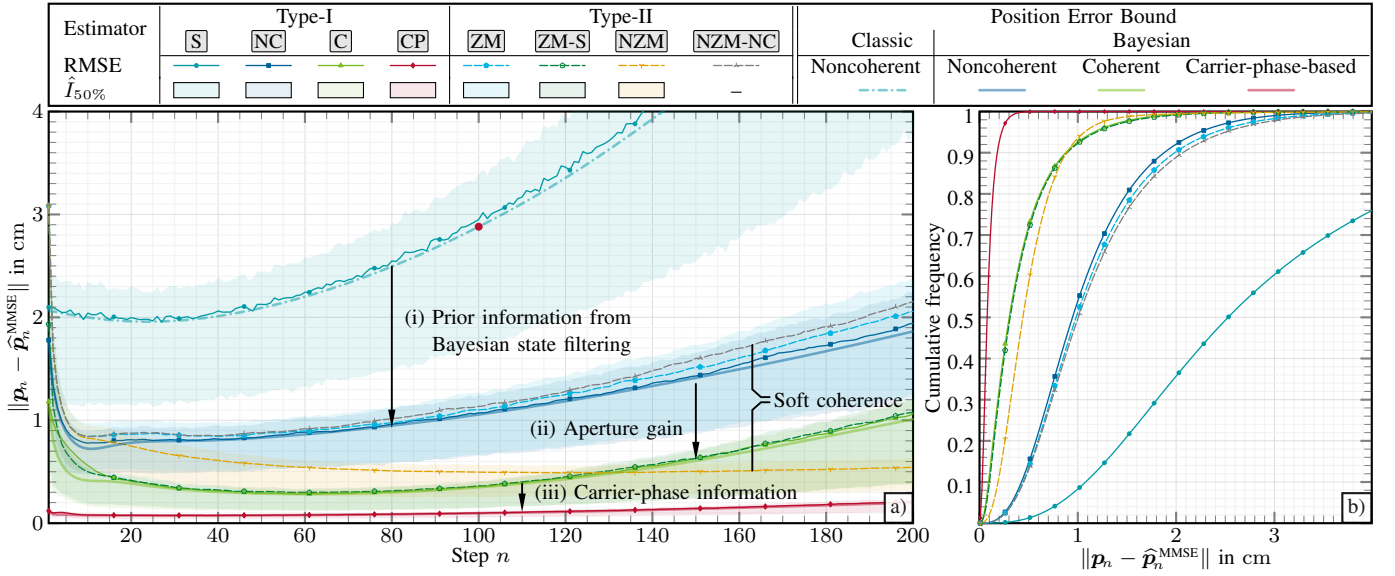

\subsection{Local Per-Anchor Channel FIM}\label{sec:FIMch}  

The \gls{pa}-local channel parameter vector 
$\RVetach{n}{j}\!\coloneqq\!\big[
{\RVelVecx{n}{j}}\iist 
{\RVazVecx{n}{j}} \iist  
{\RVdelayVecx{n}{j}} \iist 
\unslant[-.25]{\widetilde{\varphi}}_{\!\scriptscriptstyle n}^{\scriptscriptstyle(j)} \iist 
{\rv{a}_{\!\scriptscriptstyle n}^{\scriptscriptstyle(j)}} \iist
\RVetan{n}\big]^\trp \! \in \! \realsetone{\dimLocal}$
has dimension
$\dimLocal\!=\!6$
and contains the elevation and azimuth angles, delay, and amplitude phase and modulus of the \gls{los} component at each \gls{pa} $j$, together with the noise variance $\RVetan{n}$.
The local channel \gls{fim} at \gls{pa} $j$ is defined as~\cite[eq.\,(10)]{VanTrees2007PCRLB} 
$\FIMch{n}{j}\!\coloneqq\!-\mathbb{E}_{\RVobservation{n}{j}|\,\RVetach{n}{j}}\!
\big( 
\nabla_{\!\etach{n}{j}} 
(
	\nabla_{\!\etach{n}{j}}
	\! \ln f(\observation{n}{j}|\etach{n}{j})
)^\trp
\big)$, i.e.,
\begin{align}\label{eq:FIMch}
	\hspace{-2mm}\big[\FIMch{n}{j}\big]_{\scriptstyle \grave{\imath},\acute{\imath}} \! = 
	\begin{cases}
		\!
		\frac{\Nz}{{\etan{n}}^2},
		\qquad\qquad\qquad~
		[\etach{n}{j}]_{\scriptscriptstyle \grave{\imath}}
		\equiv 
		[\etach{n}{j}]_{\scriptscriptstyle \acute{\imath}}
		\equiv 
		\etan{n}
		\\
		\frac{2}{\etan{n}}\Re \left(
			\frac{\partial {\widetilde{{\varrho}}_{\!\scriptscriptstyle n}^{\scriptscriptstyle(j)^{\!\ast}}} {\signalatomnl{k}{n}{j}}^{\!\herm} }{\partial [\etach{n}{j}]_{\scriptscriptstyle \grave{\imath}}}
			\frac{\partial {\widetilde{{\varrho}}_{\!\scriptscriptstyle n}^{\scriptscriptstyle(j)}} \signalatomnl{k'}{n}{j} }{\partial [\etach{n}{j}]_{\scriptscriptstyle \acute{\imath}}}
		\right)
		\quad
		\text{else}
	\end{cases}\hspace{-2mm}
\end{align}
under a complex Gaussian likelihood $f(\observation{n}{j}|\etach{n}{j})$~\cite[Sec. 15.7]{kay1993estimation}.
The individual local channel \gls{fim} entries in~\eqref{eq:FIMch} are derived in Supplementary Material, Sec.\,S-IV-A.

\section{Experiment and Results}\label{sec:Results}  

The geometric scenario (apart from surfaces), system parameters, prior \glspl{pdf}, and state-transition \glspl{pdf} follow~\cite{Deutschmann26TSP}.

\subsection{Simulation Setup}

\paragraph*{Initial Prior PDFs}
For state filters, initial \gls{mt} state particles
$\{\particle{\state{0}}{p}\}_{p=1}^{P}$ are drawn from
$
f(\state{0})
=
\mathcal{N}(\state{0};\mathbf{x}_{\scriptscriptstyle 0},\mathbf{C}_{\scriptscriptstyle 0}),
$
with
$
\mathbf{C}_{\scriptscriptstyle 0}
=
\operatorname{blkdiag}
\big(
\sigma_{\scriptscriptstyle\mathrm{p},0}^2\eye{3},
\sigma_{\scriptscriptstyle\mathrm{v},0}^2\eye{3}
\big),
$
where $\mathbf{x}_{\scriptscriptstyle 0}$ is the mean of the initial prior
at the initial time $n\!=\!0$.
The Gaussian prior \gls{pdf} also provides the corresponding initialization of the \gls{pcrlb} recursion, detailed in Supplementary Material, Sec.\,S-IV.
Except for \estLabel{CP}, all estimators can alternatively be initialized from a broad uniform \gls{mt} state prior \gls{pdf}.

Noise variance particles are drawn from a uniform prior 
$f(\etan{0}^{\scriptscriptstyle(j)})\!=\!\mathcal{U}(\etan{0}^{\scriptscriptstyle(j)};\eta_{\text{\tiny min}},\eta_{\text{\tiny max}})$ with $\eta_{\text{\tiny min}}\!=\!10^{-9}$ and $\eta_{\text{\tiny max}}\!=\!10^{-4}$.
For the snapshot estimator \estLabel{S}, we draw particles from a prior \gls{pdf} $f(\state{n})
\!=\!
\mathcal{U}(\pos{n};\pos{n}^\star\!-\!\mathbf{p}_{\text{\tiny S}},\pos{n}^\star\!+\!\mathbf{p}_{\text{\tiny S}})
\,\mathcal{U}(\vel{n};-\mathbf{v}_{\text{\tiny max}},\mathbf{v}_{\text{\tiny max}})$ centered at the true position $\pos{n}^\star$ 
at each time step, with $\mathbf{p}_{\text{\tiny S}}\!=\! [0.2 \iist 0.2 \iist 0.2]^\trp \SI{}{\metre}$ chosen sufficiently small to mitigate particle-weight degeneracy, and an arbitrary velocity vector $\mathbf{v}_{\text{\tiny max}}\!\in\!\realsetone{3}$.

\paragraph*{MT State Evolution}
We choose a linear Gaussian \gls{mt} state-transition \gls{pdf} 
$f(\state{n}|\state{n-1})=\mathcal{N}(\state{n};\transitionmatrix_{\text{\tiny a}}\state{n-1},\processNoiseCov_{\text{\tiny a}})$ defined according to a \gls{ncv} state-transition model~\cite[Sec.\,6.3.2]{BarShalom04Tracking}
with state-transition matrix
\begin{align*}
\left[\transitionmatrix_{\text{\tiny a}}\right]_{\scriptscriptstyle \grave{\imath},\acute{\imath}} = 
    \begin{cases}
        1, & \grave{\imath}=\acute{\imath} \\
        \mathrm{T}, & 
        (\grave{\imath},\acute{\imath}) \!=\! (1,4)
        \lor 
        (\grave{\imath},\acute{\imath}) \!=\! (2,5)
        \lor 
        (\grave{\imath},\acute{\imath}) \!=\! (3,6)
        \\
        0, & \text{else}
    \end{cases}
\end{align*}
and sampling interval $\mathrm{T}\!=\!\SI{20}{\milli\second}$.
The kinematic process noise covariance matrix is
$\processNoiseCov_{\text{\tiny a}}\!=\!\sigma_{\scriptscriptstyle \mathrm{v}}^2\mathbf{\Gamma}\mathbf{\Gamma}^\trp$
with gain matrix $\mathbf{\Gamma}\!=\!\big[\frac{\mathrm{T}^2}{2}\eye{3},\mathrm{T}\eye{3}\big]^\trp$\!\! and 
acceleration-noise variance ${\sigma}_{\scriptscriptstyle \mathrm{v}}^2$.

\paragraph*{Noise Variance Evolution}
Following~\cite{LiaLeiMey:TSP2025}, the noise variance evolution is described by a Gamma \gls{pdf} $f(\etan{n}^{\scriptscriptstyle(j)}|\etan{n-1}^{\scriptscriptstyle(j)})\!=\!\mathcal{G}\big(\etan{n}^{\scriptscriptstyle(j)};\mathrm{c}_\eta,\frac{\etan{n-1}^{\scriptscriptstyle(j)}}{\mathrm{c}_\eta}\big)$ with mean $\etan{n-1}^{\scriptscriptstyle(j)}$ and variance $\frac{(\etan{n-1}^{\scriptscriptstyle(j)})^2}{\mathrm{c}_\eta}$, parameterized by a chosen constant $\mathrm{c}_\eta\!=\!\num{10}$.

\paragraph*{Amplitude State Evolution}
For \estLabel{NZM}, we assume that the amplitude state-transition \gls{pdf} factorizes as $f(\PFphi{s}{n}{1}|\PFphi{s}{n\!-\!1}{J})\!=\!f(\PFmu{s}{n}{1}|\PFmu{s}{n-1}{J})f(\PFgamma{s}{n}{1}|\PFgamma{s}{n-1}{J})$.
For both \estLabel{ZM} and \estLabel{NZM}, the prior variance evolves according to a Gamma \gls{pdf} $f(\PFgamma{s}{n}{1}|\PFgamma{s}{n-1}{J})\!=\!\mathcal{G}\big(\PFgamma{s}{n}{1};\mathrm{c}_\gamma,\nicefrac{\PFgamma{s}{n-1}{J}}{\mathrm{c}_\gamma}\big)$, 
parameterized by a constant $\mathrm{c}_\gamma\!=\!\num{100}$.
The prior mean of \estLabel{NZM} evolves according to a complex Gaussian \gls{pdf} $f(\PFmu{s}{n}{1}|\PFmu{s}{n-1}{J})\!=\!\mathcal{CN}(\PFmu{s}{n}{1};\PFmu{s}{n-1}{J},\sigma_\mu^{2})$ with a chosen standard deviation $\sigma_\mu\!=\!\num{0.03}$.

\paragraph*{System Parameters}
We choose a bandwidth $B\!=\!\SI{100}{\mega\hertz}$ with $\Nfrequency\!=\!10$ equally spaced frequency bins centered at $\fc \!=\!\SI{3.5}{\giga\hertz}$.
We use $J\!=\!4$ \glspl{pa}, each equipped with a \gls{ura} with $(\Nantennasy\!\times\!\Nantennasz)\!=\!(4\!\times\!4)$ antennas equally spaced at $\nicefrac{\lambda}{2}$.
This leads to an observation length $\Nz\!=\!160$.
\Glspl{pa} are arranged according to the scenario depicted in Fig.\,\ref{fig:scenario}, 
such that the \gls{mt} remains in the far field 
of each individual \gls{pa} but in the near field of the synthetic aperture formed by all \glspl{pa}.
We choose a time-invariant noise variance $\etan{n}$ s.t. $\mathrm{SNR}\!=\!10\log_{10}\big(\frac{P^{\scriptscriptstyle\mathrm{ch}}_{\scriptscriptstyle 1}}{\etan{1}}\big)$ is \SI{6}{\dB} at time step $n=1$ with channel power $P^{\scriptscriptstyle\mathrm{ch}}_{\scriptscriptstyle n}\!\coloneqq\!
\frac{1}{\Nz J}
\sum_{j=1}^{J}
\lVert
\varrho_{\scriptscriptstyle n}^{\scriptscriptstyle(j)}\steerVecx{k}{n}{j}
\rVert^2$.
For each estimator, we conduct an \gls{mc} analysis of \num{5000} estimation runs using $P\!=\!\num{45000}$ particles and with independently realized trajectories of $N\!=\!200$ time steps.
Random measurements are generated according to \eqref{eq:observation} and realizations of \textit{random} states $\{\RVstate{n}\}_{n=0}^{N}$ are 
generated according to the specified state-transition \glspl{pdf} (see Supplementary Material, Sec.\,S-IV for details) with ${\sigma}_{\scriptscriptstyle \mathrm{v}}\!=\!\SI{0.5}{\metre\per\second\squared}$, yielding the \gls{ncv} trajectory realizations shown in Fig.\,\ref{fig:scenario}.

\subsection{Results}
Estimator \glspl{rmse} and \glspl{peb} are shown in Fig.\,\ref{fig:PEB}\,a), while the cumulative frequencies of position errors are shown in Fig.\,\ref{fig:PEB}\,b).
For all estimators, $\interval{50\%}$ denotes the interquartile range of the empirical error distribution.

\paragraph*{Type-I Estimators}
The \gls{rmse} \lineref{pgf:RMSE-snapshot} of the noncoherent snapshot-based estimator \estLabel{S} approaches the \textit{classic} noncoherent snapshot \gls{peb} \lineref{pgf:PEB-classic} in Fig.\,\ref{fig:PEB}\,a).
It is the least accurate of all introduced estimators.

The \gls{rmse} \lineref{pgf:RMSE-Type-I-coh0} of the noncoherent Type-I filter \estLabel{NC} approaches the noncoherent \textit{Bayesian} \gls{peb} \lineref{pgf:PEB-nc} in Fig.\,\ref{fig:PEB}\,a).
Despite using the same unimodal likelihood function from Fig.\,\ref{fig:LHFs} a), the noncoherent Type-I filter gains accuracy over the snapshot-based estimator through state filtering, where the prediction message $\beta(\state{n})$ injects prior information from the previous time step $n-1$ into the current time step $n$.

The \gls{rmse} \lineref{pgf:RMSE-Type-I-coh2} of the coherent Type-I filter \estLabel{C} approaches the coherent Bayesian \gls{peb} \lineref{pgf:PEB-c} in Fig.\,\ref{fig:PEB}\,a).
It gains accuracy over the noncoherent Type-I filter as it leverages the combined aperture of the distributed \gls{pa} infrastructure, leading to high curvature---hence Fisher information---around the true mode of its likelihood function from Fig.\,\ref{fig:LHFs}\,b).
However, the resulting synthetic array is aliased,\!\footnote{The synthetic array spatially violates the Nyquist--Shannon sampling theorem.} causing multimodality in the likelihood function, which the state filter has to accommodate.
To leverage the combined aperture, the distributed \glspl{pa} must be frequency-synchronized and phase-calibrated, imposing additional requirements beyond those of the noncoherent estimators.

The \gls{rmse} \lineref{pgf:RMSE-Type-I-coh1} of the carrier-phase-based Type-I filter \estLabel{CP} approaches the carrier-phase-based Bayesian \gls{peb} \lineref{pgf:PEB-cp} in Fig.\,\ref{fig:PEB}\,a).
It gains accuracy over the coherent Type-I filter because it leverages the distance-dependent carrier-phase information, 
which results in very high curvature at the true mode of the likelihood function from Fig.\,\ref{fig:LHFs} c), but also in strong multimodality.
To leverage this carrier-phase information, all \glspl{pa} \textit{and} the \gls{mt} must be jointly frequency-synchronized and phase-calibrated, a strong assumption in practice.

\paragraph*{Type-II Estimators}
The \gls{rmse} \lineref{pgf:RMSE-Type-II-coh0} of the zero-mean Type-II filter \estLabel{ZM} approaches the noncoherent Bayesian \gls{peb} \lineref{pgf:PEB-nc} in Fig.\,\ref{fig:PEB}\,a), with performance almost identical to that of the noncoherent Type-I filter. 
However, it has to estimate noise variance and amplitude states as nuisance parameters, slightly affecting accuracy and convergence behavior.
This demonstrates that distributed processing with commonly used zero-mean Type-II likelihood models fails to gain aperture in \gls{dmimo}, as a common amplitude prior variance across distributed \glspl{pa} is insufficient to restore coherence.

As for \estLabel{C}, stacking of observations and steering vectors makes the \gls{rmse} \lineref{pgf:RMSE-Type-II-coh0-S} of \estLabel{ZM-S} approach the coherent Bayesian \gls{peb} \lineref{pgf:PEB-c}, even under a zero-mean prior model.
Stacking again restores coherence by enforcing a common amplitude (and phase) realization across distributed \glspl{pa}.

This contrasts with the nonzero-mean Type-II filter \estLabel{NZM}, whose \gls{rmse} \lineref{pgf:RMSE-Type-II-coh1} falls below the noncoherent and coherent \glspl{peb} on \textit{coherent} data.
However, it does not tightly approach the carrier-phase-based \gls{peb}, although a lower process-noise variance $\sigma_\mu^{2}$ and an informative prior $f(\PFmu{s}{0}{j})$ bring it much closer to that bound.
Multiple loopy \gls{bp} iterations within each snapshot~\cite{Frey97loopyBP} can further improve accuracy.
The \gls{rmse} \lineref{pgf:RMSE-Type-II-coh1-nc-data} of the nonzero-mean Type-II filter run on \textit{noncoherent} data%
\footnote{We generate these data by multiplying the noise-free term in each observation $\observation{n}{j}$ in \eqref{eq:observation} by a phasor $\exp\big(\mathrm{j} \varphi_{\scriptscriptstyle\mathrm{NC},n}^{\scriptscriptstyle(j)}\big)$ with phases $\unslant[-.25]{\varphi}_{\scriptscriptstyle\mathrm{NC},n}^{\scriptscriptstyle(j)}\!\sim\!\mathcal{U}(0,2\pi)$ i.i.d. across time steps and \glspl{pa}.}
\estLabel{NZM-NC} lies close to the noncoherent \gls{peb}.
Being able to automatically adapt to the coherence of the observed data is a property of this estimator that we term ``soft coherence''~\cite{Deutschmann26TSP}.
Hence, this estimator does not strictly rely on joint frequency synchronization and phase calibration. 
This contrasts with \estLabel{C}, \estLabel{CP}, and \estLabel{ZM-S}, all of which diverge on noncoherent data.

\subsection{Runtime Analysis}\label{sec:runtime}

All estimators are run in single precision in \textsc{Matlab} on both an Intel\textsuperscript{\textregistered} Xeon\textsuperscript{\textregistered} w7-3565X \gls{cpu} with \num{32} physical cores and an NVIDIA RTX
PRO\textsuperscript{\texttrademark} 4000 Blackwell \gls{gpu} with a peak
memory bandwidth of \SI{672}{\giga\byte\per\second}.
Our \gls{gpu}-accelerated implementation 
exploits parallelization across
particles, \glspl{pa}, and observation bins.
Due to the batched inner-product operations, it is expected to be memory-bandwidth-bound rather than compute-bound, a common behavior in algorithms with batched small-matrix operations~\cite{Haidar18smallMatrix}.
In practice, runtime therefore scales linearly with the memory traffic and, hence, the system parameters $\{P,J,\Nz\}$.
Runtimes of the individual estimators are summarized in Table~\ref{tab:algorithm-runtimes}.
The \gls{gpu}-accelerated implementations achieve up to \num{14}-fold speedups compared with their \gls{cpu} counterparts.
The runtimes of estimators \estLabel{S}, \estLabel{NC}, \estLabel{ZM}, and \estLabel{NZM} with distributed message computation could be further reduced through processing on separate \glspl{gpu}.

\begin{table}
    \centering
    \setlength{\tabcolsep}{3pt} 
    \caption{Estimator Runtimes per Time Step $n$}%
    \label{tab:algorithm-runtimes}%
    \vspace{-2mm}
    \begin{tabular}{@{}lr|
        c
        c
        >{\columncolor{black!6}}c
        >{\columncolor{black!6}}c|
        >{\columncolor{black!6}}c
        c
        c@{}}
        \specialrule{\heavyrulewidth}{0pt}{0pt}
        \multirow{2}{*}{Estimator} & ~ & \multicolumn{4}{c|}{Type-I} & \multicolumn{3   }{c}{Type-II} \\
         & & \estLabel{S} & \estLabel{NC} & \estLabel{C} & \estLabel{CP} & \estLabel{ZM-S} & \estLabel{ZM} & \estLabel{NZM}   \\
        \specialrule{\lightrulewidth}{0pt}{0pt}
        Runtime
        & CPU &
        \num{222} & \num{224} & \num{231} & \num{223} &  \num{887} & \num{825} & \num{815}\\
        (\SI{}{\milli\second}) & GPU &
        \num{31} & \num{36} & \num{23} & \num{26} & \num{62} & \num{70} & \num{71}\\
        \specialrule{\lightrulewidth}{0pt}{0pt}
        \multicolumn{2}{l|}{Speedup} & 
         \num{7.1} & \num{6.2} & \num{10.0} & \num{8.6} & \num{14.3} & \num{11.8} & \num{11.5} \\
         \specialrule{\lightrulewidth}{0pt}{0pt}
         \multicolumn{2}{l|}{Message comp.}
            & \multicolumn{2}{c}{distributed}
            & \multicolumn{3}{c}{\cellcolor{black!6}centralized}
            & \multicolumn{2}{c}{distributed} \\
        \specialrule{\heavyrulewidth}{0pt}{0pt}
    \end{tabular}
\end{table}

\section{Discussion}\label{sec:discussion}  
We summarize the findings of our experiment.
\begin{itemize}[leftmargin=5mm]
    \item Noncoherent filtering: \estLabel{NC} and \estLabel{ZM} perform time-sequential Bayesian filtering, gaining information over \estLabel{S}.
    They allow $J$ independently varying nuisance phase realizations, failing to exploit aperture gain in \gls{dmimo}.
    \item Coherent filtering: Stacking the observations and steering vectors in \estLabel{C} and \estLabel{ZM-S} 
    reduces the number of nuisance phase realizations to one, gaining aperture.
    \item Carrier-phase-based filtering: With no unknown amplitude phase, \estLabel{CP} can exploit distance-dependent carrier-phase information, albeit requiring perfect frequency synchronization and phase calibration.
    \item Soft-coherent filtering: The nonzero-mean amplitude prior model of \estLabel{NZM} allows the inference model to automatically operate either noncoherently or coherently depending on whether the data are coherent.
\end{itemize}
The (P)\glspl{crlb} in Fig.\,\ref{fig:PEB} show three stages of information gain:

Stage (i) is the information gain through sequential Bayesian filtering from the classic noncoherent \gls{crlb} \lineref{pgf:PEB-classic} to the noncoherent (Bayesian) \gls{pcrlb} \lineref{pgf:PEB-nc}. 
\ifthenelse{\equal{\IEEEversion}{true}}%
{%
    Both are based on the noncoherent likelihood function in Fig.\,\ref{fig:LHFs}\,a).
}%
{%
    Both are based on the noncoherent likelihood function in Fig.\,\ref{fig:LHFs}\,a), 
    where the least informative direction---meaning the direction with the lowest curvature---is the beam direction, i.e., the range direction.
    Estimation errors along this direction dominate both the classic \gls{peb} and the \gls{rmse} of the snapshot-based estimator.
}

Stage (ii) is information gain through aperture gain, leading to the coherent \gls{peb} \lineref{pgf:PEB-c}. 
\ifthenelse{\equal{\IEEEversion}{true}}%
{%
     Near-field (spherical-wavefront) processing across \glspl{pa} strongly reduces the width of the true mode of the likelihood function, as shown in Fig.\,\ref{fig:LHFs}\,b).
}%
{%
    Although the least informative direction remains the direction of the beams, as shown in Fig.\,\ref{fig:LHFs}\,b), near-field (spherical-wavefront) processing across \glspl{pa} strongly reduces the width of the true mode of the likelihood function.
}
Results for both our coherent Type-I filter \estLabel{C} and the coherent \gls{peb} show that the information gain is fundamentally tied to estimating only one nuisance phase common to all \glspl{pa} rather than a separate nuisance phase per \gls{pa} as in the noncoherent case.
The loss of position information due to jointly estimating nuisance parameters---particularly nuisance phases in \gls{dmimo} positioning---can be quantified using the equivalent \gls{fim} (EFIM)~\cite{Fascista25RadioStripes,VanTrees2002optimumASP,YuanShen10EFIM,YanjunHan16EFIM,ShenWin10widebandArrayCRLB}, whose Schur-complement representation reveals the information cost caused by information coupling between the nuisance parameters and the parameters of interest.

Stage (iii) is the information gain from leveraging carrier-phase information, which leads to the carrier-phase-based \gls{peb} \lineref{pgf:PEB-cp}. 
This is possible only if no nuisance phases need to be estimated. The resulting likelihood function in Fig.\,\ref{fig:LHFs}\,c) is strongly multimodal.
\ifthenelse{\equal{\IEEEversion}{true}}%
{%
}%
{%
    From a two-dimensional point of view, the least informative direction is now the cross-range direction, i.e., orthogonal to the beams, due to the shape of the carrier wavefronts.
}

\ifthenelse{\equal{\IEEEversion}{true}}%
{%
}%
{%
\begin{figure}
    \vskip 0pt	
    \begin{center}
        \def\datapath{.}
        \setlength{\figurewidth}{0.369\linewidth}
        \setlength{\figureheight}{0.23\linewidth}
        \setlength{\verticalSpace}{2pt}
        \tikzexternaldisable    
           \input{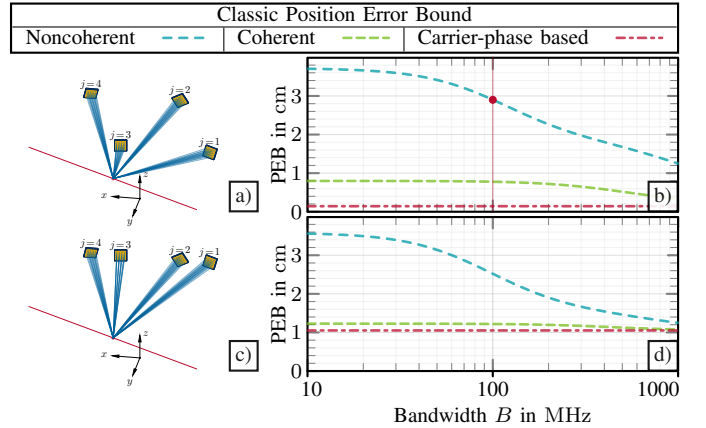}
        \tikzexternalenable
    \end{center}
    \vspace{-5mm}
    \captionof{figure}{Classic \glspl{peb} evaluated at time step $n\!=\!100$ for two different geometries: a) vertically and horizontally distributed \glspl{pa}, with corresponding \glspl{peb} in b), and c) vertically colocated and horizontally distributed \glspl{pa}, with corresponding \glspl{peb} in d).
    } 
    \label{fig:PEB-classic}
\end{figure}

At first glance, our experiments may seem to suggest that carrier-phase-based positioning offers significant gains over noncoherent processing. 
However, this is not always the case.
To illustrate the main phenomenology, we compare two different \gls{pa} placements in Figs.\,\ref{fig:PEB-classic} a) and c) with the resulting classic (snapshot) \glspl{crlb} evaluated at time step $n\!=\!100$ over a range of different bandwidths in Figs.\,\ref{fig:PEB-classic} b) and d).
The first \gls{pa} placement in Fig.\,\ref{fig:PEB-classic} a) is identical to the placement used in the scenario in Fig.\,\ref{fig:scenario}.
Here, coherent processing still shows significant accuracy gains.
A red dot \lineref{pgf:red-dot} marks the corresponding points in Figs.\,\ref{fig:PEB}\,a) and~\ref{fig:PEB-classic}\,b).
In the second \gls{pa} placement in Fig.\,\ref{fig:PEB-classic} c), the \glspl{pa} $j\!=\!1$ and $j\!=\!3$ are vertically colocated with the other \glspl{pa}, which changes the snapshot \glspl{peb} significantly:
The least informative Cartesian direction dominating the \gls{peb} in Fig.\,\ref{fig:PEB-classic}\,d) is the vertical $z$-direction. 
Neither coherent nor carrier-phase-based processing can increase the joint aperture in the vertical direction, hence both the coherent \lineref{pgf:snapshot-PEB-c} and carrier-phase-based \glspl{crlb} \lineref{pgf:snapshot-PEB-cp} lie very close together.
The noncoherent \gls{peb} \lineref{pgf:snapshot-PEB-nc} is initially higher because the least informative direction is the beam direction.
As the bandwidth increases, the least informative direction becomes vertical, and the noncoherent \gls{peb} approaches the coherent and carrier-phase-based \glspl{peb} because its vertical aperture is similar to theirs.
The information gain achievable through coherent processing depends on the geometry of the distributed \gls{pa} deployment relative to the \gls{mt}.

Most crucially, however, it depends on the coherence of the data.
}%
{%
    The positioning gains achievable through coherent processing depend on the deployment geometry relative to the \gls{mt} and on the coherence available in the data.
}%
Notably, apart from \estLabel{NZM-NC}, all estimators in Fig.\,\ref{fig:PEB} have been run on the same data, i.e., identical realizations of \gls{mt} state trajectories and noise vectors. 
This highlights that the available information gains can be realized only by an estimator that leverages the coherence present in the data.
Choosing a noncoherent approach a priori can improve robustness but often sacrifices aperture and hence estimation accuracy.

\section{Conclusion}\label{sec:conclusion}  
In this work, we presented a unified framework of estimators and performance bounds for scalable direct localization across coherence levels in \gls{dmimo} and panelized \gls{xlmimo} systems. 
Motivated by the increasing availability of spatially distributed arrays in future wireless infrastructures, these estimators enable noncoherent, coherent, or even carrier-phase-based Bayesian state estimation directly from noisy channel observations.
For all three coherence levels, the proposed estimators closely approach their respective (P)\glspl{crlb}.
We find that both the coherence level of \gls{dmimo} processing and the corresponding Fisher information are fundamentally tied to the number of phase parameters that are either actively used for positioning or treated as nuisance parameters.
In our Type-I filters, this relationship is explicit in the number of nuisance phases concentrated out: one per \gls{pa} for noncoherent processing, one phase common to all \glspl{pa} for coherent processing, and none for carrier-phase-based processing.
For our Type-II filters, we find that distributed processing with a zero-mean likelihood model is inherently noncoherent, whereas centralized observation stacking restores coherence by enforcing one common phase realization across \glspl{pa}.
Our recently proposed nonzero-mean Type-II likelihood model can automatically adapt to the coherence available in the data, a property that we term \textit{soft coherence}, while other coherent filters diverge on noncoherent data.
Partly owing to the temporal filtering of complex prior means, this model does not permit a clear classification in terms of the defined coherence levels.

Regarding computational complexity, the \textit{work} of our algorithms scales linearly with the observed data, 
whereas parallel reductions theoretically admit logarithmic scaling of their \textit{depth} under ideal parallel execution, making them well suited for future radio infrastructures with massive bandwidth and massive antenna deployments.
Our \gls{gpu}-accelerated particle-based implementation parallelizes over particles and \glspl{pa} and achieves runtimes of tens of milliseconds per time step, indicating potential for real-time operation.

In~\cite{Deutschmann26TSP}, we have shown that the framework presented here readily extends to challenging multipath channels with spatial nonstationarity.
Promising directions for future work include (i) online joint estimation of timing, frequency, and phase-calibration parameters~\cite{Ngo26DMIMOphaseCal} to accommodate partial and time-varying coherence, and (ii) neural enhancement of \gls{bp} for robustness to model mismatch~\cite{GarciaSatorras21NEBP,Mingchao24TSP_MOT,VenLeiTerWit:JSP2023,VenDeuFucKnoLei:FUSION2026}.

\bibliographystyle{IEEEtran}
\balance
\bibliography{IEEEabrv,bibliography,ThisPaper}

\ifthenelse{\equal{\IEEEversion}{true}}
{
}%
{%
    

\section*{Supplementary material}

# Extended Derivations for Coherent Direct D-MIMO Localization: Analysis of Coherence Levels

Benjamin J. B. Deutschmann, Lukas D'Angelo, Erik Leitinger, Klaus Witrisal
Institute of Comm. Networks and Satellite Comms., Graz University of Technology, Austria

This document provides derivations and particle-based implementations of the update messages for the belief propagation (BP) method introduced in the publication "Coherent Direct D-MIMO Localization: Analysis of Coherence Levels" [1] by the same authors.

## I. Approximate Amplitude State Update Message

### A. *Derivation of the Message Mean*

Let us define the shorthand notation $f_{\boldsymbol{z}} := f(\boldsymbol{z}_n^{(j)}|\boldsymbol{x}_n, \boldsymbol{\phi}_n, \eta_n^{(j)}) = \mathcal{CN}(\boldsymbol{z}_n^{(j)}; \boldsymbol{\mu}_n^{(j)}, \boldsymbol{C}_n^{(j)})$ with mean $\boldsymbol{\mu}_n^{(j)} = \mu_n \boldsymbol{\psi}_n^{(j)}$ and covariance matrix $\boldsymbol{C}_n^{(j)} = \eta_n^{(j)}\mathbf{I} + \gamma_n \boldsymbol{\psi}_n^{(j)} \boldsymbol{\psi}_n^{(j)H}$. Here, we define $\boldsymbol{\psi}_n^{(j)} := \boldsymbol{\psi}^{(j)}(\boldsymbol{x}_n)$ for brevity. Let $\kappa(\boldsymbol{\phi}_n; \boldsymbol{z}_n^{(j)})$ be the exact amplitude state update message. The approximate moment-matched complex Gaussian amplitude state update message is $\widetilde{\kappa}(\boldsymbol{\phi}_n; \boldsymbol{z}_n^{(j)}) = \mathcal{CN}\big(\boldsymbol{z}_n^{(j)}; \boldsymbol{\mu}_{\kappa,n}^{(j)}(\boldsymbol{\phi}_n), \boldsymbol{C}_{\kappa,n}^{(j)}(\boldsymbol{\phi}_n)\big)$.

The message mean is

$$
\begin{aligned}
\boldsymbol{\mu}_{\kappa,n}^{(j)}(\boldsymbol{\phi}_n) &= \mathbb{E}_{\kappa}\left(\boldsymbol{z}_n^{(j)}|\boldsymbol{\phi}_n\right) \\
&= \int \boldsymbol{z}_n^{(j)} \underbrace{\iint f(\boldsymbol{z}_n^{(j)}|\boldsymbol{x}_n, \boldsymbol{\phi}_n, \eta_n^{(j)})\beta(\boldsymbol{x}_n)\mathrm{d}\boldsymbol{x}_n \xi(\eta_n^{(j)})\mathrm{d}\eta_n^{(j)}}_{\kappa(\boldsymbol{\phi}_n;\boldsymbol{z}_n^{(j)})} \mathrm{d}\boldsymbol{z}_n^{(j)} \\
&= \iint \underbrace{\int \boldsymbol{z}_n^{(j)} \; f(\boldsymbol{z}_n^{(j)}|\boldsymbol{x}_n, \boldsymbol{\phi}_n, \eta_n^{(j)})\mathrm{d}\boldsymbol{z}_n^{(j)}}_{=\mathbb{E}_{f_{\boldsymbol{z}}}\left(\boldsymbol{z}_n^{(j)}\middle|\boldsymbol{x}_n,\boldsymbol{\phi}_n,\eta_n^{(j)}\right)} \beta(\boldsymbol{x}_n)\mathrm{d}\boldsymbol{x}_n \xi(\eta_n^{(j)})\mathrm{d}\eta_n^{(j)} \\
&= \int \mu_n \boldsymbol{\psi}_n^{(j)} \beta(\boldsymbol{x}_n)\mathrm{d}\boldsymbol{x}_n \underbrace{\int \xi(\eta_n^{(j)})\mathrm{d}\eta_n^{(j)}}_{=1} \\
&= \underbrace{\int \mu_n \boldsymbol{\psi}_n^{(j)} \beta(\boldsymbol{x}_n)\mathrm{d}\boldsymbol{x}_n}_{=:\boldsymbol{\mu}_{2,n}^{(j)}(\boldsymbol{\phi}_n)}
\end{aligned} \tag{1}
$$

Below, we derive the particle-based representation. First, we define particle-based representations of prediction messages:

**Particle-based representation: Agent state prediction message $\beta(\boldsymbol{x}_n)$**

$$\beta(\boldsymbol{x}_n) \approx \sum_{p=1}^{P} w_{\beta,n}^{(p)} \delta\big(\boldsymbol{x}_n - \boldsymbol{x}_n^{(p)}\big) \tag{2}$$

**Particle-based representation: amplitude state prediction message branch $\alpha(\boldsymbol{\phi}_n)$**

$$\alpha(\boldsymbol{\phi}_n) \approx \sum_{p=1}^{P} w_{\alpha,n}^{(p)} \delta\big(\boldsymbol{\phi}_n - \boldsymbol{\phi}_n^{(p)}\big) \tag{3}$$

**Particle-based representation: Noise variance prediction message $\xi(\eta_n^{(j)})$**

$$\xi(\eta_n^{(j)}) \approx \sum_{p=1}^{P} w_{\xi,n}^{(j,p)} \delta\big(\eta_n^{(j)} - \eta_n^{(j,p)}\big) \tag{4}$$

To arrive at a scalable solution that avoids prohibitive double or triple sums over $P^2$ or $P^3$ particle combinations, respectively, we employ the "stacking trick"[1] [2], denoted by $\circledast$ hereafter. We derive a particle-based representation for the message mean term $\boldsymbol{\mu}_{2,n}^{(j)}$.

[1] That is, after resampling, we stack the particles $\boldsymbol{x}_n^{(p)}$ and $\boldsymbol{\phi}_n^{(p')}$ into one $\left[\boldsymbol{x}_n^{(p)T} \;\; \boldsymbol{\phi}_n^{(p)T}\right]^T$.

**Particle-based representation: $\boldsymbol{\mu}_{2,n}^{(j)}(\phi_n)$**

$$\begin{aligned}\boldsymbol{\mu}_{2,n}^{(j)}(\boldsymbol{\phi}_n^{(p)}) &= \int \mu_n^{(p)}\boldsymbol{\psi}^{(j)}(\boldsymbol{x}_n)\ \beta(\boldsymbol{x}_n)\mathrm{d}\boldsymbol{x}_n \\ &\approx \int \mu_n^{(p)}\boldsymbol{\psi}^{(j)}(\boldsymbol{x}_n)\ \sum_{p'=1}^{P} w_{\beta,n}^{(p')}\delta\big(\boldsymbol{x}_n-\boldsymbol{x}_n^{(p')}\big)\mathrm{d}\boldsymbol{x}_n \\ &\overset{\star}{\approx} \underbrace{\Big(\sum_{p'=1}^{P} w_{\beta,n}^{(p')}\Big)}_{=1}\mu_n^{(p)}\boldsymbol{\psi}^{(j)}(\boldsymbol{x}_n^{(p)})\end{aligned} \tag{5}$$

### B. Derivation of the Message Covariance Matrix

The message covariance matrix is

$$\begin{aligned}\boldsymbol{C}_{\kappa,n}^{(j)}(\boldsymbol{\phi}_n) &= \mathbb{E}_\kappa\left(\boldsymbol{z}_n^{(j)}\boldsymbol{z}_n^{(j)H}\middle|\boldsymbol{\phi}_n\right) - \mathbb{E}_\kappa\left(\boldsymbol{z}_n^{(j)}\middle|\boldsymbol{\phi}_n\right)\mathbb{E}_\kappa\left(\boldsymbol{z}_n^{(j)}\middle|\boldsymbol{\phi}_n\right)^H \\ &= \boldsymbol{R}_{\kappa,j}(\boldsymbol{\phi}_n) - \boldsymbol{\mu}_{\kappa,n}^{(j)}(\boldsymbol{\phi}_n)\boldsymbol{\mu}_{\kappa,n}^{(j)H}(\boldsymbol{\phi}_n)\end{aligned} \tag{6}$$

with the second noncentral moment matrix

$$\begin{aligned}\boldsymbol{R}_{\kappa,j}(\boldsymbol{\phi}_n) &= \mathbb{E}_\kappa\left(\boldsymbol{z}_n^{(j)}\boldsymbol{z}_n^{(j)H}\middle|\boldsymbol{\phi}_n\right) = \int\left(\boldsymbol{z}_n^{(j)}\boldsymbol{z}_n^{(j)H}\right)\kappa(\boldsymbol{\phi}_n;\boldsymbol{z}_n^{(j)})\mathrm{d}\boldsymbol{z}_n^{(j)} \\ &= \int \boldsymbol{z}_n^{(j)}\boldsymbol{z}_n^{(j)H}\underbrace{\iint f(\boldsymbol{z}_n^{(j)}|\boldsymbol{x}_n,\boldsymbol{\phi}_n,\eta_n^{(j)})\ \beta(\boldsymbol{x}_n)\mathrm{d}\boldsymbol{x}_n\xi(\eta_n^{(j)})\mathrm{d}\eta_n^{(j)}}_{\kappa(\boldsymbol{\phi}_n;\boldsymbol{z}_n^{(j)})}\mathrm{d}\boldsymbol{z}_n^{(j)} \\ &= \iint\underbrace{\int \boldsymbol{z}_n^{(j)}\boldsymbol{z}_n^{(j)H}\ f(\boldsymbol{z}_n^{(j)}|\boldsymbol{x}_n,\boldsymbol{\phi}_n,\eta_n^{(j)})\mathrm{d}\boldsymbol{z}_n^{(j)}}_{=\mathbb{E}_{f_{\boldsymbol{z}}}\left(\boldsymbol{z}_n^{(j)}\boldsymbol{z}_n^{(j)H}\middle|\boldsymbol{x}_n,\boldsymbol{\phi}_n,\eta_n^{(j)}\right)}\ \beta(\boldsymbol{x}_n)\mathrm{d}\boldsymbol{x}_n\xi(\eta_n^{(j)})\mathrm{d}\eta_n^{(j)} \\ &= \iint\left(\boldsymbol{C}_n^{(j)}+\boldsymbol{\mu}_n^{(j)}\boldsymbol{\mu}_n^{(j)H}\right)\ \beta(\boldsymbol{x}_n)\mathrm{d}\boldsymbol{x}_n\xi(\eta_n^{(j)})\mathrm{d}\eta_n^{(j)} \\ &= \iint\left(\eta_n^{(j)}\mathbf{I}+\gamma_n\boldsymbol{\psi}_n^{(j)}\boldsymbol{\psi}_n^{(j)H}+\boldsymbol{\mu}_n^{(j)}\boldsymbol{\mu}_n^{(j)H}\right)\beta(\boldsymbol{x}_n)\mathrm{d}\boldsymbol{x}_n\xi(\eta_n^{(j)})\mathrm{d}\eta_n^{(j)}\,,\end{aligned} \tag{7}$$

recalling that $\boldsymbol{\mu}_n^{(j)} = \mu_n\boldsymbol{\psi}^{(j)}(\boldsymbol{x}_n)$.

For simplicity, let us first approximate the outer product term

$$\begin{aligned}\boldsymbol{K}_{\kappa,n}^{(j)}(\boldsymbol{\phi}_n) &:= \iint \boldsymbol{\mu}_n^{(j)}\boldsymbol{\mu}_n^{(j)H}\beta(\boldsymbol{x}_n)\mathrm{d}\boldsymbol{x}_n\xi(\eta_n^{(j)})\mathrm{d}\eta_n^{(j)} \\ &= \int\left(\mu_n\boldsymbol{\psi}_n^{(j)}\right)\left(\mu_n\boldsymbol{\psi}_n^{(j)}\right)^H\beta(\boldsymbol{x}_n)\mathrm{d}\boldsymbol{x}_n \\ &\overset{\diamond}{\approx} \boldsymbol{\mu}_{\kappa,n}^{(j)}(\boldsymbol{\phi}_n)\boldsymbol{\mu}_{\kappa,n}^{(j)H}(\boldsymbol{\phi}_n)\end{aligned} \tag{8}$$

through the outer product of the message means, as indicated by $\diamond$. This yields the approximate message covariance matrix below. Let $\eta_{\xi,n}^{(j)} := \int \eta_n^{(j)}\xi(\eta_n^{(j)})\mathrm{d}\eta_n^{(j)}$. Under our approximation, the second central moment matrix becomes

$$\begin{aligned}\boldsymbol{C}_{\kappa,n}^{(j)}(\boldsymbol{\phi}_n) &\approx \iint\left(\eta_n^{(j)}\mathbf{I}+\gamma_n\boldsymbol{\psi}_n^{(j)}\boldsymbol{\psi}_n^{(j)H}\right)\beta(\boldsymbol{x}_n)\mathrm{d}\boldsymbol{x}_n\xi(\eta_n^{(j)})\mathrm{d}\eta_n^{(j)} \\ &= \left(\int\left(\gamma_n\boldsymbol{\psi}_n^{(j)}\boldsymbol{\psi}_n^{(j)H}\right)\beta(\boldsymbol{x}_n)\mathrm{d}\boldsymbol{x}_n\right)\Big(\underbrace{\int\xi(\eta_n^{(j)})\mathrm{d}\eta_n^{(j)}}_{=1}\Big) + \mathbf{I}\underbrace{\int\eta_n^{(j)}\xi(\eta_n^{(j)})\mathrm{d}\eta_n^{(j)}}_{\triangleq\eta_{\xi,n}^{(j)}}\underbrace{\int\beta(\boldsymbol{x}_n)\mathrm{d}\boldsymbol{x}_n}_{=1} \\ &= \mathbf{I}\eta_{\xi,n}^{(j)} + \underbrace{\int\gamma_n\boldsymbol{\psi}_n^{(j)}\boldsymbol{\psi}_n^{(j)H}\beta(\boldsymbol{x}_n)\mathrm{d}\boldsymbol{x}_n}_{=:\boldsymbol{C}_{2,n}^{(j)}(\boldsymbol{\phi}_n)}\,.\end{aligned} \tag{9}$$

Now, let us define particle-based representations for these terms.

**Particle-based representation: $\eta^{(j)}_{\xi,n}$**

$$\begin{aligned}\eta^{(j)}_{\xi,n} &= \int \eta^{(j)}_n \xi(\eta^{(j)}_n)\mathrm{d}\eta^{(j)}_n \\ &\approx \int \eta^{(j)}_n \sum_{p=1}^{P} w^{(j,p)}_{\xi,n}\delta\big(\eta^{(j)}_n - \eta^{(j,p)}_n\big)\mathrm{d}\eta^{(j)}_n \\ &= \sum_{p=1}^{P} w^{(j,p)}_{\xi,n}\eta^{(j,p)}_n =: \widetilde{\eta}^{(j)}_{\xi,n} \end{aligned} \tag{10}$$

**Particle-based representation: $\boldsymbol{C}^{(j)}_{2,n}(\phi_n)$**

$$\begin{aligned}\boldsymbol{C}^{(j)}_{2,n}(\boldsymbol{\phi}^{(p)}_n) &= \int \gamma^{(p)}_n \boldsymbol{\psi}^{(j)}_n \boldsymbol{\psi}^{(j)}_n{}^H \beta(\boldsymbol{x}_n)\mathrm{d}\boldsymbol{x}_n \\ &\approx \int \gamma^{(p)}_n \boldsymbol{\psi}^{(j)}_n \boldsymbol{\psi}^{(j)}_n{}^H \sum_{p'=1}^{P} w^{(p')}_{\beta,n}\delta\big(\boldsymbol{x}_n - \boldsymbol{x}^{(p')}_n\big)\mathrm{d}\boldsymbol{x}_n \\ &\overset{\star}{\approx} \underbrace{\Big(\sum_{p'=1}^{P} w^{(p')}_{\beta,n}\Big)}_{=1} \gamma^{(p)}_n \boldsymbol{\psi}^{(j)}\boldsymbol{\psi}^{(j)H}(\boldsymbol{x}^{(p)}_n) \end{aligned} \tag{11}$$

Hence, we find a particle-based representation of the covariance matrix in (6) as

$$\begin{aligned}\boldsymbol{C}^{(j)}_{\kappa,n}(\boldsymbol{\phi}^{(p)}_n) &\overset{\diamond}{\approx} \eta^{(j)}_{\xi,n}\mathbf{I} + \boldsymbol{C}^{(j)}_{2,n}(\boldsymbol{\phi}^{(p)}_n) \\ &\overset{\star}{\approx} \Big(\sum_{p'=1}^{P} w^{(j,p')}_{\xi,n}\eta^{(j,p')}_n\Big)\mathbf{I} + \gamma^{(p)}_n \boldsymbol{\psi}^{(j)}\boldsymbol{\psi}^{(j)H}(\boldsymbol{x}^{(p)}_n) \\ &= \widetilde{\eta}^{(j)}_{\xi,n}\mathbf{I} + \gamma^{(p)}_n \boldsymbol{\psi}^{(j)}\boldsymbol{\psi}^{(j)H}(\boldsymbol{x}^{(p)}_n)\,. \end{aligned} \tag{12}$$

### C. Efficient Message Implementation

Let $\boldsymbol{\psi}^{(j,p)}_n := \boldsymbol{\psi}^{(j)}(\boldsymbol{x}^{(p)}_n)$ and $\boldsymbol{p}^{(j,p)}_n := \sqrt{\gamma^{(p)}_n}\boldsymbol{\psi}^{(j,p)}_n$. The covariance matrix of the complex Gaussian PDF in $\boldsymbol{z}^{(j)}_n$ used to approximate the update message, conditional on $\boldsymbol{\phi}_n$ and evaluated at the particle $\boldsymbol{\phi}^{(p)}_n$, can then be written as

$$\boldsymbol{C}^{(j)}_{\kappa,n}(\boldsymbol{\phi}^{(p)}_n) \approx \widetilde{\eta}^{(j)}_{\xi,n}\mathbf{I} + \boldsymbol{p}^{(j,p)}_n \boldsymbol{p}^{(j,p)}_n{}^H\,.$$

Hence, it is the sum of a scaled identity matrix and a rank-one outer product. This allows the inverse and determinant required to evaluate the complex Gaussian update-message approximation $\widetilde{\kappa}(\boldsymbol{\phi}^{(p)}_n;\boldsymbol{z}^{(j)}_n) = \mathcal{CN}\big(\boldsymbol{z}^{(j)}_n;\boldsymbol{\mu}^{(j)}_{\kappa,n}(\boldsymbol{\phi}^{(p)}_n),\boldsymbol{C}^{(j)}_{\kappa,n}(\boldsymbol{\phi}^{(p)}_n)\big)$ to be computed using inner products instead of directly inverting $\boldsymbol{C}^{(j)}_{\kappa,n}(\boldsymbol{\phi}^{(p)}_n) \in \mathbb{C}^{\mathrm{N_z}\times\mathrm{N_z}}$.

The matrix inversion lemma [3, eq. (156)] yields

$$\boldsymbol{C}^{(j,p)}_{\kappa,n}{}^{-1} := \Big(\boldsymbol{C}^{(j)}_{\kappa,n}(\boldsymbol{\phi}^{(p)}_n)\Big)^{-1} \approx \frac{1}{\widetilde{\eta}^{(j)}_{\xi,n}}\left(\mathbf{I}_{\mathrm{N_z}} - \frac{\boldsymbol{p}^{(j,p)}_n \boldsymbol{p}^{(j,p)}_n{}^H}{\widetilde{\eta}^{(j)}_{\xi,n} + \|\boldsymbol{p}^{(j,p)}_n\|^2}\right) \tag{13}$$

and allows expressions $\boldsymbol{a}^H\boldsymbol{C}^{(j,p)}_{\kappa,n}{}^{-1}\boldsymbol{b}$ with $\boldsymbol{a},\boldsymbol{b} \in \mathbb{C}^{\mathrm{N_z}}$ to be evaluated as

$$\boldsymbol{a}^H\boldsymbol{C}^{(j,p)}_{\kappa,n}{}^{-1}\boldsymbol{b} \approx \frac{1}{\widetilde{\eta}^{(j)}_{\xi,n}}\left(\boldsymbol{a}^H\boldsymbol{b} - \frac{\big(\boldsymbol{a}^H\boldsymbol{p}^{(j,p)}_n\big)\big(\boldsymbol{p}^{(j,p)}_n{}^H\boldsymbol{b}\big)}{\widetilde{\eta}^{(j)}_{\xi,n} + \|\boldsymbol{p}^{(j,p)}_n\|^2}\right) \tag{14}$$

in an order such that only inner products need to be computed. The generalized matrix determinant lemma [4, p. 420] gives

$$\det\Big(\boldsymbol{C}^{(j,p)}_{\kappa,n}\Big) \approx \big(\widetilde{\eta}^{(j)}_{\xi,n}\big)^{\mathrm{N_z}}\left(1 + \frac{\|\boldsymbol{p}^{(j,p)}_n\|^2}{\widetilde{\eta}^{(j)}_{\xi,n}}\right)\,. \tag{15}$$

Let $\boldsymbol{e}^{(j,p)}_{\kappa,n} := \boldsymbol{z}^{(j)}_n - \boldsymbol{\mu}^{(j)}_{\kappa,n}(\boldsymbol{\phi}^{(p)}_n)$ with $\boldsymbol{\mu}^{(j)}_{\kappa,n}(\boldsymbol{\phi}^{(p)}_n) = \boldsymbol{\mu}^{(j)}_{2,n}(\boldsymbol{\phi}^{(p)}_n)$. The approximate Gaussian amplitude state update message evaluated at particle

$\boldsymbol{\phi}_n^{(p)}$ is

$$
\begin{aligned}
\widetilde{\kappa}(\boldsymbol{\phi}_n^{(p)}; \boldsymbol{z}_n^{(j)}) &= \mathcal{CN}\Big(\boldsymbol{z}_n^{(j)}; \boldsymbol{\mu}_{\kappa,n}^{(j)}(\boldsymbol{\phi}_n^{(p)}), \boldsymbol{C}_{\kappa,n}^{(j)}(\boldsymbol{\phi}_n^{(p)})\Big) \\
&\approx \frac{\exp\left( -\frac{1}{\widetilde{\eta}_{\xi,n}^{(j)}} \left( \|\boldsymbol{e}_{\kappa,n}^{(j,p)}\|^2 - \frac{\left(\boldsymbol{e}_{\kappa,n}^{(j,p)\,H} \boldsymbol{p}_n^{(j,p)}\right)\left(\boldsymbol{p}_n^{(j,p)\,H} \boldsymbol{e}_{\kappa,n}^{(j,p)}\right)}{\widetilde{\eta}_{\xi,n}^{(j)} + \|\boldsymbol{p}_n^{(j,p)}\|^2} \right)\right)}{\pi^{\mathrm{N_z}} \big(\widetilde{\eta}_{\xi,n}^{(j)}\big)^{\mathrm{N_z}} \left(1 + \frac{\|\boldsymbol{p}_n^{(j,p)}\|^2}{\widetilde{\eta}_{\xi,n}^{(j)}}\right)} \\
&\overset{\star}{\approx} \frac{\exp\left( -\frac{1}{\widetilde{\eta}_{\xi,n}^{(j)}} \left( \left\|\boldsymbol{z}_n^{(j)} - \mu_n^{(p)} \boldsymbol{\psi}_n^{(j,p)}\right\|^2 - \frac{\gamma_n^{(p)} \left|\boldsymbol{\psi}_n^{(j,p)\,H} \left(\boldsymbol{z}_n^{(j)} - \mu_n^{(p)} \boldsymbol{\psi}_n^{(j,p)}\right)\right|^2}{\widetilde{\eta}_{\xi,n}^{(j)} + \gamma_n^{(p)} \left\|\boldsymbol{\psi}_n^{(j,p)}\right\|^2} \right)\right)}{\pi^{\mathrm{N_z}} \big(\widetilde{\eta}_{\xi,n}^{(j)}\big)^{\mathrm{N_z}} \left(1 + \frac{\gamma_n^{(p)} \left\|\boldsymbol{\psi}_n^{(j,p)}\right\|^2}{\widetilde{\eta}_{\xi,n}^{(j)}}\right)} .
\end{aligned}
\tag{16}
$$

Per PA-particle pair $(j, p)$, the required inner products are evaluated with complexity $\mathcal{O}\big(\mathrm{N_z}\big)$, leaving us with a total computational complexity of $\mathcal{O}\big(P\, J\, \mathrm{N_z}\big)$ for computing the amplitude state update message.

## II. Approximate MT State Update Message

### A. Derivation of the Message Mean

Let $\iota(\boldsymbol{x}_n; \boldsymbol{z}_n^{(j)})$ be the exact mobile terminal (MT) state update message. The approximate moment-matched complex Gaussian MT update message is $\tilde{\iota}(\boldsymbol{x}_n; \boldsymbol{z}_n^{(j)}) = \mathcal{CN}\big(\boldsymbol{z}_n^{(j)}; \boldsymbol{\mu}_{\iota,n}^{(j)}(\boldsymbol{x}_n), \boldsymbol{C}_{\iota,n}^{(j)}(\boldsymbol{x}_n)\big)$.

The message mean is

$$\begin{aligned}
\boldsymbol{\mu}_{\iota,n}^{(j)}(\boldsymbol{x}_n) &= \mathbb{E}_\iota\left(\boldsymbol{z}_n^{(j)} | \boldsymbol{x}_n\right) \\
&= \int \boldsymbol{z}_n^{(j)} \underbrace{\iint f(\boldsymbol{z}_n^{(j)}|\boldsymbol{x}_n, \boldsymbol{\phi}_n, \eta_n^{(j)}) \xi(\eta_n^{(j)}) \mathrm{d}\eta_n^{(j)} \alpha(\boldsymbol{\phi}_n) \mathrm{d}\boldsymbol{\phi}_n}_{\iota(\boldsymbol{x}_n; \boldsymbol{z}_n^{(j)})} \mathrm{d}\boldsymbol{z}_n^{(j)} \\
&= \iint \underbrace{\int \boldsymbol{z}_n^{(j)} f(\boldsymbol{z}_n^{(j)}|\boldsymbol{x}_n, \boldsymbol{\phi}_n, \eta_n^{(j)}) \mathrm{d}\boldsymbol{z}_n^{(j)}}_{=\mathbb{E}_{f_{\boldsymbol{z}}}\left(\boldsymbol{z}_n^{(j)} | \boldsymbol{x}_n, \boldsymbol{\phi}_n, \eta_n^{(j)}\right)} \xi(\eta_n^{(j)}) \mathrm{d}\eta_n^{(j)} \alpha(\boldsymbol{\phi}_n) \mathrm{d}\boldsymbol{\phi}_n \\
&= \int \mu_n \boldsymbol{\psi}^{(j)}(\boldsymbol{x}_n) \underbrace{\int \xi(\eta_n^{(j)}) \mathrm{d}\eta_n^{(j)}}_{=1} \alpha(\boldsymbol{\phi}_n) \mathrm{d}\boldsymbol{\phi}_n \\
&= \underbrace{\int \mu_n \boldsymbol{\psi}^{(j)}(\boldsymbol{x}_n) \alpha(\boldsymbol{\phi}_n) \mathrm{d}\boldsymbol{\phi}_n}_{=:\boldsymbol{\mu}_{1,n}^{(j)}(\boldsymbol{x}_n)}
\end{aligned} \tag{17}$$

Below, we derive the particle-based representation.

**Particle-based representation: $\mu_{1,n}^{(j)}$**

$$\begin{aligned}
\boldsymbol{\mu}_{1,n}^{(j)}(\boldsymbol{x}_n^{(p)}) &= \int \mu_n \boldsymbol{\psi}^{(j)}(\boldsymbol{x}_n^{(p)}) \alpha(\boldsymbol{\phi}_n) \mathrm{d}\boldsymbol{\phi}_n \\
&\approx \int \mu_n \boldsymbol{\psi}^{(j)}(\boldsymbol{x}_n^{(p)}) \Big( \sum_{p'=1}^{P} w_{\alpha,n}^{(p')} \delta\big(\boldsymbol{\phi}_n - \boldsymbol{\phi}_n^{(p')}\big)\Big) \mathrm{d}\boldsymbol{\phi}_n \\
&\overset{\star}{\approx} \Big( \sum_{p'=1}^{P} w_{\alpha,n}^{(p')}\Big) \boldsymbol{\psi}^{(j)}(\boldsymbol{x}_n^{(p)}) \int \mu_n \delta\big(\boldsymbol{\phi}_n - \boldsymbol{\phi}_n^{(p)}\big) \mathrm{d}\boldsymbol{\phi}_n \\
&= \Big( \sum_{p'=1}^{P} w_{\alpha,n}^{(p')}\Big) \mu_n^{(p)} \boldsymbol{\psi}^{(j)}(\boldsymbol{x}_n^{(p)})
\end{aligned} \tag{18}$$

### B. Derivation of the Message Covariance Matrix

The message covariance matrix is

$$\begin{aligned}
\boldsymbol{C}_{\iota,n}^{(j)}(\boldsymbol{x}_n) &= \mathbb{E}_\iota\left(\boldsymbol{z}_n^{(j)} \boldsymbol{z}_n^{(j)H} | \boldsymbol{x}_n\right) - \mathbb{E}_\iota\left(\boldsymbol{z}_n^{(j)} | \boldsymbol{x}_n\right) \mathbb{E}_\iota\left(\boldsymbol{z}_n^{(j)} | \boldsymbol{x}_n\right)^H \\
&= \boldsymbol{R}_{\iota,n}^{(j)}(\boldsymbol{x}_n) - \boldsymbol{\mu}_{\iota,n}^{(j)}(\boldsymbol{x}_n) \boldsymbol{\mu}_{\iota,n}^{(j)H}(\boldsymbol{x}_n)
\end{aligned} \tag{19}$$

with the second noncentral moment matrix

$$\begin{aligned}
\boldsymbol{R}_{\iota,n}^{(j)}(\boldsymbol{x}_n) &= \mathbb{E}_\iota\left(\boldsymbol{z}_n^{(j)} \boldsymbol{z}_n^{(j)H} | \boldsymbol{x}_n\right) = \int \left(\boldsymbol{z}_n^{(j)} \boldsymbol{z}_n^{(j)H}\right) \iota(\boldsymbol{x}_n; \boldsymbol{z}_n^{(j)}) \mathrm{d}\boldsymbol{z}_n^{(j)} \\
&= \int \boldsymbol{z}_n^{(j)} \boldsymbol{z}_n^{(j)H} \underbrace{\iint f(\boldsymbol{z}_n^{(j)}|\boldsymbol{x}_n, \boldsymbol{\phi}_n, \eta_n^{(j)}) \xi(\eta_n^{(j)}) \mathrm{d}\eta_n^{(j)} \alpha(\boldsymbol{\phi}_n) \mathrm{d}\boldsymbol{\phi}_n}_{\iota(\boldsymbol{x}_n; \boldsymbol{z}_n^{(j)})} \mathrm{d}\boldsymbol{z}_n^{(j)} \\
&= \iint \underbrace{\int \boldsymbol{z}_n^{(j)} \boldsymbol{z}_n^{(j)H} f(\boldsymbol{z}_n^{(j)}|\boldsymbol{x}_n, \boldsymbol{\phi}_n, \eta_n^{(j)}) \mathrm{d}\boldsymbol{z}_n^{(j)}}_{=\mathbb{E}_{f_{\boldsymbol{z}}}\left(\boldsymbol{z}_n^{(j)} \boldsymbol{z}_n^{(j)H} | \boldsymbol{x}_n, \boldsymbol{\phi}_n, \eta_n^{(j)}\right)} \xi(\eta_n^{(j)}) \mathrm{d}\eta_n^{(j)} \alpha(\boldsymbol{\phi}_n) \mathrm{d}\boldsymbol{\phi}_n \\
&= \iint \left(\boldsymbol{C}_n^{(j)} + \boldsymbol{\mu}_n^{(j)} \boldsymbol{\mu}_n^{(j)H}\right) \xi(\eta_n^{(j)}) \mathrm{d}\eta_n^{(j)} \alpha(\boldsymbol{\phi}_n) \mathrm{d}\boldsymbol{\phi}_n \\
&= \iint \left(\eta_n^{(j)} \mathbf{I} + \gamma_n \boldsymbol{\psi}_n^{(j)} \boldsymbol{\psi}_n^{(j)H} + \boldsymbol{\mu}_n^{(j)} \boldsymbol{\mu}_n^{(j)H}\right) \xi(\eta_n^{(j)}) \mathrm{d}\eta_n^{(j)} \alpha(\boldsymbol{\phi}_n) \mathrm{d}\boldsymbol{\phi}_n
\end{aligned} \tag{20}$$

Inserting $\boldsymbol{\mu}_n^{(j)} = \mu_n \boldsymbol{\psi}^{(j)}(\boldsymbol{x}_n)$, we again approximate

$$\begin{aligned}
\boldsymbol{K}_{\iota,n}^{(j)}(\boldsymbol{x}_n) &\coloneqq \iint \boldsymbol{\mu}_n^{(j)} \boldsymbol{\mu}_n^{(j)^H} \xi(\eta_n^{(j)}) \mathrm{d}\eta_n^{(j)} \alpha(\boldsymbol{\phi}_n) \mathrm{d}\boldsymbol{\phi}_n \\
&= \int \left(\mu_n \boldsymbol{\psi}_n^{(j)}\right) \left(\mu_n \boldsymbol{\psi}_n^{(j)}\right)^H \underbrace{\int \xi(\eta_n^{(j)}) \mathrm{d}\eta_n^{(j)}}_{=1} \alpha(\boldsymbol{\phi}_n) \mathrm{d}\boldsymbol{\phi}_n \\
&\overset{\diamond}{\approx} \boldsymbol{\mu}_{\iota,n}^{(j)}(\boldsymbol{x}_n) \boldsymbol{\mu}_{\iota,n}^{(j)^H}(\boldsymbol{x}_n)
\end{aligned} \tag{21}$$

through the outer product of message means. Under this approximation, the second central moment matrix becomes

$$\begin{aligned}
\boldsymbol{C}_{\iota,n}^{(j)}(\boldsymbol{x}_n) &\approx \iint \left(\eta_n^{(j)} \mathbf{I} + \gamma_n \boldsymbol{\psi}_n^{(j)} \boldsymbol{\psi}_n^{(j)^H}\right) \xi(\eta_n^{(j)}) \mathrm{d}\eta_n^{(j)} \alpha(\boldsymbol{\phi}_n) \mathrm{d}\boldsymbol{\phi}_n \\
&= \left( \int \left(\gamma_n \boldsymbol{\psi}_n^{(j)} \boldsymbol{\psi}_n^{(j)^H}\right) \alpha(\boldsymbol{\phi}_n) \mathrm{d}\boldsymbol{\phi}_n \underbrace{\int \xi(\eta_n^{(j)}) \mathrm{d}\eta_n^{(j)}}_{=1} \right) \\
&\quad + \mathbf{I} \underbrace{\int \eta_n^{(j)} \xi(\eta_n^{(j)}) \mathrm{d}\eta_n^{(j)}}_{\triangleq \eta_{\xi,n}^{(j)}} \underbrace{\int \alpha(\boldsymbol{\phi}_n) \mathrm{d}\boldsymbol{\phi}_n}_{=1} \\
&= \mathbf{I} \eta_{\xi,n}^{(j)} + \underbrace{\int \gamma_n \boldsymbol{\psi}_n^{(j)} \boldsymbol{\psi}_n^{(j)^H} \alpha(\boldsymbol{\phi}_n) \mathrm{d}\boldsymbol{\phi}_n}_{=: \boldsymbol{C}_{1,n}^{(j)}(\boldsymbol{x}_n)} .
\end{aligned} \tag{22}$$

**Particle-based representation: $\boldsymbol{C}_{1,n}^{(j)}$**

$$\begin{aligned}
\boldsymbol{C}_{1,n}^{(j)}(\boldsymbol{x}_n^{(p)}) &= \int \gamma_n \boldsymbol{\psi}^{(j)} \boldsymbol{\psi}^{(j)^H}(\boldsymbol{x}_n^{(p)}) \alpha(\boldsymbol{\phi}_n) \mathrm{d}\boldsymbol{\phi}_n \\
&\approx \int \gamma_n \boldsymbol{\psi}^{(j)} \boldsymbol{\psi}^{(j)^H}(\boldsymbol{x}_n^{(p)}) \Big( \sum_{p'=1}^{P} w_{\alpha,n}^{(p')} \delta\big(\boldsymbol{\phi}_n - \boldsymbol{\phi}_n^{(p')}\big) \Big) \mathrm{d}\boldsymbol{\phi}_n \\
&\overset{\star}{\approx} \Big( \sum_{p'=1}^{P} w_{\alpha,n}^{(p')} \Big) \boldsymbol{\psi}^{(j)} \boldsymbol{\psi}^{(j)^H}(\boldsymbol{x}_n^{(p)}) \int \gamma_n \delta\big(\boldsymbol{\phi}_n - \boldsymbol{\phi}_n^{(p)}\big) \mathrm{d}\boldsymbol{\phi}_n \\
&= \Big( \sum_{p'=1}^{P} w_{\alpha,n}^{(p')} \Big) \gamma_n^{(p)} \boldsymbol{\psi}_n^{(j,p)} \boldsymbol{\psi}_n^{(j,p)^H}
\end{aligned} \tag{23}$$

We find a particle-based representation of the covariance matrix in (19) as

$$\begin{aligned}
\boldsymbol{C}_{\iota,n}^{(j)}(\boldsymbol{x}_n^{(p)}) &= \boldsymbol{R}_{\iota,n}^{(j)}(\boldsymbol{x}_n^{(p)}) - \boldsymbol{\mu}_{\iota,n}^{(j)}(\boldsymbol{x}_n^{(p)}) \boldsymbol{\mu}_{\iota,n}^{(j)^H}(\boldsymbol{x}_n^{(p)}) \\
&\overset{\diamond}{\approx} \eta_{\xi,n}^{(j)} \mathbf{I} + \boldsymbol{C}_{1,n}^{(j)}(\boldsymbol{x}_n^{(p)}) \\
&\overset{\star}{\approx} \Big( \sum_{\acute{p}=1}^{P} w_{\xi,n}^{(j,\acute{p})} \eta_n^{(j,\acute{p})} \Big) \mathbf{I} + \Big( \sum_{p'=1}^{P} w_{\alpha,n}^{(p')} \Big) \gamma_n^{(p)} \boldsymbol{\psi}^{(j)} \boldsymbol{\psi}^{(j)^H}(\boldsymbol{x}_n^{(p)})
\end{aligned} \tag{24}$$

### *C. Efficient Message Implementation*

The covariance matrix of the complex Gaussian PDF used to approximate the update message, when interpreted as a PDF in $\boldsymbol{z}_n^{(j)}$, conditioned on $\boldsymbol{x}_n$, can again be written as

$$\boldsymbol{C}_{\iota,n}^{(j)}(\boldsymbol{x}_n^{(p)}) \approx \widetilde{\eta}_{\xi,n}^{(j)} \mathbf{I} + \boldsymbol{p}_n^{(j,p)} \boldsymbol{p}_n^{(j,p)^H} ,$$

i.e., it is again the sum of a scaled identity matrix and a rank-one outer product. This allows the inverse and determinant required to evaluate the complex Gaussian update-message approximation $\widetilde{\iota}\big(\boldsymbol{x}_n^{(p)}; \boldsymbol{z}_n^{(j)}\big) = \mathcal{CN}\big(\boldsymbol{z}_n^{(j)}; \boldsymbol{\mu}_{\iota,n}^{(j)}(\boldsymbol{x}_n^{(p)}), \boldsymbol{C}_{\iota,n}^{(j)}(\boldsymbol{x}_n^{(p)})\big)$ to be computed using inner products instead of directly inverting $\boldsymbol{C}_{\iota,n}^{(j)}(\boldsymbol{x}_n^{(p)}) \in \mathbb{C}^{\mathrm{N_z} \times \mathrm{N_z}}$.

The matrix inversion lemma [3, eq. (156)] yields

$$\boldsymbol{C}_{\iota,n}^{(j,p)^{-1}} \coloneqq \left(\boldsymbol{C}_{\iota,n}^{(j)}(\boldsymbol{x}_n^{(p)})\right)^{-1} \approx \frac{1}{\widetilde{\eta}_{\xi,n}^{(j)}} \left( \mathbf{I}_{\mathrm{N_z}} - \frac{\boldsymbol{p}_n^{(j,p)} \boldsymbol{p}_n^{(j,p)^H}}{\widetilde{\eta}_{\xi,n}^{(j)} + \|\boldsymbol{p}_n^{(j,p)}\|^2} \right) \tag{25}$$

and allows expressions $\boldsymbol{a}^H \boldsymbol{C}_{\iota,n}^{(j,p)^{-1}} \boldsymbol{b}$ with $\boldsymbol{a}, \boldsymbol{b} \in \mathbb{C}^{\mathrm{N_z}}$ to be evaluated as

$$\boldsymbol{a}^H \boldsymbol{C}_{\iota,n}^{(j,p)^{-1}} \boldsymbol{b} \approx \frac{1}{\widetilde{\eta}_{\xi,n}^{(j)}} \left( \boldsymbol{a}^H \boldsymbol{b} - \frac{\big(\boldsymbol{a}^H \boldsymbol{p}_n^{(j,p)}\big) \big(\boldsymbol{p}_n^{(j,p)^H} \boldsymbol{b}\big)}{\widetilde{\eta}_{\xi,n}^{(j)} + \|\boldsymbol{p}_n^{(j,p)}\|^2} \right) \tag{26}$$

in an order such that only inner products need to be computed. The generalized matrix determinant lemma gives

$$\det\!\Big(\boldsymbol{C}_{\iota,n}^{(j,p)}\Big) \approx \big(\widetilde{\eta}_{\xi,n}^{(j)}\big)^{\mathrm{N_z}} \left(1 + \frac{\|\boldsymbol{p}_n^{(j,p)}\|^2}{\widetilde{\eta}_{\xi,n}^{(j)}}\right). \tag{27}$$

Let $\boldsymbol{e}_{\iota,n}^{(j,p)} \coloneqq \boldsymbol{z}_n^{(j)} - \boldsymbol{\mu}_{\iota,n}^{(j)}(\boldsymbol{x}_n^{(p)})$ with $\boldsymbol{\mu}_{\iota,n}^{(j)}(\boldsymbol{x}_n^{(p)}) = \boldsymbol{\mu}_{1,n}^{(j)}(\boldsymbol{x}_n^{(p)})$. The approximate Gaussian MT state update message evaluated at particle $\boldsymbol{x}_n^{(p)}$ is

$$\begin{aligned}
\widetilde{\iota}(\boldsymbol{x}_n^{(p)};\boldsymbol{z}_n^{(j)}) &= \mathcal{CN}\Big(\boldsymbol{z}_n^{(j)};\boldsymbol{\mu}_{\iota,n}^{(j)}(\boldsymbol{x}_n^{(p)}),\boldsymbol{C}_{\iota,n}^{(j)}(\boldsymbol{x}_n^{(p)})\Big)\\
&\approx \frac{\exp\left(-\frac{1}{\widetilde{\eta}_{\xi,n}^{(j)}}\left(\|\boldsymbol{e}_{\iota,n}^{(j,p)}\|^2 - \frac{\left(\boldsymbol{e}_{\iota,n}^{(j,p)\,H}\boldsymbol{p}_n^{(j,p)}\right)\left(\boldsymbol{p}_n^{(j,p)\,H}\boldsymbol{e}_{\iota,n}^{(j,p)}\right)}{\widetilde{\eta}_{\xi,n}^{(j)}+\|\boldsymbol{p}_n^{(j,p)}\|^2}\right)\right)}{\pi^{\mathrm{N_z}}\big(\widetilde{\eta}_{\xi,n}^{(j)}\big)^{\mathrm{N_z}}\left(1+\frac{\|\boldsymbol{p}_n^{(j,p)}\|^2}{\widetilde{\eta}_{\xi,n}^{(j)}}\right)}\\
&\overset{\star}{\approx} \frac{\exp\left(-\frac{1}{\widetilde{\eta}_{\xi,n}^{(j)}}\left(\left\|\boldsymbol{z}_n^{(j)} - \mu_n^{(p)}\boldsymbol{\psi}_n^{(j,p)}\right\|^2 - \frac{\gamma_n^{(p)}\left|\boldsymbol{\psi}_n^{(j,p)\,H}\left(\boldsymbol{z}_n^{(j)}-\mu_n^{(p)}\boldsymbol{\psi}_n^{(j,p)}\right)\right|^2}{\widetilde{\eta}_{\xi,n}^{(j)}+\gamma_n^{(p)}\left\|\boldsymbol{\psi}_n^{(j,p)}\right\|^2}\right)\right)}{\pi^{\mathrm{N_z}}\big(\widetilde{\eta}_{\xi,n}^{(j)}\big)^{\mathrm{N_z}}\left(1+\frac{\gamma_n^{(p)}\left\|\boldsymbol{\psi}_n^{(j,p)}\right\|^2}{\widetilde{\eta}_{\xi,n}^{(j)}}\right)}.
\end{aligned} \tag{28}$$

Per PA-particle pair $(j,p)$, the required inner products are evaluated with complexity $\mathcal{O}\big(\mathrm{N_z}\big)$, leaving us with a total computational complexity of $\mathcal{O}\big(P\,J\,\mathrm{N_z}\big)$ for computing the MT state update message.

## III. Approximate Noise Variance Update Message

### A. Derivation of the Message Mean

Let $\nu(\eta_n^{(j)};\boldsymbol{z}_n^{(j)})$ be the exact noise variance update message. The approximate moment-matched complex Gaussian noise variance update message is $\widetilde{\nu}(\eta_n^{(j)};\boldsymbol{z}_n^{(j)}) = \mathcal{CN}(\boldsymbol{z}_n^{(j)};\boldsymbol{\mu}_{\nu,n}^{(j)},\boldsymbol{C}_{\nu,n}^{(j)}(\eta_n^{(j)}))$.

The message mean is

$$\begin{aligned}
\boldsymbol{\mu}_{\nu,n}^{(j)} &= \mathbb{E}_\nu\left(\boldsymbol{z}_n^{(j)}\middle|\eta_n^{(j)}\right) \qquad (29)\\
&= \int \boldsymbol{z}_n^{(j)} \underbrace{\iint f(\boldsymbol{z}_n^{(j)}|\boldsymbol{x}_n,\boldsymbol{\phi}_n,\eta_n^{(j)})\beta(\boldsymbol{x}_n)\mathrm{d}\boldsymbol{x}_n\alpha(\boldsymbol{\phi}_n)\mathrm{d}\boldsymbol{\phi}_n}_{\nu(\eta_n^{(j)};\boldsymbol{z}_n^{(j)})}\,\mathrm{d}\boldsymbol{z}_n^{(j)}\\
&= \iint \underbrace{\int \boldsymbol{z}_n^{(j)} f(\boldsymbol{z}_n^{(j)}|\boldsymbol{x}_n,\boldsymbol{\phi}_n,\eta_n^{(j)})\mathrm{d}\boldsymbol{z}_n^{(j)}}_{=\mathbb{E}_{f_{\boldsymbol{z}}}\left(\boldsymbol{z}_n^{(j)}\middle|\boldsymbol{x}_n,\boldsymbol{\phi}_n,\eta_n^{(j)}\right)} \beta(\boldsymbol{x}_n)\mathrm{d}\boldsymbol{x}_n\alpha(\boldsymbol{\phi}_n)\mathrm{d}\boldsymbol{\phi}_n\\
&= \underbrace{\iint \mu_n\boldsymbol{\psi}^{(j)}(\boldsymbol{x}_n)\beta(\boldsymbol{x}_n)\mathrm{d}\boldsymbol{x}_n\alpha(\boldsymbol{\phi}_n)\mathrm{d}\boldsymbol{\phi}_n}_{\triangleq\boldsymbol{\mu}_{3,n}^{(j)}} \qquad (30)
\end{aligned}$$

for which we find a particle-based representation as follows.

**Particle-based representation: $\boldsymbol{\mu}_{3,n}^{(j)}$**

$$\begin{aligned}
\boldsymbol{\mu}_{3,n}^{(j)} &= \iint \mu_n\boldsymbol{\psi}_n^{(j)}\beta(\boldsymbol{x}_n)\mathrm{d}\boldsymbol{x}_n\alpha(\boldsymbol{\phi}_n)\mathrm{d}\boldsymbol{\phi}_n\\
&\approx \iint \mu_n\boldsymbol{\psi}^{(j)}(\boldsymbol{x}_n)\sum_{p=1}^{P} w_{\beta,n}^{(p)}\delta(\boldsymbol{x}_n-\boldsymbol{x}_n^{(p)})\mathrm{d}\boldsymbol{x}_n \sum_{p'=1}^{P} w_{\alpha,n}^{(p')}\delta(\boldsymbol{\phi}_n-\boldsymbol{\phi}_n^{(p')})\mathrm{d}\boldsymbol{\phi}_n\\
&\overset{\star}{\approx} \left(\sum_{p'=1}^{P} w_{\alpha,n}^{(p')}\right)\sum_{p=1}^{P} w_{\beta,n}^{(p)}\iint \mu_n\boldsymbol{\psi}^{(j)}(\boldsymbol{x}_n)\delta(\boldsymbol{x}_n-\boldsymbol{x}_n^{(p)})\delta(\boldsymbol{\phi}_n-\boldsymbol{\phi}_n^{(p)})\mathrm{d}\boldsymbol{x}_n\mathrm{d}\boldsymbol{\phi}_n\\
&= \left(\sum_{p'=1}^{P} w_{\alpha,n}^{(p')}\right)\sum_{p=1}^{P} w_{\beta,n}^{(p)}\mu_n^{(p)}\boldsymbol{\psi}^{(j)}(\boldsymbol{x}_n^{(p)}) \qquad (31)
\end{aligned}$$

### B. Derivation of the Message Covariance Matrix

The message covariance matrix is

$$\begin{aligned}
\boldsymbol{C}_{\nu,n}^{(j)}(\eta_n^{(j)}) &= \mathbb{E}_\nu\left(\boldsymbol{z}_n^{(j)}\boldsymbol{z}_n^{(j)H}\middle|\eta_n^{(j)}\right) - \mathbb{E}_\nu\left(\boldsymbol{z}_n^{(j)}\middle|\eta_n^{(j)}\right)\mathbb{E}_\nu\left(\boldsymbol{z}_n^{(j)}\middle|\eta_n^{(j)}\right)^H\\
&= \boldsymbol{R}_{\nu,n}^{(j)}(\eta_n^{(j)}) - \boldsymbol{\mu}_{\nu,n}^{(j)}\boldsymbol{\mu}_{\nu,n}^{(j)H} \qquad (32)
\end{aligned}$$

with the second noncentral moment matrix

$$\begin{aligned}
\boldsymbol{R}_{\nu,n}^{(j)}(\eta_n^{(j)}) &= \mathbb{E}_\nu\left(\boldsymbol{z}_n^{(j)}\boldsymbol{z}_n^{(j)H}\middle|\eta_n^{(j)}\right) = \int\left(\boldsymbol{z}_n^{(j)}\boldsymbol{z}_n^{(j)H}\right)\nu(\eta_n^{(j)};\boldsymbol{z}_n^{(j)})\mathrm{d}\boldsymbol{z}_n^{(j)}\\
&= \int \boldsymbol{z}_n^{(j)}\boldsymbol{z}_n^{(j)H}\underbrace{\iint f(\boldsymbol{z}_n^{(j)}|\boldsymbol{x}_n,\boldsymbol{\phi}_n,\eta_n^{(j)})\beta(\boldsymbol{x}_n)\mathrm{d}\boldsymbol{x}_n\alpha(\boldsymbol{\phi}_n)\mathrm{d}\boldsymbol{\phi}_n}_{\nu(\eta_n^{(j)};\boldsymbol{z}_n^{(j)})}\,\mathrm{d}\boldsymbol{z}_n^{(j)}\\
&= \iint \underbrace{\int \boldsymbol{z}_n^{(j)}\boldsymbol{z}_n^{(j)H}\; f(\boldsymbol{z}_n^{(j)}|\boldsymbol{x}_n,\boldsymbol{\phi}_n,\eta_n^{(j)})\mathrm{d}\boldsymbol{z}_n^{(j)}}_{=\mathbb{E}_{f_{\boldsymbol{z}}}\left(\boldsymbol{z}_n^{(j)}\boldsymbol{z}_n^{(j)H}\middle|\boldsymbol{x}_n,\boldsymbol{\phi}_n,\eta_n^{(j)}\right)}\beta(\boldsymbol{x}_n)\mathrm{d}\boldsymbol{x}_n\alpha(\boldsymbol{\phi}_n)\mathrm{d}\boldsymbol{\phi}_n\\
&= \iint\left(\boldsymbol{C}_n^{(j)} + \boldsymbol{\mu}_n^{(j)}\boldsymbol{\mu}_n^{(j)H}\right)\beta(\boldsymbol{x}_n)\mathrm{d}\boldsymbol{x}_n\alpha(\boldsymbol{\phi}_n)\mathrm{d}\boldsymbol{\phi}_n\\
&= \iint\left(\eta_n^{(j)}\mathbf{I} + \gamma_n\boldsymbol{\psi}_n^{(j)}\boldsymbol{\psi}_n^{(j)H} + \boldsymbol{\mu}_n^{(j)}\boldsymbol{\mu}_n^{(j)H}\right)\beta(\boldsymbol{x}_n)\mathrm{d}\boldsymbol{x}_n\alpha(\boldsymbol{\phi}_n)\mathrm{d}\boldsymbol{\phi}_n \qquad (33)
\end{aligned}$$

After inserting $\boldsymbol{\mu}_n^{(j)} = \mu_n\boldsymbol{\psi}^{(j)}(\boldsymbol{x}_n)$, we again approximate

$$\begin{aligned}
\boldsymbol{K}_{\nu,n}^{(j)} &\coloneqq \iint \boldsymbol{\mu}_n^{(j)}\boldsymbol{\mu}_n^{(j)H}\beta(\boldsymbol{x}_n)\mathrm{d}\boldsymbol{x}_n\alpha(\boldsymbol{\phi}_n)\mathrm{d}\boldsymbol{\phi}_n\\
&= \iint\left(\mu_n\boldsymbol{\psi}_n^{(j)}\right)\left(\mu_n\boldsymbol{\psi}_n^{(j)}\right)^H\beta(\boldsymbol{x}_n)\mathrm{d}\boldsymbol{x}_n\alpha(\boldsymbol{\phi}_n)\mathrm{d}\boldsymbol{\phi}_n\\
&\overset{\diamond}{\approx} \boldsymbol{\mu}_{\nu,n}^{(j)}\boldsymbol{\mu}_{\nu,n}^{(j)H} \qquad (34)
\end{aligned}$$

through the outer product of message means. Under this approximation, the second central moment matrix becomes

$$
\begin{aligned}
\boldsymbol{C}_{\nu,n}^{(j)}(\eta_n^{(j)}) &\approx \iint \left(\eta_n^{(j)}\mathbf{I} + \gamma_n \boldsymbol{\psi}_n^{(j)} {\boldsymbol{\psi}_n^{(j)}}^H\right)\beta(\boldsymbol{x}_n)\mathrm{d}\boldsymbol{x}_n \alpha(\boldsymbol{\phi}_n)\mathrm{d}\boldsymbol{\phi}_n \\
&= \left(\iint \left(\gamma_n \boldsymbol{\psi}_n^{(j)} {\boldsymbol{\psi}_n^{(j)}}^H\right)\beta(\boldsymbol{x}_n)\mathrm{d}\boldsymbol{x}_n \alpha(\boldsymbol{\phi}_n)\mathrm{d}\boldsymbol{\phi}_n\right) + \mathbf{I}\eta_n^{(j)} \underbrace{\iint \alpha(\boldsymbol{\phi}_n)\mathrm{d}\boldsymbol{\phi}_n \beta(\boldsymbol{x}_n)\mathrm{d}\boldsymbol{x}_n}_{=1} \\
&= \mathbf{I}\eta_n^{(j)} + \underbrace{\iint \gamma_n \boldsymbol{\psi}_n^{(j)} {\boldsymbol{\psi}_n^{(j)}}^H \beta(\boldsymbol{x}_n)\mathrm{d}\boldsymbol{x}_n \alpha(\boldsymbol{\phi}_n)\mathrm{d}\boldsymbol{\phi}_n}_{\triangleq \boldsymbol{C}_{3,n}^{(j)}} \; .
\end{aligned} \tag{35}
$$

with a covariance matrix term that could be implemented through the following particle-based representation.

**Particle-based representation: $\boldsymbol{C}_{3,n}^{(j)}$ (full rank)**

$$
\begin{aligned}
\boldsymbol{C}_{3,n}^{(j)} &= \iint \gamma_n \boldsymbol{\psi}_n^{(j)} {\boldsymbol{\psi}_n^{(j)}}^H \beta(\boldsymbol{x}_n)\mathrm{d}\boldsymbol{x}_n \alpha(\boldsymbol{\phi}_n)\mathrm{d}\boldsymbol{\phi}_n \\
&\approx \iint \gamma_n \boldsymbol{\psi}^{(j)} {\boldsymbol{\psi}^{(j)}}^H(\boldsymbol{x}_n) \sum_{p=1}^{P} w_{\beta,n}^{(p)} \delta(\boldsymbol{x}_n - \boldsymbol{x}_n^{(p)})\mathrm{d}\boldsymbol{x}_n \Big(\sum_{p'=1}^{P} w_{\alpha,n}^{(p')} \delta(\boldsymbol{\phi}_n - \boldsymbol{\phi}_n^{(p')})\mathrm{d}\boldsymbol{\phi}_n\Big) \\
&\overset{\star}{\approx} \Big(\sum_{p'=1}^{P} w_{\alpha,n}^{(p')}\Big) \iint \gamma_n \boldsymbol{\psi}^{(j)} {\boldsymbol{\psi}^{(j)}}^H(\boldsymbol{x}_n) \sum_{p=1}^{P} w_{\beta,n}^{(p)} \delta(\boldsymbol{x}_n - \boldsymbol{x}_n^{(p)})\mathrm{d}\boldsymbol{x}_n \delta(\boldsymbol{\phi}_n - \boldsymbol{\phi}_n^{(p)})\mathrm{d}\boldsymbol{\phi}_n \\
&= \Big(\sum_{p'=1}^{P} w_{\alpha,n}^{(p')}\Big) \sum_{p=1}^{P} w_{\beta,n}^{(p)} \gamma_n^{(p)} \boldsymbol{\psi}^{(j)} {\boldsymbol{\psi}^{(j)}}^H(\boldsymbol{x}_n^{(p)})
\end{aligned} \tag{36}
$$

However, the matrix $\boldsymbol{C}_{3,n}^{(j)}$ as written in (36) is of full rank in practice, which leads to computationally expensive approximate moment-matched update messages $\widetilde{\nu}(\eta_n^{(j)}; \boldsymbol{z}_n^{(j)})$ that do not scale linearly in the length $\mathrm{N_z}$ of the data. To arrive at a computationally efficient implementation, we make use of the outer-product approximation $\boldsymbol{C}_{3,n}^{(j)} \approx: \boldsymbol{m}_{3,n}^{(j)} {\boldsymbol{m}_{3,n}^{(j)}}^H$ from [5]. Let us first define a particle-based representation of $\boldsymbol{m}_{3,n}^{(j)}$.

**Particle-based representation: $\boldsymbol{m}_{3,n}^{(j)}$**

$$
\begin{aligned}
\boldsymbol{m}_{3,n}^{(j)} &:= \iint \sqrt{\gamma_n} \boldsymbol{\psi}^{(j)}(\boldsymbol{x}_n) \beta(\boldsymbol{x}_n)\mathrm{d}\boldsymbol{x}_n \alpha(\boldsymbol{\phi}_n)\mathrm{d}\boldsymbol{\phi}_n \\
&\approx \iint \sqrt{\gamma_n} \boldsymbol{\psi}^{(j)}(\boldsymbol{x}_n) \sum_{p=1}^{P} w_{\beta,n}^{(p)} \delta(\boldsymbol{x}_n - \boldsymbol{x}_n^{(p)})\mathrm{d}\boldsymbol{x}_n \Big(\sum_{p'=1}^{P} w_{\alpha,n}^{(p')} \delta(\boldsymbol{\phi}_n - \boldsymbol{\phi}_n^{(p')})\mathrm{d}\boldsymbol{\phi}_n\Big) \\
&\overset{\star}{\approx} \Big(\sum_{p'=1}^{P} w_{\alpha,n}^{(p')}\Big) \sum_{p=1}^{P} w_{\beta,n}^{(p)} \sqrt{\gamma_n{}^{(p)}} \boldsymbol{\psi}^{(j)}(\boldsymbol{x}_n^{(p)})
\end{aligned} \tag{37}
$$

We then redefine the particle-based representation in (36) using the outer-product approximation

**Particle-based representation: $\boldsymbol{C}_{3,n}^{(j)}$ (rank-one outer-product approximation)**

$$
\boldsymbol{C}_{3,n}^{(j)} \approx \Big(\sum_{p'=1}^{P} w_{\alpha,n}^{(p')}\Big) \sum_{p=1}^{P} w_{\beta,n}^{(p)} \gamma_n^{(p)} \boldsymbol{\psi}^{(j)} {\boldsymbol{\psi}^{(j)}}^H(\boldsymbol{x}_n^{(p)}) \overset{\diamond}{\approx} \boldsymbol{m}_{3,n}^{(j)} {\boldsymbol{m}_{3,n}^{(j)}}^H \; . \tag{38}
$$

We find a particle-based representation of the covariance matrix in (32) as

$$
\begin{aligned}
\boldsymbol{C}_{\nu,n}^{(j)}(\eta_n^{(j,p)}) &= \boldsymbol{R}_{\nu,n}^{(j)}(\eta_n^{(j,p)}) - \boldsymbol{\mu}_{\nu,n}^{(j)} {\boldsymbol{\mu}_{\nu,n}^{(j)}}^H \\
&\approx \eta_n^{(j,p)}\mathbf{I} + \boldsymbol{C}_{3,n}^{(j)} \\
&\overset{\star}{\approx} \eta_n^{(j,p)}\mathbf{I} + \Big(\sum_{\grave{p}=1}^{P} w_{\alpha,n}^{(\grave{p})}\Big) \sum_{\acute{p}=1}^{P} w_{\beta,n}^{(\acute{p})} \gamma_n^{(\acute{p})} \boldsymbol{\psi}^{(j)} {\boldsymbol{\psi}^{(j)}}^H(\boldsymbol{x}_n^{(\acute{p})}) \\
&\overset{\diamond}{\approx} \eta_n^{(j,p)}\mathbf{I} + \Bigg( \underbrace{\bigg(\Big(\sum_{\grave{p}=1}^{P} w_{\alpha,n}^{(\grave{p})}\Big) \sum_{\acute{p}=1}^{P} w_{\beta,n}^{(\acute{p})} \sqrt{\gamma_n^{(\acute{p})}} \boldsymbol{\psi}^{(j)}(\boldsymbol{x}_n^{(\acute{p})})\bigg)}_{\triangleq \boldsymbol{m}_{3,n}^{(j)}} \\
&\qquad \times \bigg(\Big(\sum_{\grave{p}=1}^{P} w_{\alpha,n}^{(\grave{p})}\Big) \sum_{\acute{p}=1}^{P} w_{\beta,n}^{(\acute{p})} \sqrt{\gamma_n^{(\acute{p})}} \boldsymbol{\psi}^{(j)}(\boldsymbol{x}_n^{(\acute{p})})\bigg)^H \Bigg) \\
&= \eta_n^{(j,p)}\mathbf{I} + \boldsymbol{m}_{3,n}^{(j)} {\boldsymbol{m}_{3,n}^{(j)}}^H \; .
\end{aligned} \tag{39}
$$

### C. Efficient Message Implementation

The covariance matrix of the complex Gaussian PDF used to approximate the update message, when interpreted as a PDF in $\boldsymbol{z}_n^{(j)}$, conditioned on $\eta_n^{(j)}$, is

$$\boldsymbol{C}_{\nu,n}^{(j)}(\eta_n^{(j,p)}) \approx \eta_n^{(j,p)}\mathbf{I} + \boldsymbol{m}_{3,n}^{(j)}\boldsymbol{m}_{3,n}^{(j)\,H}, \tag{40}$$

i.e., it is again the sum of a scaled identity matrix and a rank-one outer product. This allows the inverse and determinant required to evaluate the complex Gaussian update-message approximation $\widetilde{\nu}\big(\eta_n^{(j,p)}; \boldsymbol{z}_n^{(j)}\big) = \mathcal{CN}\big(\boldsymbol{z}_n^{(j)}; \boldsymbol{\mu}_{\nu,n}^{(j)}, \boldsymbol{C}_{\nu,n}^{(j)}(\eta_n^{(j,p)})\big)$ to be computed using inner products instead of directly inverting $\boldsymbol{C}_{\nu,n}^{(j)}(\eta_n^{(j,p)}) \in \mathbb{C}^{\mathrm{N_z}\times\mathrm{N_z}}$.

The matrix inversion lemma [3, eq. (156)] yields

$$\boldsymbol{C}_{\nu,n}^{(j,p)^{-1}} := \Big(\boldsymbol{C}_{\nu,n}^{(j)}(\eta_n^{(j,p)})\Big)^{-1} \approx \frac{1}{\eta_n^{(j,p)}}\left(\mathbf{I}_{\mathrm{N_z}} - \frac{\boldsymbol{m}_{3,n}^{(j)}\boldsymbol{m}_{3,n}^{(j)\,H}}{\eta_n^{(j,p)} + \|\boldsymbol{m}_{3,n}^{(j)}\|^2}\right) \tag{41}$$

and allows expressions $\boldsymbol{a}^H\boldsymbol{C}_{\nu,n}^{(j,p)^{-1}}\boldsymbol{b}$ with $\boldsymbol{a}, \boldsymbol{b} \in \mathbb{C}^{\mathrm{N_z}}$ to be evaluated as

$$\boldsymbol{a}^H\boldsymbol{C}_{\nu,n}^{(j,p)^{-1}}\boldsymbol{b} \approx \frac{1}{\eta_n^{(j,p)}}\left(\boldsymbol{a}^H\boldsymbol{b} - \frac{\big(\boldsymbol{a}^H\boldsymbol{m}_{3,n}^{(j)}\big)\big(\boldsymbol{m}_{3,n}^{(j)\,H}\boldsymbol{b}\big)}{\eta_n^{(j,p)} + \|\boldsymbol{m}_{3,n}^{(j)}\|^2}\right) \tag{42}$$

in an order such that only inner products need to be computed. The generalized matrix determinant lemma gives

$$\det\Big(\boldsymbol{C}_{\nu,n}^{(j,p)}\Big) \approx \big(\eta_n^{(j,p)}\big)^{\mathrm{N_z}}\left(1 + \frac{\|\boldsymbol{m}_{3,n}^{(j)}\|^2}{\eta_n^{(j,p)}}\right). \tag{43}$$

Let $\boldsymbol{e}_{\nu,n}^{(j)} := \boldsymbol{z}_n^{(j)} - \boldsymbol{\mu}_{\nu,n}^{(j)}$ with $\boldsymbol{\mu}_{\nu,n}^{(j)} = \boldsymbol{\mu}_{3,n}^{(j)}$. The approximate Gaussian noise variance update message evaluated at particle $\eta_n^{(j,p)}$ is

$$\begin{aligned}\widetilde{\nu}(\eta_n^{(j,p)}; \boldsymbol{z}_n^{(j)}) &= \mathcal{CN}\Big(\boldsymbol{z}_n^{(j)}; \boldsymbol{\mu}_{\nu,n}^{(j)}, \boldsymbol{C}_{\nu,n}^{(j)}(\eta_n^{(j,p)})\Big) \\ &\overset{\diamond}{\approx} \frac{\exp\left(-\frac{1}{\eta_n^{(j,p)}}\left(\|\boldsymbol{e}_{\nu,n}^{(j)}\|^2 - \frac{\big|\boldsymbol{e}_{\nu,n}^{(j)\,H}\boldsymbol{m}_{3,n}^{(j)}\big|^2}{\eta_n^{(j,p)} + \|\boldsymbol{m}_{3,n}^{(j)}\|^2}\right)\right)}{\pi^{\mathrm{N_z}}\big(\eta_n^{(j,p)}\big)^{\mathrm{N_z}}\left(1 + \frac{\|\boldsymbol{m}_{3,n}^{(j)}\|^2}{\eta_n^{(j,p)}}\right)}.\end{aligned} \tag{44}$$

Computing the particle-independent terms $\boldsymbol{\mu}_{3,n}^{(j)}$ and $\boldsymbol{m}_{3,n}^{(j)}$ has computational complexity $\mathcal{O}\big(P\,J\,\mathrm{N_z}\big)$. The particle-independent inner products are evaluated with complexity $\mathcal{O}\big(\mathrm{N_z}\big)$ once per PA and reused for all $P$ noise variance particles, leaving the total computational complexity of the algorithm at $\mathcal{O}\big(P\,J\,\mathrm{N_z}\big)$. Note, however, that operations over $P$ particles and $J$ PAs can be parallelized.


}

\end{document}


\title{
{Supplementary Material for\\
\PaperTitle}}

\author{\IEEEauthorblockN{Benjamin J.\,B. Deutschmann, Lukas D'Angelo, Erik Leitinger, Klaus Witrisal}

\IEEEauthorblockA{Institute of Comm. Networks and Satellite Comms., Graz University of Technology, Graz, Austria}}

\maketitle

\noindent%
This manuscript provides derivations of the \gls{bp} method and \glspl{pcrlb} for the publication ``\PaperTitle''~\cite{ThisPaper} by the same authors.


\section{CP Amplitude and Noise Variance Estimators}\label{Ssec:CP-ML-estimators}

The likelihood function for \estLabel{CP} is
\begin{align}\label{Seq:conditional-CP-likelihood}
    f_{\scriptscriptstyle\mathrm{I}}\big(\observationn{n}|\state{n},\amplitudeBarCP{n},\etan{n}\big) 
    &= 
    \CN\big(\observationn{n};
    \amplitudeBarCP{n}
    \steerVecBar(\state{n}),
    \etan{n}\eye{J\Nz}\big)
    \\
    &=
    (\pi\etan{n})^{-J \Nz}
    \exp
    \Big(  
    \!-\!\frac{1}{\etan{n}}
    \big\|
        \observationn{n} 
        \!-\!
        \amplitudeBarCP{n}
        \steerVecBar(\state{n})
    \big\|^2
    \Big)
    \nonumber
\end{align}
where 
\begin{align*}
    \big\|
        \observationn{n} 
        \!-\!
        \amplitudeBarCP{n}
        \steerVecBar_n
    \big\|^2
    &=
    \big\| 
    \observationn{n}
    \big\|^2
    -
    2
    \amplitudeBarCP{n}
    \Re\big( \observationn{n}^\herm\steerVecBar_n \big)
    +
    \amplitudeBarCP{n}^2
    \big\| 
    \steerVecBar_n 
    \big\|^2\,,
\end{align*}
with $\amplitudeBarCP{n}\!\geq\!0$, $\etan{n}\!>\!0$, and $\|\steerVecBar_n\|\!>\!0$.
We aim to obtain the parameters $\amplitudeBarHatCP{n}$ and $\etaHatBarCP{n}$ that maximize~\eqref{Seq:conditional-CP-likelihood} conditional on the \gls{mt} state $\state{n}$ for some observed data $\observationn{n}$.
First, we take the partial derivative w.r.t. the modulus
\begin{align}
    \frac{\partial}{\partial \amplitudeBarCP{n}}
    \log
    f_{\scriptscriptstyle\mathrm{I}}\big(\observationn{n}|\state{n},\amplitudeBarCP{n},\etan{n}\big) 
    &= 
    \frac{
    2
    \Re\big( \observationn{n}^\herm\steerVecBar_n \big)
    -
    2\amplitudeBarCP{n}
    \big\| 
    \steerVecBar_n 
    \big\|^2
    }
    {\etan{n}}\,,
\end{align}
which we set to zero to yield the stationary point
\begin{align}
    2
    \Re\big( \observationn{n}^\herm\steerVecBar_n \big)
    &= 
    2\amplitudeBarCP{n}
    \big\| 
    \steerVecBar_n 
    \big\|^2
    \nonumber
    \\
    \frac{
    \Re\big( \observationn{n}^\herm\steerVecBar_n \big)
    }{\big\| 
    \steerVecBar_n 
    \big\|^2}
    &= 
    \amplitudeBarCP{n}
\end{align}
and enforce the constraint $\amplitudeBarCP{n}\!\geq\!0$, which yields the \gls{ml} estimator
\begin{align}\label{Seq:amplitudeBarHatCP}
    \amplitudeBarHatCP{n}\big|\state{n}
    \!=\!
    \frac{
    \max \Big(
    0,
    \Re\big( \observationn{n}^\herm\steerVecBar_n \big)
    \Big)
    }{\big\| 
    \steerVecBar_n 
    \big\|^2}\,.
\end{align}
Likewise, we compute the partial derivative w.r.t. the noise variance
\begin{align*}
    \frac{\partial}{\partial \etan{n}}
    \log
    f_{\scriptscriptstyle\mathrm{I}}\big(\observationn{n}|\state{n},\amplitudeBarCP{n},\etan{n}\big) 
    &= 
    - \frac{J \Nz}{\etan{n}}
    +
    \frac{
    \big\|
        \observationn{n} 
        \!-\!
        \amplitudeBarCP{n}
        \steerVecBar(\state{n})
    \big\|^2
    }{\big( \etan{n} \big)^2}
\end{align*}
which we set to zero to yield the stationary point
\begin{align}
    J \Nz \frac{1}{\etan{n}}
    &=
    \big( \etan{n} \big)^{-2}
    \big\|
        \observationn{n} 
        \!-\!
        \amplitudeBarCP{n}
        \steerVecBar(\state{n})
    \big\|^2
    \nonumber
    \\
    \etan{n}
    &=
    \frac{
    \big\|
        \observationn{n} 
        \!-\!
        \amplitudeBarCP{n}
        \steerVecBar(\state{n})
    \big\|^2
    }{
        J \Nz
    }
    \nonumber
    \\
    \etan{n}
    &=
    \frac{
    \big\| 
    \observationn{n}
    \big\|^2
    }{J \Nz}
    -
    \frac{2
    \amplitudeBarCP{n}
    \Re\big( \observationn{n}^\herm\steerVecBar_n \big)
    -
    \amplitudeBarCP{n}^2
    \big\| 
    \steerVecBar_n 
    \big\|^2}
    {J \Nz}\,.
\end{align}
Substituting \eqref{Seq:amplitudeBarHatCP} then yields the \gls{ml} estimator
\begin{align}\label{Seq:etaHatBarCP}
    \etaHatBarCP{n} \big|\state{n}
    &= 
    \frac{
    \big\| 
    \observationn{n}
    \big\|^2
    }{J \Nz}
    -
    \frac{
    \Big(\!
        \max \big(
        0,
        \Re\big( \observationn{n}^\herm\steerVecBar_n )
        \big)
    \!\Big)^2
    }{
    J \Nz
    \big\| 
    \steerVecBar_n 
    \big\|^2}
    \nonumber \\
    &= 
    \frac{
    \big\| 
    \observationn{n}
    \big\|^2
    }{J \Nz}
    -
    \frac{
        (\amplitudeBarHatCP{n})^2
        \big\| 
        \steerVecBar_n 
        \big\|^2
    }
    {J \Nz}\,.
\end{align}
Notably, the derivation applies analogously to PA-wise processing, as in \estLabel{NC}, by replacing the stacked observations and steering vectors with their PA-local counterparts. 

\section{Particle-Based Representations of Beliefs}\label{Ssec:PR-beliefs}

We derive the \gls{pr} of the \gls{mt} state belief in \eqref{eq:belief-state} as
\begin{align}
    &\belief(\state{n}) 
    \propto \beta(\state{n}) \prod\nolimits_{\scriptscriptstyle j\in \setAnchors} \iota(\state{n};\observation{n}{j}) 
    \nn\\[-2pt]
    &\approx
    \!\sum\nolimits_{\scriptscriptstyle p=1}^{\scriptscriptstyle P}
    \!\!\weight{\beta,n}{p} \delta\big( \state{n}-\particle{\state{n}}{p} \big)
    \prod\nolimits_{\scriptscriptstyle j\in \setAnchors} 
    \widetilde{\iota}(\state{n};\observation{n}{j})
    \\[-4pt]
    &
    =
    \!\sum_{\scriptscriptstyle p=1}^{\scriptscriptstyle P}
    \underbrace{
        \weight{\beta,n}{p} 
        \prod_{\scriptscriptstyle j\in \setAnchors}\! 
        \mathcal{CN}\Big(\!\observation{n}{j};\!\muiota(\particle{\state{n}}{p}),\Ciota(\particle{\state{n}}{p})\!\Big)
    }_{= \weightt{\bm{x},n}{p}}
    \delta\big(\state{n}\!-\!\particle{\state{n}}{p} \big)
    \nn
\end{align}
where the last line 
uses the property 
$f(x)\delta(x-a) \!=\! f(a)\delta(x-a)$ of the Dirac delta function~\cite[p.\,60]{Dirac1947QuantumMechanics}.
We derive the \gls{pr} of the belief about the amplitude state in \eqref{eq:belief-PFphi} as
\begingroup             
\begin{align}
    &\belief(\PFphi{s}{n}{j}) \propto \alpha(\PFphi{s}{n}{j}) 
    \prod\nolimits_{\scriptscriptstyle j\in \setAnchors}
    \kappa(\PFphi{s}{n}{j};\observation{n}{j}) 
    \nn\\[-1pt]
    &\approx
    \!\sum\nolimits_{\scriptscriptstyle p=1}^{\scriptscriptstyle P}
    \!\!\weight{\alpha,n}{p}
    \delta\big( \PFphi{s}{n}{j}-\particle{\PFphi{s}{n}{j,p}}{p} \big)
    \prod\nolimits_{\scriptscriptstyle j\in \setAnchors}
    \widetilde{\kappa}(\PFphi{s}{n}{j};\observation{n}{j}) 
    \nn\\[-3pt]
    &
    =
    \!\sum_{\scriptscriptstyle p=1}^{\scriptscriptstyle P}
    \underbrace{
        \weight{\alpha,n}{p}
        \prod\nolimits_{\scriptscriptstyle j\in \setAnchors}
        \mathcal{CN}\Big(\!\observation{n}{j};\!
        \mukappa(\particle{\PFphi{s}{n}{j,p}}{p}),
        \Ckappa(\particle{\PFphi{s}{n}{j,p}}{p})
        \!\Big)
    }_{= \weightt{\bm{\phi},n}{p}}
    \nn\\[-2pt]
    &~\hspace{4.7cm}~\times\delta\big(\PFphi{s}{n}{j}\!-\!\particle{\PFphi{s}{n}{j,p}}{p} \big)\,,
\end{align}
\endgroup
and the \gls{pr} of the noise variance belief in \eqref{eq:belief-eta} as
\begin{align}
    &\belief(\etan{n}^{\scriptscriptstyle(j)}) \propto \xi(\etan{n}^{\scriptscriptstyle(j)}) \nu(\etan{n}^{\scriptscriptstyle(j)};\observation{n}{j})
    \nn\\[-2pt]
    &\approx
    \!\sum\nolimits_{\scriptscriptstyle p=1}^{\scriptscriptstyle P}
    \!\!\weight{\xi,n}{j,p} \delta\big( \etan{n}^{\scriptscriptstyle(j)}-\particle{\etan{n}}{j,p} \big)
    \widetilde{\nu}(\etan{n}^{\scriptscriptstyle(j)};\observation{n}{j})
    \\[-3pt]
    &
    =
    \!\sum_{\scriptscriptstyle p=1}^{\scriptscriptstyle P}
    \underbrace{
        \weight{\xi,n}{j,p} 
        \mathcal{CN}\Big(\!\observation{n}{j};\!\munu,\Cnu(\particle{\etan{n}}{j,p}) \!\Big)
    }_{= \weightt{\eta,n}{j,p}}
    \delta\big(\etan{n}^{\scriptscriptstyle(j)}\!-\!\particle{\etan{n}}{j,p} \big).
    \nn
\end{align}

\section{Efficient Implementation}\label{Ssec:efficient-implementation}

For the Type-I estimators, the \gls{mt} state update messages are
obtained by evaluating the profile likelihoods in
\eqref{eq:likelihood-Type-I-NC}, \eqref{eq:likelihood-Type-I-C}, and
\eqref{eq:likelihood-Type-I-CP} at the \gls{mt} state particles $\particle{\state{n}}{p}$, which have complexity $\mathcal{O}(P\,J\,\Nz)$.
Therefore, we only derive efficient particle-based implementations of our Type-II filters in the following.

Particle-based evaluations of approximate Gaussian messages
$\widetilde{\iota}(\particle{\state{n}}{p};\observation{n}{j})$,
$\widetilde{\kappa}(\particle{\PFphi{s}{n}{j,p}}{p};\observation{n}{j})$, and
$\widetilde{\nu}(\particle{\etan{n}}{j,p};\observation{n}{j})$ demand the computation of determinants and inverses of 
$\Ciota(\particle{\state{n}}{p})$,
$\Ckappa(\particle{\PFphi{s}{n}{j,p}}{p})$, 
and
$\Cnu(\particle{\etan{n}}{j,p})$.
Computing these determinants and inverses has complexity $\mathcal{O}(\Nz^3)$.
We seek an approximation that reduces the computation of inverses and determinants of matrices of dimension $\Nz$ (i.e., the observation length) to inner products of $\Nz$-vectors and scalar operations, resulting in complexity $\mathcal{O}(\Nz)$.
To this end, we use the alternative approximation from \cite{LiaLeiMey:TSP2025-SuppDoc}.
The idea is to find low-rank approximations of the covariance matrices \eqref{eq:Ciota}--\eqref{eq:Ckappa}.
For brevity, we present only the \glspl{pr} of update message covariance matrices introduced at the beginning of the following subsections. 
Detailed derivations are provided in the Extended Derivations~\cite{ThisPaperDerivations}.
As shorthand notation, we define the noise variance prior mean $\etaXit{n}\!\coloneqq\!\sum_{\scriptscriptstyle p=1}^{\scriptscriptstyle P} \weight{\xi,n}{j,p} \particle{\etan{n}}{j,p}$.

\subsection{Amplitude State Update Message}
A \gls{pr} of the covariance matrix parameterizing the amplitude state update message is
\begin{align}
    \Ckappa\!(\particle{\PFphi{s}{n}{j,p}}{p})
    \!\approx\! 
    \eye{\Nz} \etaXit{n}
    \!+\!
    \pnjp
    {\pnjp}^\herm
    \!\eqqcolon\!
    \Ckappap
\end{align}
where 
$\pnjp
\!\coloneqq\!
\sqrt{\particle{\PFgamma{s}{n}{j,p}}{p}} \steerVecx{s}{n}{j,p}$
and 
$\steerVecx{s}{n}{j,p} \!\coloneqq\! \steerVec{j}(\particle{\state{n}}{p})$.
The 
matrix inversion lemma~\cite[eq.\,(156)]{Cookbook} 
yields
\begin{align}
    {\Ckappap}^{\!-1}
    =
    \frac{1}{\etaXit{n}}
    \Bigg(
        \eye{\Nz}
        -
        \frac{
           \pnjp
           {\pnjp}^\herm
        }{
            \etaXit{n} 
            + 
            \big\|
            \pnjp
            \big\|^2
        }
    \Bigg)
    \,,
\end{align}
and the generalized determinant lemma~\cite[p.\,420]{harville2008matrixAlgebra}
gives
\begin{align}
    \det (\Ckappap)
    =
    \big(\etaXit{n}\big)^{\Nz}
    \Bigg(
        1
        \!+\!
        \frac{
            \big\|
            \pnjp
            \big\|^2
        }{
           \etaXit{n} 
        }
    \Bigg)
    .
\end{align}
Let $\ekappa{(j,p)}\!\coloneqq\!\observation{n}{j}\!-\!\particle{\PFmu{s}{n}{j,p}}{p}\steerVecx{s}{n}{j,p}$, where 
$\mukappa(\particle{\PFphi{s}{n}{j,p}}{p})\!=\! \musnj{2}(\particle{\PFphi{s}{n}{j,p}}{p})\approx\particle{\PFmu{s}{n}{j,p}}{p}\steerVecx{s}{n}{j,p}$.
The approximate complex Gaussian amplitude state update message evaluated at particle $\particle{\PFphi{s}{n}{j,p}}{p}$ is
\begin{align}
    &\widetilde{\kappa}(\particle{\PFphi{s}{n}{j,p}}{p};\observation{n}{j}) 
    =
    \mathcal{CN}\big(\observation{n}{j};\mukappa\!(\particle{\PFphi{s}{n}{j,p}}{p}),\Ckappa\!(\particle{\PFphi{s}{n}{j,p}}{p})\big)
    \\
    &\!\approx\!
    \frac{
        \exp\big(
            -
            {\ekappa{(j,p)}}^\herm
                {\Ckappap}^{\!-1}
            \ekappa{(j,p)}
        \big)
    }{
        \pi^{\Nz}
        \det(\Ckappap)
    }
    \nn\\
    &\!=\!
    \frac{
    \exp\!
    \bigg(
            \frac{
                \particle{\PFgamma{s}{n}{j,p}}{p}
                \left|
                    {\steerVecx{s}{n}{j,p}}^\herm
                    \!
                    \big(
                        \observation{n}{j} 
                        -
                        \particle{\PFmu{s}{n}{j,p}}{p}
                        \steerVecx{s}{n}{j,p}
                    \big)
                \right|^2
            }{
                \etaXit{n}
                \big(
                    \etaXit{n}
                    +
                    \big\|
                        \pnjp
                    \big\|^2
                \big)
            }
            \!-\!
            \frac{
            \big\|
                \observation{n}{j} 
                -
                \particle{\PFmu{s}{n}{j,p}}{p}
                \steerVecx{s}{n}{j,p}
            \big\|^2
            }{
                \etaXit{n}
            }
        \!
    \bigg)
    }
    {
    \pi^{\Nz}
    \big(
        \etaXit{n}
    \big)^{\Nz}
    \bigg(
        1
        +
        \frac{
            \particle{\PFgamma{s}{n}{j,p}}{p}
            \left\|
                \steerVecx{s}{n}{j,p}
            \right\|^2
        }{
            \etaXit{n}
        }
    \bigg)
    }
    \nn
\end{align}
which involves inner products that are evaluated with complexity $\mathcal{O}(\Nz)$ per parallelized \gls{pa}-particle pair $(j,p)$.
This brings the total computational complexity of computing this update message to $\mathcal{O}(P\,J\,\Nz)$.

\subsection{MT State Update Message}
As for the amplitude state update message, we find a \gls{pr} of the covariance matrix parameterizing the \gls{mt} state update message as 
\begin{align}
    \Ciota(\particle{\state{n}}{p}) 
    \!\approx\! 
    \eye{\Nz} \etaXit{n}
    \!+\!
    \pnjp
    {\pnjp}^\herm
    \!\eqqcolon\!\Ciotap
\end{align}
Again, the matrix inversion lemma
gives
\begin{align}
    {\Ciotap}^{\!-1}=
    \frac{1}{\etaXit{n}}
    \Bigg(
        \eye{\Nz}
        -
        \frac{
           \pnjp
           {\pnjp}^\herm
        }{
            \etaXit{n} 
            + 
            \big\|
            \pnjp
            \big\|^2
        }
    \Bigg)
    \,,
\end{align}
and the generalized determinant lemma
gives
\begin{align}
    \det (\Ciotap)\!=\!
    \big(\etaXit{n}\big)^{\Nz}
    \Bigg(
        1
        \!+\!
        \frac{
            \big\|
            \pnjp
            \big\|^2
        }{
           \etaXit{n} 
        }
    \Bigg)
    .
\end{align}
Again, let $\eiota{(j,p)}\!\coloneqq\!\observation{n}{j}\!-\!\particle{\PFmu{s}{n}{j,p}}{p}\steerVecx{s}{n}{j,p}$, where 
$\muiota(\particle{\state{n}}{p})\!=\! \musnj{1}(\particle{\state{n}}{p})\approx\particle{\PFmu{s}{n}{j,p}}{p}\steerVecx{s}{n}{j,p}$.
The approximate complex Gaussian \gls{mt} state update message evaluated at particle $\particle{\state{n}}{p}$ is
\begin{align}
    &\widetilde{\iota}(\particle{\state{n}}{p};\observation{n}{j}) 
    =
    \mathcal{CN}\big(\observation{n}{j};\muiota\!(\particle{\state{n}}{p}),\Ciota\!(\particle{\state{n}}{p})\big)
    \\
    &\!\approx\!
    \frac{
        \exp\big(
            -
            {\eiota{(j,p)}}^\herm
                {\Ciotap}^{\!-1}
            \eiota{(j,p)}
        \big)
    }{
        \pi^{\Nz}
        \det(\Ciotap)
    }
    \nn\\
    &\!=\!
    \frac{
    \exp\!
    \bigg(
            \frac{
                \particle{\PFgamma{s}{n}{j,p}}{p}
                \left|
                    {\steerVecx{s}{n}{j,p}}^\herm
                    \!
                    \big(
                        \observation{n}{j} 
                        -
                        \particle{\PFmu{s}{n}{j,p}}{p}
                        \steerVecx{s}{n}{j,p}
                    \big)
                \right|^2
            }{
                \etaXit{n}
                \big(
                    \etaXit{n}
                    +
                    \big\|
                        \pnjp
                    \big\|^2
                \big)
            }
            \!-\!
            \frac{
            \big\|
                \observation{n}{j} 
                -
                \particle{\PFmu{s}{n}{j,p}}{p}
                \steerVecx{s}{n}{j,p}
            \big\|^2
            }{
                \etaXit{n}
            }
        \!
    \bigg)
    }
    {
    \pi^{\Nz}
    \big(
        \etaXit{n}
    \big)^{\Nz}
    \bigg(
        1
        +
        \frac{
            \particle{\PFgamma{s}{n}{j,p}}{p}
            \left\|
                \steerVecx{s}{n}{j,p}
            \right\|^2
        }{
            \etaXit{n}
        }
    \bigg)
    }
    \,.
    \nn
\end{align}
This expression involves inner products that are again evaluated with complexity $\mathcal{O}(\Nz)$ per parallelized \gls{pa}-particle pair $(j,p)$,
leaving the total computational complexity of computing this update message at $\mathcal{O}(P\,J\,\Nz)$.

\subsection{Noise Variance Update Message}
We find a \gls{pr} of the covariance matrix parameterizing the noise-variance update message as
\begin{align}
    \Cnu(\particle{\etan{n}}{j,p}) 
    \approx 
    \eye{\Nz}\particle{\etan{n}}{j,p} 
    \!+\!
    \mnj
    {\mnj}^\herm
    \!\eqqcolon\!
    \Cnup\,,
\end{align}
where 
$\mnj\!\coloneqq\!
\sum_{\scriptscriptstyle p=1}^{\scriptscriptstyle P}
\weight{\beta,n}{p}
\sqrt{\particle{\PFgamma{s}{n}{j}}{p}}
\steerVecx{s}{n}{j,p}$.
The matrix inversion lemma yields
\begin{align}
    {\Cnup}^{\!-1}=
    \frac{1}{\particle{\etan{n}}{j,p}}
    \Bigg(
        \eye{\Nz}
        -
        \frac{
           \mnj
           {\mnj}^\herm
        }{
            \particle{\etan{n}}{j,p} 
            + 
            \big\|
            \mnj
            \big\|^2
        }
    \Bigg)
    \,,
\end{align}
and the generalized determinant lemma
gives
\begin{align}
    \det (\Cnup)\!=\!
    \big(\particle{\etan{n}}{j,p}\big)^{\Nz}
    \Bigg(
        1
        \!+\!
        \frac{
            \big\|
            \mnj
            \big\|^2
        }{
           \particle{\etan{n}}{j,p} 
        }
    \Bigg)
    .
\end{align}
Let $\enu{(j)}\!\coloneqq\!\observation{n}{j}\!-\!\munu$, where 
$\munu\!=\! 
\musnj{3}
\approx
\sum_{\scriptscriptstyle p=1}^{\scriptscriptstyle P}
\weight{\beta,n}{p}
\particle{\PFmu{s}{n}{j,p}}{p}
\steerVecx{s}{n}{j,p}$.
The approximate complex Gaussian noise variance update message evaluated at particle $\particle{\etan{n}}{j,p}$ is
\begin{align}
    \widetilde{\nu}(\particle{\etan{n}}{j,p};\observation{n}{j}) 
    &=
    \mathcal{CN}\big(\observation{n}{j};\munu,\Cnu\!(\particle{\etan{n}}{j,p})\big)
    \\
    &\approx\!
    \frac{
        \exp\big(
            -
            {\enu{(j)}}^\herm
                {\Cnup}^{\!-1}
            \enu{(j)}
        \big)
    }{
        \pi^{\Nz}
        \det(\Cnup)
    }
    \nn\\
    &=\!
    \frac{
    \exp\!
    \bigg(
            \frac{
                \left|
                    {\enu{(j)}}^\herm
                    \!
                    \mnj
                \right|^2
            }{
                \particle{\etan{n}}{j,p}
                \big(
                    \particle{\etan{n}}{j,p}
                    +
                    \big\|
                        \mnj
                    \big\|^2
                \big)
            }
            \!-\!
            \frac{
            \big\|
                \enu{(j)}
            \big\|^2
            }{
                \particle{\etan{n}}{j,p}
            }
        \!
    \bigg)
    }
    {
    \pi^{\Nz}
    \big(
        \particle{\etan{n}}{j,p}
    \big)^{\Nz}
    \bigg(
        1
        +
        \frac{
            \left\|
                \mnj
            \right\|^2
        }{
            \particle{\etan{n}}{j,p}
        }
    \bigg)
    }
    \,. 
    \nn
\end{align}
While the $\mathcal{O}(\Nz)$ inner products in this expression are particle-independent, computing $\mnj$ and $\musnj{3}$ requires $\mathcal{O}(P\,J\,\Nz)$, which also determines the total complexity of the update message.
Again, operations across $J$ \glspl{pa} and $P$ particles are parallelized in our \gls{gpu}-accelerated implementation.

\section{Posterior Cram\'er--Rao Lower Bound}\label{Ssec:PCRLB}  
The (i) noncoherent, (ii) coherent, and (iii) carrier-phase-based bounds derived in the following are \glspl{pcrlb} for the statistical models underlying the respective Type-I estimators. 
For the Type-II estimators, they merely serve as performance benchmarks.
The derivations in Sections~\ref{Ssec:FIMch} and~\ref{Ssec:Jacobian} follow the approach of~\cite{Deutschmann26TSP}, specialized to \gls{los} channels. 

For completeness, a closed-form solution for (i) the noncoherent \gls{pcrlb} of the Type-II model can be found under the respective statistical model of the \gls{zm} method, however, its derivation is involved and therefore omitted from this Supplementary Material.
As mentioned in Sec.\,\ref{sec:PCRLB}, throughout the \gls{pcrlb} derivation, the amplitudes 
$\widetilde{{\varrho}}_{\!\scriptscriptstyle n}^{\scriptscriptstyle(j)}
\coloneqq
\varrho_{\!\scriptscriptstyle n}^{\scriptscriptstyle(j)}
g_{\scriptscriptstyle n}^{\scriptscriptstyle (j)}
\exp (\!-\mathrm{j}\frac{2\pi}{\lightspeed}\fc \lVert \rangep{k}{n}{j}\rVert )$
absorb the path loss $g_{\scriptscriptstyle n}^{\scriptscriptstyle (j)}\!=\!\frac{\lambda}{4 \pi \lVert \acute{\V{r}} \rVert}$ and the carrier-phase term $\exp (\!-\mathrm{j}\frac{2\pi}{\lightspeed}\fc \lVert \rangep{k}{n}{j}\rVert)$ from~\eqref{eq:unit-modulus-array-response} instead of the steering vectors $\steerVecx{s}{n}{j}$. 
The distance-dependent phase is $\widetilde{\unslant[-.25]{\varphi}}_{\!\scriptscriptstyle n}^{\scriptscriptstyle(j)}\!=\!\unslant[-.25]{\varphi}_{\!\scriptscriptstyle n}^{\scriptscriptstyle(j)} \!-\!\frac{2\pi}{\lightspeed}\fc \lVert \rangep{k}{n}{j}\rVert$
leading to the polar form of the amplitudes 
$\widetilde{\unslant[-.25]{\varrho}}_{\!\scriptscriptstyle n}^{\scriptscriptstyle(j)}=
\rv{a}_{\!\scriptscriptstyle n}^{\scriptscriptstyle(j)}
\exp(\mathrm{j} \unslant[-.25]{\widetilde{\varphi}}_{\!\scriptscriptstyle n}^{\scriptscriptstyle(j)} )$ which enters the local channel parameter vector.

{\slshape Approximations.} 
For the state transition of amplitudes in global parameters and the noise variance $\RVetan{n}$, we introduce pseudo state-transition \glspl{pdf} 
$f(\modulivec{n}{}|\modulivec{n-1}{})=\mathcal{N}(\modulivec{n}{};\mathbf{0},\sigma_a^2 \eye{J})$,
$f(\phasevec{n}{}|\phasevec{n-1}{})=\mathcal{N}(\phasevec{n}{};\mathbf{0},\sigma_\varphi^2 \eye{\dimPhase})$,
and
$f(\etan{n}|\etan{n-1})=\mathcal{N}(\etan{n};0,\sigma_\eta^2)$, with $\{\sigma_a^2,\sigma_\varphi^2,\sigma_\eta^2\}$ sufficiently large that the information about $\{{\RVphasevec{n}{}}, {\RVmodulivec{n}{}}, \RVetan{n}\}$ injected through $\FIMstep{n}{n\!-\!1}$ into $\FIMstep{n}{n}$ is negligible compared to the information injected through $\FIMglobal{n}$, while still permitting the computation of $\FIMstep{n}{n\!-\!1}$ and the inverses involved in its computation~\eqref{eq:priorFIM} without rank deficiency.
Furthermore, we make the simplifying approximation that the amplitude moduli 
${\rv{a}_{\!\scriptscriptstyle n}^{\scriptscriptstyle(j)}}$
do not contribute information about the parameters of interest
$\RVstate{n}$ and hence the respective Jacobians are zero vectors, i.e., 
$\jacobP{\modulus} = \frac{\partial {{a}_{\!\scriptscriptstyle n}^{\scriptscriptstyle(j)}}}{\partial \pos{n}} \approx 
\mathbf{0}$.
This approximation is valid because the \textit{other} local channel parameters
$\big\{
\RVelVecx{n}{j}, 
\RVazVecx{n}{j},
\RVdelayVecx{n}{j}, 
\unslant[-.25]{\widetilde{\varphi}}_{\!\scriptscriptstyle n}^{\scriptscriptstyle(j)}
\big\}$
generally contribute substantially more information about the parameters of interest, i.e., the \gls{mt} state.

The \gls{pcrlb} matrix $\PCRLB$ is a lower bound on the 
\gls{mse} matrix 
$\mathbb{E}_{\,\RVetaglobalSmall{0:n}, \RVobservationn{1:n}} 
\!\big((\RVetaglobalHat{n}\!-\RVetaglobal{n} ) (\RVetaglobalHat{n}\!-\RVetaglobal{n} )^\trp\big)$
of any \textit{Bayesian} estimator~\cite[eq.\,(29)]{VanTrees2007PCRLB}, implying that \textit{both} the observations $\RVobservationn{n}$ and the state $\RVetaglobal{n}$ are \glspl{rv}. 
For the \gls{pcrlb} $\PCRLB$ to lower-bound the estimation error, the \gls{mt} positions $\pos{n}$ cannot be deterministic unknowns as in a ``classic'' estimation problem, but they must become \glspl{rv} $\RVpos{n}$ not only in the estimator, but also in the data-generating process~\cite{Fritsche16ParametricCRLB}.

{\slshape Synthetic Data Generation:} 
That is, for each \gls{mc} run, at each time $n\!\in\!\{1\ist\dots\ist N\}$, an 
\gls{mt} state $\state{n} \!\coloneqq\! [\pos{n}^\trp \iist \vel{n}^\trp]^\trp$
realization is drawn from its state-transition \gls{pdf} $f(\state{n}|\state{n-1})=\mathcal{N}(\state{n};\transitionmatrix_{\text{\tiny a}}\state{n-1},\processNoiseCov_{\text{\tiny a}})$, thereby
leading to a different \gls{ncv} trajectory $\big\{\pos{n}\big\}_{\scriptscriptstyle n=1}^{\scriptscriptstyle N}$, as shown in Fig.\,\ref{fig:scenario}.

We draw $\state{0}$ at time $n\!=\!0$ from 
$
f(\state{0})
\!=\!
\mathcal{N}(\state{0};\mathbf{x}_{\scriptscriptstyle 0},\mathbf{C}_{\scriptscriptstyle 0})$.
This Gaussian prior yields the initial \gls{mt} state information
submatrix $\big[\FIMstep{0}{0}\big]_{\scriptscriptstyle1:6,1:6}\!=\!\mathbf{C}_{\scriptscriptstyle 0}^{\scriptscriptstyle-1}$.
Propagation through the state-transition model gives
$
f(\state{1})
=
\mathcal{N}\!\left(
\state{1};
\transitionmatrix_{\text{\tiny a}}\mathbf{x}_{\scriptscriptstyle 0},
\transitionmatrix_{\text{\tiny a}}\mathbf{C}_{\scriptscriptstyle 0}
\transitionmatrix_{\text{\tiny a}}^\trp
+\processNoiseCov_{\text{\tiny a}}
\right)
$,
which defines the \gls{mt} state information submatrix
$\big[\FIMstep{1}{0}\big]_{\scriptscriptstyle1:6,1:6}\!=\!\big(\transitionmatrix_{\text{\tiny a}}\mathbf{C}_{\scriptscriptstyle 0}
\transitionmatrix_{\text{\tiny a}}^\trp
+\processNoiseCov_{\text{\tiny a}}\big)^{\scriptscriptstyle-1}$ 
of the initial predicted information matrix $\FIMstep{1}{0}$ serving as
the starting point of the recursion in~\eqref{eq:posteriorCRLB}--\eqref{eq:priorFIM}.

We use the \gls{mc} analysis both to compute the \gls{mse} matrices of our estimators and to numerically implement the expectation $\FIMglobal{n} \!=\! \mathbb{E}_{\RVetaglobalSmall{n}}\!(\FIMclassic{n})$ in the Bayesian snapshot information matrix from Sec.\,\ref{sec:FIMg} by means of \gls{mc} integration (cf.\,\cite{Hernandez02PCRLB}). 
Under the above approximations, the linear Gaussian state-transition \gls{pdf} factorizes as
$f(\etaglobal{n}|\etaglobal{n-1})\!=\!
f(\state{n}|\state{n-1})
f(\phasevec{n}{}|\phasevec{n-1}{})
f(\modulivec{n}{}|\modulivec{n-1}{})
f(\etan{n}|\etan{n-1})$
such that
$f(\etaglobal{n}|\etaglobal{n-1})\!=\!\mathcal{N}(\etaglobal{n};\transitionmatrix \etaglobal{n-1},\processNoiseCov)$ is parameterized by the state-transition matrix and the process-noise covariance matrix
\begin{align*}
    \transitionmatrix \!=\! 
    \begin{bmatrixs}
        \transitionmatrix_{\text{\tiny a}}     & \mathbf{0} & \mathbf{0} & 0 \\
        \mathbf{0}  & \mathbf{0}       & \mathbf{0} & 0 \\
        \mathbf{0}  & \mathbf{0} & \mathbf{0}   & 0      \\
        {0}  & {0} & {0}   & 0
    \end{bmatrixs}
\text{,} \,
    \processNoiseCov \!=\! 
    \begin{bmatrixs}
        \processNoiseCov_{\text{\tiny a}}      & \mathbf{0} & \mathbf{0} & 0 \\
        \mathbf{0}  & \sigma_\varphi^2 \eye{\dimPhase}       & \mathbf{0} & 0 \\
        \mathbf{0}  & \mathbf{0} & \sigma_a^2 \eye{J}   & 0      \\
        {0}  & {0} & {0}   & \sigma_\eta^2 
    \end{bmatrixs} \,
\end{align*}
respectively.
Under this (approximate) linear Gaussian state-transition \gls{pdf}, the predicted information matrix from \eqref{eq:priorFIM} takes the form
$\FIMstep{n}{n\!-\!1} = \big(\transitionmatrix \, \FIMstep{n\!-\!1}{n\!-\!1}^{-1} \, \transitionmatrix^\trp + \processNoiseCov\big)^{-1}$,
as shown by Hernandez et al.\,\cite[eq.\,(16)]{Hernandez02PCRLB}.

With the prior information matrix defined, the remaining quantities needed to complete the recursion in~\eqref{eq:posteriorCRLB}--\eqref{eq:priorFIM} are the entries of the local channel \glspl{fim} $\FIMch{n}{j}$ from~\eqref{eq:FIMclassic} and the Jacobian matrices $\jacobgn{j}= \nicefrac{\partial {\etach{n}{j}}^{\!\!\!\trp}}{\partial \etaglobal{n}}$ from~\eqref{eq:FIMclassic} mapping from \gls{pa}-local channel parameters $\etach{n}{j}$ to global (infrastructure) parameters $\etaglobal{n}$.

Recall the definitions of the delay, elevation, and azimuth
\begin{align}\label{eq:delay}
    \delay(\acute{\V{r}}) &= \nicefrac{\lVert \acute{\V{r}} \rVert}{\lightspeed} \,,
    \\
    \label{eq:elevation}
    \elevation(\acute{\V{r}}) &= \arccos (\rz/ \lVert \acute{\V{r}} \rVert)\,,
    \\
    \label{eq:azimuth}
    \azimuth(\acute{\V{r}}) &= \arctantwo (\ry,\rx)\,,
\end{align}
representing local spherical \gls{pa}-coordinates parameterized by 
\begin{align}\label{eq:rangepSupp}
    \rangep{k}{n}{j}(\state{n},\posPA{j},\rotM{j}) 
    &\coloneqq \rotM{j}^{-1} \range{k}{n}{j} 
    = \rotM{j}^{\trp} \range{k}{n}{j} \,,
\\
\label{eq:range}
    \range{k}{n}{j}(\state{n},\posPA{j}) &\coloneqq \pos{n} - \posPA{j}\,.
\end{align}
Here, the vector $\range{k}{n}{j}$ points from a \gls{pa} to the \gls{mt} in global coordinates, and $\rangep{k}{n}{j}$ points from the \gls{pa} to the \gls{mt} in local \gls{pa}-coordinates.

\subsection{Derivation of Local Channel FIM Terms}\label{Ssec:FIMch}  
The \gls{pa}-local channel parameter vector 
$\RVetach{n}{j}\!\coloneqq\!\big[
{\RVelVecx{n}{j}}\iist 
{\RVazVecx{n}{j}} \iist  
{\RVdelayVecx{n}{j}} \iist 
\unslant[-.25]{\widetilde{\varphi}}_{\!\scriptscriptstyle n}^{\scriptscriptstyle(j)} \iist 
{\rv{a}_{\!\scriptscriptstyle n}^{\scriptscriptstyle(j)}} \iist
\RVetan{n}\big]^\trp \! \in \! \realsetone{\dimLocal}$
contains the elevation and azimuth angles, delay, amplitude phase and modulus, and noise variance.
Here, 
$\delayx{k}{n}{j}\coloneqq \tau(\rangep{k}{n}{j})$ is defined through \eqref{eq:delay},
$\elx{k}{n}{j}\coloneqq \elevation(\rangep{k}{n}{j})$ is defined through \eqref{eq:elevation}, and
$\azx{k}{n}{j}\coloneqq \azimuth(\rangep{k}{n}{j})$ is defined through \eqref{eq:azimuth}.
For the derivation of the channel \gls{fim} terms in \eqref{eq:FIMch}, we assume the Type-I likelihood with spatially and temporally uncorrelated, circular \gls{awgn} $\RVnoise{n}{j}$ leading to
$f(\observation{n}{j}|\etach{n}{j})\!\coloneqq\!\CN\big(\observation{n}{j};  \widetilde{\varrho}_{\!\scriptscriptstyle n}^{\scriptscriptstyle(j)} \signalatomnl{k}{n}{j},\etan{n} \eye{\Nz}\big)$.
Evaluation of~\eqref{eq:FIMch} gives the local channel \gls{fim} entries listed in Table~\ref{Stab:channel-fim}.
\begin{table}[t]
    \setlength{\tabcolsep}{3.5pt} 
    \renewcommand{\arraystretch}{1.05} 
    \centering \small
    \caption
    {Individual entries $[\FIMch{n}{j}]_{\scriptstyle \grave{\imath},\acute{\imath}}$ of the channel \gls{fim} in~\eqref{eq:FIMch}.
    The listed row and column indices follow the parameter ordering of $\etach{n}{j}$ and describe the upper triangular matrix of $\FIMch{n}{j}$.
        }%
\label{Stab:channel-fim}
    \begin{tabularx}{1\columnwidth}{@{}cc|c|c|c}
        \toprule
        $[\etach{n}{j}]_{\scriptscriptstyle \grave{\imath}}$ & $[\etach{n}{j}]_{\scriptscriptstyle \acute{\imath}}$
        & Term & \scalebox{0.9}{Row\,$\grave{\imath}$} & \scalebox{0.9}{Column\,$\acute{\imath}$}  \\
        \midrule 
        $\elx{k}{n}{j}$&$\elx{k'\!}{n}{j}$ &
            $\frac{2}{\etan{n}}\Re\big(
                     \left.\complexamplitude{\ncomponent}\right.^{\!\ast}\complexamplitude{\ncomponent'\!}
                    {\atomelnl{(j) \herm}{n}}
                    \atomelnl{(j)}{n}
                \big)$  &
            $1$ &  
            $1 $ \\ 
        $\elx{k}{n}{j}$&$\azx{k'\!}{n}{j}$ &
            $\frac{2}{\etan{n}}\Re\big(
                    \left.\complexamplitude{\ncomponent}\right.^{\!\ast}\complexamplitude{\ncomponent'\!} 
                    {\atomelnl{(j) \herm}{n}}
                    \atomaznl{(j)}{n}
                \big)$  &
            $1$ &  
            $2$ \\ 
        $\elx{k}{n}{j}$&$\delayx{k'\!}{n}{j}$ & 
            $\frac{2}{\etan{n}}\Re\big(
                    \left.\complexamplitude{\ncomponent}\right.^{\!\ast}\complexamplitude{\ncomponent'\!} 
                    {\atomelnl{(j) \herm}{n}} 
                    \atomdelaynl{(j)}{n}
                \big)$  &
            $1$ &  
            $3$ \\ 
        $\elx{k}{n}{j}$&$\phasex{k'\!}{n}{j}$ & 
            $\frac{2}{\etan{n}}\Re\big(
                    \mathrm{j} \, \left.\complexamplitude{\ncomponent}\right.^{\!\ast}\complexamplitude{\ncomponent'\!}
                    {\atomelnl{(j) \herm}{n}}
                    {\signalatomnl{k'\!}{n}{j}}
                \big)$  &
            $1$ &  
            $4$ \\ 
        $\elx{k}{n}{j}$&$\modulusx{k'\!}{n}{j}$ & 
            $\frac{2}{\etan{n}}\Re\big(
                    \left.\complexamplitude{\ncomponent}\right.^{\!\ast} e^{\mathrm{j}\phasex{\ncomponent'\!}{n}{j}} 
                    {\atomelnl{(j) \herm}{n}}
                    {\signalatomnl{k'\!}{n}{j}}
                \big)$  &
            $1$ &  
            $5$ \\ 
        $\elx{k}{n}{j}$&$\etan{n}$ & 
            $0$  &
            $1$ &  
            $6$ \\ 
   \midrule
        $\azx{k}{n}{j}$&$\azx{k'\!}{n}{j}$ & 
            $\frac{2}{\etan{n}}\Re\big(
                    \left.\complexamplitude{\ncomponent}\right.^{\!\ast} \complexamplitude{\ncomponent'\!} 
                    {\atomaznl{(j) \herm}{n}} 
                    \atomaznl{(j)}{n}
                \big)$  &
            $2$ &  
            $2$ \\ 
        $\azx{k}{n}{j}$&$\delayx{k'\!}{n}{j}$ & 
            $\frac{2}{\etan{n}}\Re\big(
                    \left.\complexamplitude{\ncomponent}\right.^{\!\ast} \complexamplitude{\ncomponent'\!} 
                    {\atomaznl{(j) \herm}{n}}
                    \atomdelaynl{(j)}{n}
                \big)$  &
            $2$ &  
            $3$ \\ 
        $\azx{k}{n}{j}$&$\phasex{k'\!}{n}{j}$ & 
            $\frac{2}{\etan{n}}\Re\big(
                    \mathrm{j} \, \left.\complexamplitude{\ncomponent}\right.^{\!\ast} \complexamplitude{\ncomponent'\!} 
                    {\atomaznl{(j) \herm}{n}}
                    {\signalatomnl{k'\!}{n}{j}}
                \big)$  &
            $2$ &  
            $4$ \\ 
        $\azx{k}{n}{j}$&$\modulusx{k'\!}{n}{j}$ & 
            $\frac{2}{\etan{n}}\Re\big(
                    \left.\complexamplitude{\ncomponent}\right.^{\!\ast} e^{\mathrm{j}\phasex{\ncomponent'\!}{n}{j}} 
                    {\atomaznl{(j) \herm}{n}}
                    {\signalatomnl{k'\!}{n}{j}}
                \big)$  &
            $2$ &  
            $5$ \\ 
        $\azx{k}{n}{j}$&$\etan{n}$ & 
            $0$  &
            $2$ &  
            $6$ \\ 
   \midrule
        $\delayx{k}{n}{j}$&$\delayx{k'\!}{n}{j}$ & 
            $\frac{2}{\etan{n}}\Re\big(
                    \left.\complexamplitude{\ncomponent}\right.^{\!\ast} \complexamplitude{\ncomponent'\!} 
                    {\atomdelaynl{(j) \herm}{n}}
                    \atomdelaynl{(j)}{n}
                \big)$  &
            $3$ &  
            $3$ \\ 
        $\delayx{k}{n}{j}$&$\phasex{k'\!}{n}{j}$ & 
            $\frac{2}{\etan{n}}\Re\big(
                    \mathrm{j} \, \left.\complexamplitude{\ncomponent}\right.^{\!\ast} \complexamplitude{\ncomponent'\!} 
                    {\atomdelaynl{(j) \herm}{n}}
                    {\signalatomnl{k'\!}{n}{j}}
                \big)$  &
            $3$ &  
            $4$ \\ 
        $\delayx{k}{n}{j}$&$\modulusx{k'\!}{n}{j}$ & 
            $\frac{2}{\etan{n}}\Re\big(
                    \left.\complexamplitude{\ncomponent}\right.^{\!\ast} e^{\mathrm{j}\phasex{\ncomponent'\!}{n}{j}} 
                    {\atomdelaynl{(j) \herm}{n}}
                    {\signalatomnl{k'\!}{n}{j}}
                \big)$  &
            $3$ &  
            $5$ \\ 
        $\delayx{k}{n}{j}$&$\etan{n}$ & 
            $0$  &
            $3$ &  
            $6$ \\ 
    \midrule
        $\phasex{k}{n}{j}$&$\phasex{k'\!}{n}{j}$ & 
            $\frac{2}{\etan{n}}\Re\big(
                    \left.\complexamplitude{\ncomponent}\right.^{\!\ast} \complexamplitude{\ncomponent'\!} 
                    \left.\signalatomnl{k}{n}{j}\right.^{\!\herm}
                    {\signalatomnl{k'\!}{n}{j}}
                \big)$  &
            $4$ &  
            $4$ \\ 
        $\phasex{k}{n}{j}$&$\modulusx{k'\!}{n}{j}$ & 
            $\frac{2}{\etan{n}}\Re\big(
                    - \mathrm{j}  \!\left.\complexamplitude{\ncomponent}\right.^{\!\ast} \!\!e^{\mathrm{j}\phasex{\ncomponent'\!}{n}{j}}
                    \!\left.\signalatomnl{k}{n}{j}\right.^{\!\herm}
                    \!{\signalatomnl{k'\!}{n}{j}}
                \big)$  &
            $4$ &  
            $5$ \\ 
            $\phasex{k}{n}{j}$&$\etan{n}$ & 
            $0$  &
            $4$ &  
            $6$ \\ 
            \midrule
        $\modulusx{k}{n}{j}$&$\modulusx{k'\!}{n}{j}$ & 
            $\frac{2}{\etan{n}}\Re\big(
                    e^{-\mathrm{j}\phasex{k}{n}{j}}  e^{\mathrm{j}\phasex{k'\!}{n}{j}}
                    \left.\signalatomnl{k}{n}{j}\right.^{\!\herm}
                    {\signalatomnl{k'\!}{n}{j}}
                \big)$  &
            $5 $ &  
            $5 $ \\ 
            $\modulusx{k}{n}{j}$&$\etan{n}$ & 
            $0$  &
            $5 $ &  
            $6$ \\ 
            \midrule
            $\etan{n}$&$\etan{n}$ & 
            $\frac{\Nz }{\etan{n}^2}$  &
            $6$ &  
            $6$ \\
            \bottomrule
    \end{tabularx}
    \vspace{-8mm}
\end{table}
We use $\grave{\imath}$ and $\acute{\imath}$ to index the rows and columns, respectively, of $\FIMch{n}{j}$. 
Table~\ref{Stab:channel-fim} specifies its upper triangular part, while its lower triangular part follows from the symmetry of \glspl{fim}, i.e., $\left[\FIMch{n}{j}\right]_{\grave{\imath},\acute{\imath}}= \left[\FIMch{n}{j}\right]_{\acute{\imath},\grave{\imath}}$. 
Unlike the unit-modulus array response in~\eqref{eq:unit-modulus-array-response}, 
with the carrier-phase term $\exp (\!-\mathrm{j}\frac{2\pi}{\lightspeed}\fc \lVert \acute{\V{r}}\rVert)$ absorbed into the amplitudes $\widetilde{{\varrho}}_{\!\scriptscriptstyle n}^{\scriptscriptstyle(j)}$, the array response used for deriving the \glspl{pcrlb} factorizes as 
\begin{align}\label{Seq:signalatom}
    \signalatom 
    \big(\delay,\elevation,\azimuth\big) \coloneqq 
        \bm{b}(\delay)\otimes 
        \bm{a}_y(\elevation,\azimuth) \otimes \bm{a}_z(\elevation)
        \quad \in \complexsetone{\Nz} \,,
\end{align}
again, under the assumption of a \gls{ura} and per-anchor plane-wave propagation.
For notational brevity, the dependencies of $\bm{b}(\delay)$, 
$\bm{a}_y(\elevation,\azimuth)$, and $\bm{a}_z(\elevation)$ on their parameters are occasionally omitted.
To completely describe the terms in Table~\ref{Stab:channel-fim}, we further derive the following partial derivatives of 
the array response w.r.t. the channel parameters in $\etach{n}{j}$:
The derivatives of the array response in~\eqref{Seq:signalatom} w.r.t. incidence angles $(\elevation, \azimuth)$ and delay $\delay$ evaluate to 
\begin{align}
    \atomel\big(\delay,\elevation,\azimuth\big) \coloneqq \frac{\partial \signalatom}{\partial \elevation}
    &= \bm{b}  \otimes \dot{\bm{a}}_{y,\elevation} \!\otimes \bm{a}_z + 
     \bm{b} \otimes \bm{a}_y \!\otimes \dot{\bm{a}}_{z,\elevation} 
     \\
    \atomaz\big(\delay,\elevation,\azimuth\big) \coloneqq \frac{\partial \signalatom}{\partial \azimuth}   
    &=  \bm{b}  \otimes \dot{\bm{a}}_{y,\azimuth} \otimes \bm{a}_z 
    {+ \underbrace{ \bm{b}  \otimes \bm{a}_y \otimes \dot{\bm{a}}_{z,\azimuth}}_{=\mathbf{0}} } 
    \\[-9pt]
    \atomdelay\big(\delay,\elevation,\azimuth\big) \coloneqq \frac{\partial \signalatom}{\partial \delay}   
    &= \dot{\bm{b}} \otimes \bm{a}_y \otimes \bm{a}_z 
\end{align}
where the individual derivatives of temporal (denoted $\dot{\bm{b}}$) and spatial (denoted $\dot{\bm{a}}$) array responses are given below.
The partial derivative of the horizontal spatial array response in~\eqref{eq:spatial-response-y} w.r.t. elevation angle $\elevation$ becomes
\begin{align*}
    \dot{\bm{a}}_{y,\elevation}(\elevation,\azimuth)  = \frac{\partial \bm{a}_y}{\partial \elevation } 
    = 
    \mathrm{j} 
    \frac{2 \pi}{\lambda} 
    \cos (\elevation) 
    \sin (\azimuth) 
    \,\mathbf{p}_y 
    \!\odot\!
    \bm{a}_y(\elevation,\azimuth)  
    \in \complexsetone{\Nantennasy} ,
\end{align*}
and the partial derivative w.r.t. azimuth angle $\azimuth$ becomes
\begin{align*}
    \dot{\bm{a}}_{y,\azimuth}(\elevation,\azimuth)  = \frac{\partial \bm{a}_y}{\partial \azimuth }
    = 
    \mathrm{j} \frac{2 \pi}{\lambda} 
    \sin (\elevation) 
    \cos (\azimuth) 
    \,\mathbf{p}_y
    \!\odot\! 
    \bm{a}_y(\elevation,\azimuth)  
    \in \complexsetone{\Nantennasy}.
\end{align*}
The partial derivative of the vertical spatial array response in \eqref{eq:spatial-response-z} w.r.t. elevation angle $\elevation$ becomes
\begin{align*}
    \dot{\bm{a}}_{z,\elevation}(\elevation) = \frac{\partial \bm{a}_z}{\partial \elevation }
      =  -\mathrm{j} \frac{2 \pi}{\lambda}  \sin (\elevation) ~ \mathbf{p}_z 
    \odot \bm{a}_z(\elevation)    
    \,,
\end{align*}
and the partial derivative w.r.t. azimuth angle $\azimuth$ becomes
\begin{align}\label{Seq:derivative-vertical-response-wrt-azimuth}
    \dot{\bm{a}}_{z,\azimuth} = \frac{\partial \bm{a}_z}{\partial \azimuth }
    =  \bm{0} \quad \in \complexsetone{\Nantennasz} \,,
\end{align}
For the temporal response in \eqref{eq:delay-array-response}, the derivative w.r.t. delay is
\begin{align}
    \dot{\bm{b}}(\delay) = \frac{\partial \bm{b}}
    {\partial \delay }
    =  -\mathrm{j} 2 \pi \mathbf{f} \odot \bm{b}(\delay) \quad \in \complexsetone{\Nfrequency} \,.
\end{align}
We further define the shorthand notation
$\signalatomnl{k}{n}{j}\coloneqq\signalatom\big(\delayx{k}{n}{j},\elx{k}{n}{j},\azx{k}{n}{j}\big)$,
$\atomdelaynl{(j)}{n}\coloneqq\atomdelay\big(\delayx{k}{n}{j},\elx{k}{n}{j},\azx{k}{n}{j}\big)$,
$\atomelnl{(j)}{n}\coloneqq\atomel\big(\delayx{k}{n}{j},\elx{k}{n}{j},\azx{k}{n}{j}\big)$, and
$\atomaznl{(j)}{n}\coloneqq\atomaz\big(\delayx{k}{n}{j},\elx{k}{n}{j},\azx{k}{n}{j}\big)$.

\subsection{Derivation of Jacobian Matrices}\label{Ssec:Jacobian}  

As mentioned in Sec.\,\ref{sec:PCRLB}, the differences among (i) the noncoherent, (ii) the coherent, and (iii) the carrier-phase-based \glspl{pcrlb} are fundamentally tied to the number $\dimPhase$ of nuisance phases that need to be estimated.
For (i) the noncoherent \gls{pcrlb}, the inversion in~\eqref{eq:posteriorCRLB} eventually costs the information on the global parameters of interest that the $J$ phases of all $J$ \glspl{pa} in $\RVetach{n}{j}$ could otherwise have contributed, because they map to $J$ distinct nuisance phases in $\RVetaglobal{n}$ that are different across the $J$ \glspl{pa} (cf.\,\cite[eq.\,(8.43)]{VanTrees2002optimumASP}).
By contrast, for (ii) the coherent \gls{pcrlb}, the $J$ \gls{pa} phases in $\RVetach{n}{j}$ map to only one nuisance phase in $\RVetaglobal{n}$, so the information not consumed by estimating this nuisance phase remains available for positioning.
Notably, this information gain comes from retaining the aperture of a jointly coherent \gls{dmimo} infrastructure of distributed \glspl{pa}.
Finally, in the case of (iii) the carrier-phase-based \gls{pcrlb}, no nuisance phases need to be estimated and hence all phase information contributes to the \gls{mt} position.

The parameter \textit{mapping} from local channel parameters $\RVetach{n}{j}$ 
to global parameters $\RVetaglobal{n}$
is contained in the Jacobian matrices $\jacobgn{j}$ that we derive in the following.
Mapping from local channel parameters $\etach{n}{j}$ to global parameters $\etaglobal{n}$, the $(\dimGlobal\!\times\!\dimLocal)$ Jacobian matrices are defined as\footnote{Note that the unusual definition of the Jacobian matrices $\jacobgn{j}\!=\! \nicefrac{\partial {\etach{n}{j}}^{\!\!\trp}}{\partial \etaglobal{n}} \!\in\!  \realset{\dimGlobal}{\dimLocal}$ comes from the fact that they are defined to propagate Fisher information (i.e., precision) rather than covariance.}
\begin{align}\label{eq:jacobian-main}
    \jacobgn{j} \!\coloneqq\!   
    \frac{\partial \etach{n}{j}\!^\trp }{\partial \etaglobal{n}} \!=\!
    \begin{bmatrixs}
        \jacobP{\elevation} & \jacobP{\azimuth}& \jacobP{\delay}& \jacobP{\varphi}& \bm{0} & 0 \\
        \bm{0} & \bm{0} & \bm{0}            &\bm{0}                 & \bm{0}    & 0 \\
        \bm{0} & \bm{0} & \bm{0}            & \jacobC{\varphi}      & \bm{0}    & 0 \\
        \bm{0} & \bm{0} & \bm{0}            & \bm{0}                & \jacobA{a}& 0 \\
        0 & 0 & 0            & 0            & 0    & 1 \\
    \end{bmatrixs}
\end{align}
with the Jacobian submatrices $\jacobP{\cdot}$ mapping information from the respective channel parameters to the \gls{mt} position $\pos{n}$, the submatrix $\jacobC{\varphi}$ mapping information to the nuisance phases $\phasevec{n}{(j)}$, and $\jacobA{a}$ mapping information to the nuisance amplitude moduli $\modulivec{n}{(j)}$ in $\etaglobal{n}$.
The zero rows belonging to the \gls{mt} velocity $\vel{n}$ in the Jacobian $\jacobgn{j}$ make the respective global per-snapshot \gls{fim} $\FIMclassic{n}$ in~\eqref{eq:FIMclassic} singular. 
Information about the \gls{mt} velocity $\RVvel{n}$ in the posterior \gls{fim} $\FIMstep{n}{n}$ of~\eqref{eq:posteriorCRLB}---restoring its full rank---is then obtained only by propagating information about the position over the state-space model in~\eqref{eq:priorFIM} (via off-diagonal elements in $\transitionmatrix_{\text{\tiny a}}$).

\subsubsection{Fundamental Jacobian Matrices}
As a result of the chain rule, the Jacobian submatrices in~\eqref{eq:jacobian-main} are products of more fundamental Jacobian building blocks, which we derive first:

Abbreviating $\rx:=[\rangep{k}{n}{j}]_{\scriptscriptstyle 1}$, $\ry:=[\rangep{k}{n}{j}]_{\scriptscriptstyle 2}$, and $\rz:=[\rangep{k}{n}{j}]_{\scriptscriptstyle 3}$ for notational brevity, we derive the mapping of Fisher information in \textit{local} spherical \gls{pa} coordinates $\{\delayx{k}{n}{j},\elx{k}{n}{j},\azx{k}{n}{j}\}$ to \textit{local} Cartesian \gls{pa} coordinates $\rangep{k}{n}{j}\in\realsetone{3}$.
The mapping in elevation $\elx{k}{n}{j} = \arccos (\rz/ \lVert \rangep{k}{n}{j} \rVert)$ from \eqref{eq:elevation} is 
\begin{align}
    \frac{\partial \elx{k}{n}{j} }{\partial\rangep{k}{n}{j} } = \frac{1}{\lVert \rangep{k}{n}{j} \rVert^2 \, \sqrt{ \rx^2 + \ry^2} } 
    \begin{bmatrix}
        \rx \rz \\
        \ry \rz \\
        -\!\left(\rx^2 + \ry^2\right) 
    \end{bmatrix}
    \quad \in \realset{3}{1} \,,
\end{align}
and in azimuth $\azx{k}{n}{j} = \arctantwo(\ry,\rx)$ from \eqref{eq:azimuth} 
\begin{align}
     \frac{\partial \azx{k}{n}{j} }{\partial\rangep{k}{n}{j} } = 
    \frac{1}{\rx^2 + \ry^2}
    \begin{bmatrix}
        -\ry \\
        \rx \\
        0
    \end{bmatrix} 
    \quad \in \realset{3}{1} \,.
\end{align}
The mapping of Fisher information in \textit{local} spherical \gls{pa} coordinates to \textit{local} Cartesian \gls{pa} coordinates in perceived delay $\delayx{k}{n}{j} =  \frac{\lVert \rangep{k}{n}{j} \rVert}{\lightspeed} $ from~\eqref{eq:delay} is 
\begin{align}
     \frac{\partial \delayx{k}{n}{j} }{\partial\rangep{k}{n}{j} } = \frac{\rangep{k}{n}{j}}{\lightspeed ~ \lVert \rangep{k}{n}{j} \rVert} \quad \in \realset{3}{1} \,,
\end{align}
and in carrier phase 
$\phasex{k}{n}{j} = -\frac{2\pi}{\lightspeed}\fc \lVert \acute{\V{r}}\rVert + {\varphi}_{\!\scriptscriptstyle n}^{\scriptscriptstyle(j)}$
is 
\begin{align}
     \frac{\partial \phasex{k}{n}{j} }{\partial\rangep{k}{n}{j} } = 
    \frac{-2 \pi ~ \rangep{k}{n}{j}}{\lambda ~ \lVert \rangep{k}{n}{j} \rVert} \quad \in \realset{3}{1} \,,
\end{align}
 with wavelength $\lambda = \frac{\lightspeed}{\fc}$.
As mentioned above, we assume no mapping from amplitude moduli to \gls{mt} position, i.e., 
$\jacobP{\modulus} = \frac{\partial {\modulivec{n}{(j)}}^\trp}{\partial \pos{n}} \approx \mathbf{0}$ under the approximation $\frac{\partial g_{\scriptscriptstyle n}^{\scriptscriptstyle (j)}}{\partial \rangep{k}{n}{j}}\approx \mathbf{0}$, because their information contribution to the \gls{mt} position is negligible compared to the contributions from the other channel parameters 
$\big\{
\RVelVecx{n}{j}, 
\RVazVecx{n}{j},
\RVdelayVecx{n}{j}, 
\unslant[-.25]{\widetilde{\varphi}}_{\!\scriptscriptstyle n}^{\scriptscriptstyle(j)}
\big\}$.

The following Jacobian terms map between coordinate systems:
The mapping from $\rangep{k}{n}{j}$ defined in~\eqref{eq:rangepSupp} in \textit{local} Cartesian \gls{pa} coordinates to \textit{global} Cartesian coordinates is 
\begin{align}
    \frac{\partial\rangep{k}{n}{j}\!^\trp}{\partial\range{k}{n}{j}}  = \rotM{j} 
    \quad \in SO(3) \,,
\end{align}
due to the orthogonality of rotation matrices, i.e., $\rotM{j}^{-1} = \rotM{j}^\trp$.
Through~\eqref{eq:range}, information about the vector $\range{k}{n}{j}$ maps to the \gls{mt} position $\pos{n}$ via
\begin{align}\label{eq:Jacob-ranget-pos}
    \frac{\partial\range{k}{n}{j}\!^\trp}{\partial\pos{n}}  = 
    \eye{3}
    \quad \in \realset{3}{3} \,.
\end{align}

These fundamental Jacobian building blocks will reappear in the submatrices of the Jacobians $\jacobgn{j}$ in~\eqref{eq:jacobian-main}, which we define next.

\subsubsection{Positioning Submatrices}       
Using the chain rule, the Jacobians mapping to the \gls{mt} position are assembled from these fundamental building blocks.
\begingroup             
\allowdisplaybreaks[4]
\begin{align}
    \jacobP{\elevation} 
    &= 
    \frac{\partial\range{k}{n}{j}\!^\trp}{\partial\pos{n}}
    \frac{\partial\rangep{k}{n}{j}\!^\trp}{\partial\range{k}{n}{j}}
    \frac{\partial \elx{k}{n}{j} }{\partial\rangep{k}{n}{j} } \quad \in \realset{3}{1}  &&\\
    &= 
    \eye{3} ~ \rotM{j} ~ 
    \frac{1}{\lVert \rangep{k}{n}{j} \rVert^2 \, \sqrt{ \rx^2 + \ry^2} } 
    \begin{bmatrix}
        \rx \rz \\
        \ry \rz \\
        -\!\left(\rx^2 + \ry^2\right) 
    \end{bmatrix}  & 
    \nonumber&\\
    \jacobP{\azimuth} 
    &= 
    \frac{\partial\range{k}{n}{j}\!^\trp}{\partial\pos{n}}
    \frac{\partial\rangep{k}{n}{j}\!^\trp}{\partial\range{k}{n}{j}}
     \frac{\partial \azx{k}{n}{j} }{\partial\rangep{k}{n}{j} }\nonumber &&\\
    &= 
    \eye{3} ~ \rotM{j} ~ 
    \frac{1}{\rx^2 + \ry^2}
    \begin{bmatrix}
        -\ry \\
        \rx \\
        0
    \end{bmatrix} \quad \in \realset{3}{1} & 
    &\\
    \jacobP{\delay} 
    &= 
    \frac{\partial\range{k}{n}{j}\!^\trp}{\partial\pos{n}}
    \frac{\partial\rangep{k}{n}{j}\!^\trp}{\partial\range{k}{n}{j}}
    \frac{\partial \delayx{k}{n}{j} }{\partial\rangep{k}{n}{j} } \nonumber &&\\
    &= 
    \eye{3} ~ \rotM{j} ~ 
    \frac{\rangep{k}{n}{j}}{\lightspeed ~ \lVert \rangep{k}{n}{j} \rVert} \quad \in \realset{3}{1}     \, & 
    &\\
    \jacobP{\varphi} 
    &= 
    \frac{\partial\range{k}{n}{j}\!^\trp}{\partial\pos{n}}
    \frac{\partial\rangep{k}{n}{j}\!^\trp}{\partial\range{k}{n}{j}}
    \frac{\partial \phasex{k}{n}{j} }{\partial\rangep{k}{n}{j} } \nonumber &&\\
    &= 
    \eye{3} ~ \rotM{j} ~ 
    \frac{-2 \pi ~ \rangep{k}{n}{j}}{\lambda ~ \lVert \rangep{k}{n}{j} \rVert} \quad \in \realset{3}{1}  \,.   
    &
\end{align}
\endgroup

\subsubsection{Nuisance Parameter Submatrices}       
Let $\binsetone{}\!\coloneqq\!\{0,1\}$ denote the binary set.
The Jacobian matrices mapping from phases in $\etach{n}{j}$ to phases in $\etaglobal{n}$ are
\begin{align*}
    \jacobC{\varphi} = 
    \frac{
	\partial{\widetilde{\varphi}}_{\!\scriptscriptstyle n}^{\scriptscriptstyle(j)}
	}{\partial \phasevec{n}{}} = 
    \begin{cases}
        \unitvector{j}   \quad \in \binset{J}{1} & \text{noncoherent}
        \\
        \,1\, \quad\in \binsetone{} & \text{coherent}
        \\
         [~] \quad \in \binset{0}{1} & \text{carrier-phase-based}
    \end{cases}
\end{align*}
and the Jacobian matrices mapping from moduli are
\begin{align}
    \jacobA{a} = 
    \frac{
	\partial{{a}_{\!\scriptscriptstyle n}^{\scriptscriptstyle(j)}}
    }{\partial \modulivec{n}{}} 
    = \unitvector{j}  \quad \in \binset{J}{1}\,,
\end{align}
with $\unitvector{j}\!\in\!\binsetone{J}$ denoting a unit vector whose $j$\textsuperscript{th} entry is one.

\bibliographystyle{IEEEtran}
\balance
\bibliography{IEEEabrv,bibliography,ThisPaper}

%% file: definitions.tex
\newcommand{\IEEEversion}{false}  

\usepackage[normalem]{ulem}      
\makeatletter
\font\redwavefont=lasyb10 scaled 850 
\newcommand{\revise}{%
  \bgroup
  \markoverwith{\lower3.5\p@\hbox{\sixly\textcolor{red}{\redwavefont\char58}}}%
  \ULon
}
\makeatother

\newcommand{\estLabel}[1]{%
  \tikzexternaldisable%
  \tikz[baseline=(char.base)]{%
    \node[shape=rectangle,
          draw=black,
          fill=gray!20,
          rounded corners=0.5pt,
          inner sep=1pt,
          minimum size=1.0em,
          text=black] (char) {\adjustbox{max height=0.65em}{#1}};%
  }%
  \tikzexternalenable%
}

\usepackage{tikz}
\usepackage{tkz-euclide}
\usetikzlibrary{arrows.meta}

\usepackage{setspace}
\usepackage{bm}
\usepackage{mathtools}
\usepackage{graphicx}
\usepackage{relsize}
\usepackage{cite}
\usepackage{amsmath}
\usepackage[dvipsnames]{xcolor}
\usepackage{amssymb}
\usepackage{dblfloatfix}
\usepackage{tabularx}
\usepackage{makecell}
\usepackage{lipsum}
\usepackage{float}
\usepackage{subfiles}
\usepackage{epstopdf}
\usepackage{verbatim}
\usepackage{amsmath}
\usepackage{ifthen}
\usepackage{pgfplots}
\usepackage{pgfplotstable}
\usepackage{tikzscale}
\usepackage{acronym}
\usepackage{amsfonts}
\usepackage{tikz-3dplot}
\usepackage{cuted}

\usepackage{capt-of}

\usetikzlibrary{intersections}  
\usepgfplotslibrary{fillbetween}

\usetikzlibrary{spy,backgrounds}
\usetikzlibrary{fit}
\usetikzlibrary{calc}

\usepackage{placeins}
\usepackage{tikzscale}		
\usepackage{tcolorbox}

\usepackage{ifthen}

\usepackage[framemethod=TikZ]{mdframed}	

\usepackage{pgfplots}
  \tcbset{shield externalize} 
  \pgfplotsset{compat=newest}
  \usetikzlibrary{plotmarks}
  \usetikzlibrary{arrows.meta}
  
  \usetikzlibrary{arrows}
  \usetikzlibrary{mindmap,trees}
  \usetikzlibrary{positioning,spy}
  \usetikzlibrary{calc,fadings,decorations.pathreplacing}
  \usetikzlibrary{decorations.pathmorphing,decorations.text}
  \usetikzlibrary{datavisualization.polar}      
  \usepackage{pgfplotstable}
  \usepgfplotslibrary{polar}
  \usetikzlibrary{patterns}

  \usetikzlibrary{plotmarks}
  \usepgfplotslibrary{patchplots}
  \usepackage{grffile}
  \usepackage{amsmath}
  \usepackage{tikzscale}		
  \usepackage{tikz-3dplot} 
  \usetikzlibrary{positioning}
  \usetikzlibrary{arrows}
  \usetikzlibrary{fit,backgrounds}  

\pgfdeclarelayer{behindBackground}
\pgfdeclarelayer{background}
\pgfdeclarelayer{foreground}
\pgfsetlayers{behindBackground,background,main,foreground}   

\tikzset{%
  >=latex,
  inner sep=0pt,%
  outer sep=2pt,%
  mark coordinate/.style={inner sep=0pt,outer sep=0pt,minimum size=3pt,
  fill=black,circle}%
}

\renewcommand{\baselinestretch}{1} 

\colorlet{veccol}{green!50!black}
\colorlet{projcol}{blue!70!black}
\colorlet{myblue}{blue!80!black}
\colorlet{myred}{red!90!black}
\colorlet{mydarkblue}{blue!50!black}
\tikzset{>=latex} 
\tikzstyle{proj}=[projcol!80,line width=0.08] 
\tikzstyle{area}=[draw=veccol,fill=veccol!80,fill opacity=0.6]
\tikzstyle{vector}=[-stealth,myblue,thick,line cap=round]
\tikzstyle{unit vector}=[->,veccol,thick,line cap=round]
\tikzstyle{dark unit vector}=[unit vector,veccol!70!black]
\usetikzlibrary{angles,quotes} 

\definecolor{RDlightgreen}{RGB}{141 192 69}
\definecolor{RDgreen}{rgb}{0.3647, 0.4275, 0.2667}
\definecolor{RDdarkgreen}{rgb}{0.2196, 0.2196, 0.2196}
\definecolor{RDmaroon}{rgb}{.522,.22,.353} %

\definecolor{IEEEblue}{RGB}{0 98 155}
\definecolor{IEEElightblue}{RGB}{0 181 226}
\definecolor{IEEEturquoise}{RGB}{0 156 166}
\definecolor{IEEEred}{RGB}{186 12 47}
\definecolor{IEEEgreen}{RGB}{0 132 61}
\definecolor{IEEElightgreen}{RGB}{120 190 32}
\definecolor{IEEEorange}{RGB}{225 163 0}
\definecolor{IEEEyellow}{RGB}{255 209 0}
\definecolor{IEEEviolet}{RGB}{152 29 151}
\definecolor{IEEEdarkmaroon}{RGB}{134 31 65}
\definecolor{IEEEdarkorange}{RGB}{232 119 34}

\colorlet{red}{IEEEred}
\colorlet{blue}{IEEEblue}
\colorlet{orange}{IEEEorange}
\colorlet{green}{IEEEgreen}
\definecolor{violet}{RGB}{152 29 151}

\definecolor{sand}{rgb}{0.88235,0.63922,0.00000}
\colorlet{IEEEyellow}{sand}
\colorlet{SFVcolor}{IEEElightblue}

\newcommand{\ist}{\hspace*{.3mm}}

\newcommand{\iist}{\hspace*{1mm}}

\newcommand{\nn}{\nonumber}

\DeclareMathAlphabet{\mathsfbr}{OT1}{cmss}{m}{n}
\SetMathAlphabet{\mathsfbr}{bold}{OT1}{cmss}{bx}{n}
\DeclareRobustCommand{\msf}[1]{%
	\ifcat\noexpand#1\relax\msfgreek{#1}\else\mathsfbr{#1}\fi
}

\makeatletter
\newcommand{\msfgreek}[1]{\csname s\expandafter\@gobble\string#1\endcsname}
\makeatother

\DeclareRobustCommand{\mcal}[1]{%
	\ifcat\noexpand#1\relax\mathnormal{#1}\else\cal{#1}\fi
}
\DeclareRobustCommand{\BM}[1]{%
	\ifcat\noexpand#1\relax\bm{\boldUppercaseItalicGreek{#1}}\else\bm{#1}\fi
}
\makeatletter
\newcommand{\boldUppercaseItalicGreek}[1]{\csname var\expandafter\@gobble\string#1\endcsname}
\makeatother

\newcommand{\rv}[1]{\msf{#1}}
\newcommand{\RV}[1]{\bm{\msf{#1}}}

\newcommand{\V}[1]{\bm{#1}}
\newcommand{\M}[1]{\BM{#1}}

\let\geq\geqslant

\let\succeq\succcurlyeq

\providecommand{\IEEEQED}{\IEEEQEDopen}

\newcommand{\E}[0]{\mathbb{E}}

\DeclareMathVersion{timesmath}
\SetSymbolFont{letters}{timesmath}{OML}{ntxmi}{m}{it}
\DeclareMathVersion{timesmathbold}
\SetSymbolFont{letters}{timesmathbold}{OML}{ntxmi}{b}{it}

\makeatletter
\newif\ifAC@uppercase@first%
\def\Aclp#1{\AC@uppercase@firsttrue\aclp{#1}\AC@uppercase@firstfalse}%
\def\AC@aclp#1{%
	\ifcsname fn@#1@PL\endcsname%
	\ifAC@uppercase@first%
	\expandafter\expandafter\expandafter\MakeUppercase\csname fn@#1@PL\endcsname%
	\else%
	\csname fn@#1@PL\endcsname%
	\fi%
	\else%
	\AC@acl{#1}s%
	\fi%
}%
\def\Acp#1{\AC@uppercase@firsttrue\acp{#1}\AC@uppercase@firstfalse}%
\def\AC@acp#1{%
	\ifcsname fn@#1@PL\endcsname%
	\ifAC@uppercase@first%
	\expandafter\expandafter\expandafter\MakeUppercase\csname fn@#1@PL\endcsname%
	\else%
	\csname fn@#1@PL\endcsname%
	\fi%
	\else%
	\AC@ac{#1}s%
	\fi%
}%
\def\Acfp#1{\AC@uppercase@firsttrue\acfp{#1}\AC@uppercase@firstfalse}%
\def\AC@acfp#1{%
	\ifcsname fn@#1@PL\endcsname%
	\ifAC@uppercase@first%
	\expandafter\expandafter\expandafter\MakeUppercase\csname fn@#1@PL\endcsname%
	\else%
	\csname fn@#1@PL\endcsname%
	\fi%
	\else%
	\AC@acf{#1}s%
	\fi%
}%
\def\Acsp#1{\AC@uppercase@firsttrue\acsp{#1}\AC@uppercase@firstfalse}%
\def\AC@acsp#1{%
	\ifcsname fn@#1@PL\endcsname%
	\ifAC@uppercase@first%
	\expandafter\expandafter\expandafter\MakeUppercase\csname fn@#1@PL\endcsname%
	\else%
	\csname fn@#1@PL\endcsname%
	\fi%
	\else%
	\AC@acs{#1}s%
	\fi%
}%
\edef\AC@uppercase@write{\string\ifAC@uppercase@first\string\expandafter\string\MakeUppercase\string\fi\space}%
\def\AC@acrodef#1[#2]#3{%
	\@bsphack%
	\protected@write\@auxout{}{%
		\string\newacro{#1}[#2]{\AC@uppercase@write #3}%
	}\@esphack%
}%
\def\Acl#1{\AC@uppercase@firsttrue\acl{#1}\AC@uppercase@firstfalse}
\def\Acf#1{\AC@uppercase@firsttrue\acf{#1}\AC@uppercase@firstfalse}
\def\Ac#1{\AC@uppercase@firsttrue\ac{#1}\AC@uppercase@firstfalse}
\def\Acs#1{\AC@uppercase@firsttrue\acs{#1}\AC@uppercase@firstfalse}
\robustify\Aclp
\robustify\Acfp
\robustify\Acp
\robustify\Acsp
\robustify\Acl
\robustify\Acf
\robustify\Ac
\robustify\Acs
\makeatother

\makeatletter
\def\underbracex#1#2{\mathop{\vtop{\m@th\ialign{##\crcr
				$\hfil\displaystyle{#2}\hfil$\crcr
				\noalign{\kern3\p@\nointerlineskip}%
				#1\crcr\noalign{\kern3\p@}}}}\limits}

\def\upbracefilla{$\m@th \setbox\z@\hbox{$\braceld$}%
	\bracelu\leaders\vrule \@height\ht\z@ \@depth\z@\hfill 
	\leaders\vrule \@height\ht\z@ \@depth\z@\hfill\bracerd
	\braceld\leaders\vrule \@height\ht\z@ \@depth\z@\hfill
	\kern\p@\vrule \@width\p@\kern\p@\vrule \@width\p@\kern\p@\vrule \@width\p@ 
	$}

\def\upbracefilld{$\m@th \setbox\z@\hbox{$\braceld$}%
	\vrule \@width\p@\kern\p@\vrule \@width\p@\kern\p@\vrule \@width\p@\kern\p@
	\leaders\vrule \@height\ht\z@ \@depth\z@\hfill\braceru$}

\makeatother


%% file: tikzStyles.tex
\def\colorNoncohVA{IEEEred}         
\def\colorNoncoh{IEEEblue}          
\def\colorCoh{IEEEgreen}            

\newcommand{\markRepeat}{15}
\newcommand{\markSize}{1}
\newcommand{\LWbound}{1.0pt}
\newcommand{\LWestimates}{0.5pt}
\def\colorSnapshot{IEEEturquoise}        
\def\markSnapshot{*}
\def\colorNC{IEEEblue}        
\def\markNC{square*}
\def\colorC{IEEElightgreen}        
\def\markC{triangle*}
\def\colorCP{IEEEred}        
\def\markCP{diamond*}
\def\colorNCtypeII{IEEElightblue}        
\def\markNCtypeII{pentagon*}
\def\colorNCtypeIIS{IEEEgreen}        
\def\markNCtypeIIS{pentagon}
\def\colorCtypeII{IEEEyellow}        
\def\markCtypeII{Mercedes star flipped}

\tikzstyle{LineCohNC} = [color=gray, 
            line cap = round, 
            line join=round, 
            line width=\LWestimates,
            mark=Mercedes star,
            mark repeat=\markRepeat, 
            mark size=\markSize,
            mark options={solid, line width=0.3pt,fill=gray}]

\tikzstyle{LineNoncohVA} = [        
            color=\colorNoncohVA, 
            line cap = round, 
            line join=round, 
            line width=\LWestimates,
            mark=*, 
            mark repeat=\markRepeat, 
            mark size=0.7\markSize,
            mark options={solid, line width=0.3pt,fill=\colorNoncohVA}]

\tikzstyle{LineNoncoh} = [color=\colorNoncoh, 
            line cap = round, 
            line join=round, 
            line width=\LWestimates,
            mark=triangle*, 
            mark repeat=\markRepeat, 
            mark size=0.7\markSize,
            mark options={solid, line width=0.3pt,fill=\colorNoncoh}]

\tikzstyle{LineCoh} = [color=\colorCoh, 
            line cap = round, 
            line join=round, 
            line width=\LWestimates,
            mark=diamond*, 
            mark repeat=\markRepeat, 
            mark size=\markSize,
            mark options={solid, line width=0.3pt,fill=\colorCoh}]

\tikzstyle{LineCohNC} = [color=gray, 
            line cap = round, 
            line join=round, 
            line width=\LWestimates,
            mark=Mercedes star,
            mark repeat=\markRepeat, 
            mark size=\markSize,
            mark options={solid, line width=0.3pt,fill=gray}]


%% file: abbr.tex
\newacronym{2d}{2D}{two-dimensional}
\newacronym{3d}{3D}{three-dimensional}
\newacronym{aoa}{AoA}{angle of arrival}
\newacronym{aod}{AoD}{angle of departure}
\newacronym{awgn}{AWGN}{additive white Gaussian noise}
\newacronym{ap}{AP}{access point}
\newacronym{bp}{BP}{belief propagation}
\newacronym{cdf}{CDF}{cumulative distribution function}
\newacronym{csi}{CSI}{channel state information}
\newacronym{crlb}{CRLB}{Cram\'er--Rao lower bound}
\newacronym{cpu}{CPU}{central processing unit}
\newacronym{pcrlb}{PCRLB}{posterior CRLB}
\newacronym{dm}{DM}{diffuse multipath}
\newacronym{dmc}{DMC}{diffuse multipath component}
\newacronym{efim}{EFIM}{equivalent FIM}
\newacronym{fg}{FG}{factor graph}
\newacronym[longplural={Fisher information matrices}]{fim}{FIM}{Fisher information matrix}
\newacronym{gpu}{GPU}{graphics processing unit}
\newacronym{iot}{IoT}{Internet of Things}
\newacronym{isac}{ISAC}{integrated sensing and communications}
\newacronym{los}{LoS}{line-of-sight}
\newacronym{mc}{MC}{Monte Carlo}
\newacronym{meb}{MEB}{mapping error bound}
\newacronym{mimo}{MIMO}{multiple-input multiple-output}
\newacronym{miso}{MISO}{multiple-input single-output}
\newacronym{ml}{ML}{maximum likelihood}
\newacronym{dmimo}{D-MIMO}{distributed MIMO}
\newacronym{xlmimo}{XL-MIMO}{extremely large-scale multiple-input multiple-output}
\newacronym{mmse}{MMSE}{minimum mean square error}
\newacronym{mmwave}{mmWave}{millimeter wave}
\newacronym{mpc}{MPC}{multipath component}
\newacronym{mpslam}{MP-SLAM}{multipath-based simultaneous localization and mapping}
\newacronym{mrt}{MRT}{maximum ratio transmission}
\newacronym{mse}{MSE}{mean squared error}
\newacronym{mt}{MT}{mobile terminal}
\newacronym{mva}{MVA}{master virtual anchor}
\newacronym{nlos}{NLoS}{non-LoS}
\newacronym{ncv}{NCV}{nearly constant velocity}
\newacronym{oeb}{OEB}{orientation error bound}
\newacronym{olos}{OLoS}{obstructed LoS}
\newacronym{ota}{OTA}{over-the-air}
\newacronym{pa}{PA}{physical anchor}
\newacronym{bs}{BS}{base station}
\newacronym{peb}{PEB}{position error bound}
\newacronym{pdf}{PDF}{probability density function}
\newacronym{pm}{PM}{physical mobile}
\newacronym{pmf}{PMF}{probability mass function}
\newacronym{pf}{PF}{potential feature}
\newacronym{pg}{PG}{path gain}
\newacronym{pl}{PL}{path loss}
\newacronym{pr}{PBR}{particle-based representation}
\newacronym{ppr}{PR}{potential ray}
\newacronym{prop}{PROP}{proposed algorithm}
\newacronym{rcs}{RCS}{radar cross section}
\newacronym{ref}{REF}{reference algorithm}
\newacronym{rf}{RF}{radio frequency}
\newacronym{rfid}{RFID}{radio frequency identification}
\newacronym{rmse}{RMSE}{root mean square error}
\newacronym{roi}{ROI}{region of interest}
\newacronym{rssi}{RSSI}{received signal strength indicator}
\newacronym{rv}{RV}{random variable}
\newacronym{s-parameter}{S-parameter}{scattering parameter}
\newacronym{sfv}{SFV}{surface feature vector}
\newacronym{siso}{SISO}{single-input single-output}
\newacronym{sir}{SIR}{sequential importance resampling}
\newacronym{sis}{SIS}{sequential importance sampling}
\newacronym{slam}{SLAM}{simultaneous localization and mapping}
\newacronym{smc}{SMC}{specular multipath component}
\newacronym{snr}{SNR}{signal-to-noise ratio}
\newacronym{spa}{SPA}{sum-product algorithm}
\newacronym{tdoa}{TDOA}{time-difference-of-arrival}
\newacronym{toa}{TOA}{time-of-arrival}
\newacronym{tosm}{TOSM}{through – open – short – match}
\newacronym{ue}{UE}{user equipment}
\newacronym{ura}{URA}{uniform rectangular array}
\newacronym{uwb}{UWB}{ultrawideband}
\newacronym{va}{VA}{virtual anchor}
\newacronym{vm}{VM}{virtual mobile}
\newacronym{vna}{VNA}{vector network analyzer}
\newacronym{wb}{WB}{wideband}
\newacronym{wpt}{WPT}{wireless power transfer}
\newacronym{xets}{XETS}{cross exponentially tapered slot}

\newacronym{nzm}{NZM}{nonzero-mean}
\newacronym{zm}{ZM}{zero-mean}

\newacronym{sb}{SB}{single-bounce}
\newacronym{db}{DB}{double-bounce}

\newacronym{nc}{NC}{noncoherent}
\newacronym{c}{C}{coherent}
\newacronym{cp}{CP}{carrier-phase-based}

%% file: math-notation.tex
\newcommand{\Nz}[0]{{N}_{\scriptscriptstyle \mathrm{z}}}             
\newcommand{\Nfrequency}[0]{{N}_{\scriptscriptstyle \mathrm{f}}}     
\newcommand{\Nantennas}[0]{{M}}                                      
\newcommand{\Nantennasy}[0]{{M}_{\scriptscriptstyle \mathrm{y}}}     
\newcommand{\Nantennasz}[0]{{M}_{\scriptscriptstyle \mathrm{z}}}     
\newcommand{\setAnchors}[0]{\mathcal{J}}                                            
\newcommand{\RVstate}[1]{\RV{x}_{\scriptscriptstyle #1}}        
\newcommand{\state}[1]{\V{x}_{\scriptscriptstyle #1}}           
\newcommand{\stateHat}[1]{\widehat{\V{x}}_{\scriptscriptstyle #1}}  
\newcommand{\RVpos}[1]{\RV{p}_{\!\scriptscriptstyle #1}}        
\newcommand{\pos}[1]{\V{p}_{\!\scriptscriptstyle #1}}           
\newcommand{\posHat}[1]{\widehat{\V{p}}_{\!\scriptscriptstyle #1}^{\,\text{\tiny MMSE}}}           
\newcommand{\RVvel}[1]{\RV{v}_{\!\scriptscriptstyle #1}}        
\newcommand{\vel}[1]{\V{v}_{\!\scriptscriptstyle #1}}        
\newcommand{\RVPFphi}[2]{\unslant[-.25]{\V{\phi}}_{\!\scriptscriptstyle #2}}        
\newcommand{\PFphiHat}[3]{\widehat{\V{\phi}}_{\!\scriptscriptstyle #2}}                  
\newcommand{\PFphi}[3]{\V{\phi}_{\!\scriptscriptstyle #2}}  

\newcommand{\amplitude}[3]{\rho_{\!\scriptscriptstyle #2}^{\scriptscriptstyle (#3)}}         
\newcommand{\RVamplitude}[3]{\unslant[-.25]{\rho}_{\!\scriptscriptstyle #2}^{\scriptscriptstyle (#3)}}         

\newcommand{\amplitudeHat}[2]{\widehat{\rho}_{\!\scriptscriptstyle #1}^{\,\text{\tiny{ML}}\scriptscriptstyle (#2)}}         

\newcommand{\etaHat}[2]{\widehat{\eta}_{\scriptscriptstyle #1}^{\,\text{\tiny{ML}}\scriptscriptstyle (#2)}}         
\newcommand{\RVamplitudeBar}[1]{\unslant[-.25]{\rho}_{\!\scriptscriptstyle #1}}         
\newcommand{\amplitudeBar}[1]{\rho_{\!\scriptscriptstyle #1}}         
\newcommand{\amplitudeBarHat}[1]{\widehat{\rho}_{\!\scriptscriptstyle #1}^{\,\text{\tiny{C}}}}         
\newcommand{\etaHatBar}[1]{\widehat{\eta}_{\scriptscriptstyle #1}^{\,\text{\tiny{C}}}}         

\newcommand{\amplitudeBarCP}[1]{a_{\!\scriptscriptstyle #1}}         
\newcommand{\amplitudeBarHatCP}[1]{\widehat{a}_{\!\scriptscriptstyle #1}^{\,\text{\tiny{CP}}}}         
\newcommand{\etaHatBarCP}[1]{\widehat{\eta}_{\scriptscriptstyle #1}^{\,\text{\tiny{CP}}}}

\newcommand{\PFmu}[3]{\mu_{\!\scriptscriptstyle #2}}                              
\newcommand{\RVPFmu}[3]{\unslant[-.25]{\mu}_{\!\scriptscriptstyle #2}}           
\newcommand{\PFgamma}[3]{\gamma_{\!\scriptscriptstyle #2}}                       
\newcommand{\RVPFgamma}[3]{\unslant[-.25]{\gamma}_{\!\scriptscriptstyle #2}}     
\newcommand{\observation}[2]{\V{z}_{\scriptscriptstyle #1}^{\scriptscriptstyle(#2)}}         
\newcommand{\RVobservation}[2]{\RV{z}_{\scriptscriptstyle #1}^{\scriptscriptstyle(#2)}}      
\newcommand{\observationn}[1]{\V{z}_{\scriptscriptstyle #1}}                    
\newcommand{\RVobservationn}[1]{\RV{z}_{\scriptscriptstyle #1}}                 
\newcommand{\RVetan}[1]{\unslant[-.25]{\eta}_{\scriptscriptstyle #1}}           
\newcommand{\etan}[1]{{\eta}_{\scriptscriptstyle #1}}                           
\newcommand{\RVetann}[1]{\unslant[-.25]{\bm{\eta}}_{\scriptscriptstyle #1}}                           
\newcommand{\etann}[1]{\bm{\eta}_{\scriptscriptstyle #1}}                       
\newcommand{\etanHat}[1]{\widehat{\eta}_{\scriptscriptstyle #1}}                
\newcommand{\etaXi}[1]{{\eta}_{\scriptscriptstyle\xi,#1}^{\scriptscriptstyle (j)}}                       
\newcommand{\RVnoise}[2]{\RV{n}_{\scriptscriptstyle #1}^{\scriptscriptstyle(j)}}   
\newcommand{\noise}[2]{\V{n}_{\scriptscriptstyle #1}^{\scriptscriptstyle(j)}}   
\newcommand{\elevation}{\theta}                                     
\newcommand{\azimuth}{\vartheta}                                    
\newcommand{\delay}{\tau}     

\newcommand{\steerVec}[1]{\V{\psi}^{\scriptscriptstyle (#1)}}      
\newcommand{\steerVecx}[3]{\V{\psi}_{\!\scriptscriptstyle #2}^{\scriptscriptstyle (#3)}} 
\newcommand{\steerVecBar}[0]{\overline{\V{\psi}}}      
\newcommand{\steerVecxBar}[1]{\overline{\V{\psi}}_{\!\scriptscriptstyle #1}} 

\newcommand{\RVelVecx}[2]{{\unslant[-.25]{\elevation}}_{\scriptscriptstyle #1}^{\!\scriptscriptstyle (#2)}}           
\newcommand{\RVazVecx}[2]{{\unslant[-.25]{\azimuth}}_{\scriptscriptstyle #1}^{\!\scriptscriptstyle (#2)}}             
\newcommand{\RVdelayVecx}[2]{{\unslant[-.25]{\tau}}_{\scriptscriptstyle #1}^{\!\scriptscriptstyle (#2)}}                               

\newcommand{\range}[3]{\bm{r}_{\scriptscriptstyle#2}^{\scriptscriptstyle (#3)}} 
\newcommand{\ranget}[3]{\widetilde{\bm{r}}_{\scriptscriptstyle#2}^{\scriptscriptstyle (#3)}} 
\newcommand{\rangep}[3]{\acute{\bm{r}}_{\scriptscriptstyle#2}^{\scriptscriptstyle (#3)}} 
\newcommand{\rx}[0]{\acute{r}_{\scriptscriptstyle x}}
\newcommand{\ry}[0]{\acute{r}_{\scriptscriptstyle y}}
\newcommand{\rz}[0]{\acute{r}_{\scriptscriptstyle z}}
\newcommand{\pwk}[0]{ \bm{p}^\text{\tiny{w}}_{\scriptscriptstyle k} }           
\newcommand{\nw}[0]{ {\bm{n}^\text{\tiny{w}}_{\scriptscriptstyle k}} }          
\newcommand{\fc}{\mathrm{f}_{\text{\tiny c}}}                       
\newcommand{\lightspeed}{\mathrm{c}}                                
\newcommand{\posPA}[1]{ \mathbf{p}_{\scriptscriptstyle\mathrm{pa}}^{\scriptscriptstyle(#1)}  }           
\newcommand{\rotM}[1]{ \mathbf{M}_{\scriptscriptstyle#1}  }         
\newcommand{\origin}[0]{ \mathbf{0}  }                              

\newcommand{\psfvDet}[1]{\mathbf{p}_{\scriptscriptstyle\mathrm{sfv},\scriptscriptstyle#1} }           
\newcommand{\Rknj}[3]{\bm{R}_{\scriptscriptstyle #1,#2}^{\scriptscriptstyle (#3)}}    

\newcommand{\transitionmatrix}{{\mathbf{\Phi}}}                         
\newcommand{\processNoiseCov}{{\mathbf{Q}}}                             

\newcommand{\belief}[0]{\widetilde{f}}                                          
\newcommand{\Mbeta}[0]{\beta_{\scriptscriptstyle }}                             
\newcommand{\Miota}[1]{\iota_{\scriptscriptstyle }^{\scriptscriptstyle (#1)}}   
\newcommand{\Mkappa}[2]{\kappa^{\scriptscriptstyle (#2)}} 
\newcommand{\Mxi}[0]{\xi}                                 
\newcommand{\Mnu}[1]{\nu^{\scriptscriptstyle (#1)}}       
\newcommand{\muiota}[0]{\bm{\mu}_{\scriptscriptstyle \iota,n}^{\scriptscriptstyle (j)}}      
\newcommand{\Kiota}[0]{\bm{K}_{\scriptscriptstyle \iota,n}^{\scriptscriptstyle (j)}}      
\newcommand{\Ciota}[0]{\bm{C}_{\scriptscriptstyle \iota,n}^{\scriptscriptstyle (j)}}         
\newcommand{\mukappa}[0]{\bm{\mu}_{\scriptscriptstyle \kappa,n}^{\scriptscriptstyle (j)}}      
\newcommand{\Kkappa}[0]{\bm{K}_{\scriptscriptstyle \kappa,n}^{\scriptscriptstyle (j)}}      
\newcommand{\Ckappa}[0]{\bm{C}_{\scriptscriptstyle \kappa,n}^{\scriptscriptstyle (j)}}         
\newcommand{\munu}[0]{\bm{\mu}_{\scriptscriptstyle \nu,n}^{\scriptscriptstyle (j)}}          
\newcommand{\Knu}[0]{\bm{K}_{\scriptscriptstyle \nu,n}^{\scriptscriptstyle (j)}}      
\newcommand{\Cnu}[0]{\bm{C}_{\scriptscriptstyle \nu,n}^{\scriptscriptstyle (j)}}             
\newcommand{\musnj}[1]{\bm{\mu}_{\scriptscriptstyle #1,n}^{\scriptscriptstyle (j)}}       
\newcommand{\Csnj}[1]{\bm{C}_{\scriptscriptstyle #1,n}^{\scriptscriptstyle (j)}}        
\newcommand{\Csnjp}[1]{\bm{C}_{\scriptscriptstyle #1,n}^{\scriptscriptstyle (j)}}        
\newcommand{\weight}[2]{w_{\scriptscriptstyle#1}^{\scriptscriptstyle(#2)}}
\newcommand{\weightt}[2]{\widetilde{w}_{\scriptscriptstyle#1}^{\scriptscriptstyle(#2)}}
\newcommand{\particle}[2]{#1^{\scriptscriptstyle(#2)}}
\newcommand{\normConst}[2]{C_{\scriptscriptstyle#1}^{\scriptscriptstyle#2}} 
\usepackage{dsfont} 

\newcommand{\dimLocal}[0]{ \mathrm{D}_\text{\tiny{ch}} }            
\newcommand{\dimGlobal}[0]{ \mathrm{D}_\text{\tiny{g}} }            
\newcommand{\dimPhase}[0]{ \mathrm{D}_{\scriptscriptstyle \varphi} }            
\newcommand{\parvec}{\theta}                                     
\newcommand{\etaglobal}[1]{\V{\parvec}^{\text{\tiny{g}}}_{\!\scriptscriptstyle#1}}            
\newcommand{\RVetaglobal}[1]{\unslant[-.25]{\V{\parvec}}^{\text{\!\tiny{g}}}_{\!\scriptscriptstyle#1}}            
\newcommand{\RVetaglobalSmall}[1]{\unslant[-.25]{\scriptstyle\V{\parvec}}^{\text{\!\tiny{g}}}_{\!\scriptscriptstyle#1}}            
\newcommand{\RVetaglobalHat}[1]{\hat{\unslant[-.25]{\bm{\parvec}}}^{\text{\!\tiny{g}}}_{\!\scriptscriptstyle#1}}            
\newcommand{\etach}[2]{\V{\parvec}_{\!\text{\tiny{ch}}\scriptscriptstyle,#1}^{\scriptscriptstyle(#2)}}            
\newcommand{\RVetach}[2]{\unslant[-.25]{\V{\parvec}}_{\text{\tiny{ch}}\scriptscriptstyle,#1}^{\!\scriptscriptstyle(#2)}}           
\newcommand{\RVphasevec}[2]{\unslant[-.25]{\bm{\varphi}}_{\!\scriptscriptstyle#1}^{\!\scriptscriptstyle#2}}   
\newcommand{\RVmodulivec}[2]{\RV{a}_{\scriptscriptstyle#1}^{\scriptscriptstyle#2}}   

\newcommand{\PCRLB}[0]{\CRLB_{\scriptscriptstyle n|n}}                 
\newcommand{\classicCRLB}[0]{\CRLB_{\scriptscriptstyle n}^{\text{\tiny F}}}             
\newcommand{\FIMstep}[2]{\FIM_{\scriptscriptstyle #1|#2}}     
\newcommand{\FIM}[0]{\V{I}}                              
\newcommand{\CRLB}[0]{\bm{P}}                      
\newcommand{\FIMch}[2]{\FIM_{\scriptscriptstyle \text{\tiny ch},#1}^{\scriptscriptstyle(#2)}}   
\newcommand{\FIMglobal}[1]{\FIM_{\scriptscriptstyle #1}^{\text{\tiny g}}}   
\newcommand{\FIMclassic}[1]{\FIM_{\scriptscriptstyle #1}^{\text{\tiny F}}}   

\newcommand{\jacobian}{ \boldsymbol{J} }
\newcommand{\jacobgn}[1]{ \jacobian_{\!\scriptscriptstyle n}^{\scriptscriptstyle(#1)}}             
\newcommand{\jacobMVAblock}[2]{ \jacobian_{\scriptscriptstyle k,n}^{\scriptscriptstyle \text{\tiny sb},j}}   
\newcommand{\jacobMVAsb}[1]{ \jacobian_{\scriptscriptstyle n,j}^{\scriptscriptstyle \text{\tiny sb},s}}   

\newcommand{\PEB}[0]{\sigma_{\bm{p}_n}}                                      

\newcommand{\complexamplitude}[1]{{\widetilde{{\varrho}}_{\scriptscriptstyle n}^{\scriptscriptstyle\,(j)}}}
\newcommand{\signalatom}[0]{\V{\uppsi}}                                 
\newcommand{\signalatomnl}[3]{\signalatom_{\!\scriptscriptstyle #2}^{\scriptscriptstyle (#3)}} 

\newcommand{\interval}[1]{\hat{I}_{\scriptscriptstyle#1}}				

%% file: figures/LHFs/LHFs.tex
\pgfplotsset{every axis/.append style={
  label style={font=\footnotesize},
  legend style={font=\scriptsize},
  tick label style={font=\footnotesize},
}}

\definecolor{mycolor1}{rgb}{0.00000,0.51765,0.23922}%
\definecolor{mycolor2}{rgb}{0.72941,0.04706,0.18431}%
\definecolor{mycolor3}{rgb}{0.49400,0.18400,0.55600}%
\begin{tikzpicture}

\begin{axis}[%
name=boundaryLeft,  
width=0.878\figurewidth,
height=0.625\figurewidth,
at={(0\figurewidth,0\figurewidth)},
axis line style = thick,	
line cap = round,	
line join = round,	
scale only axis,
point meta min=-5,
point meta max=0,
axis on top,
xmin=-3.5,
xmax=4.487327631,
ylabel={$y$ in \SI{}{\metre}},
ylabel style={yshift=-1.5mm},
ymin=-1.1,
ymax=4.481919232,
tick align=inside,
xtick = {-3,-2,...,4},
ytick = {-1,0,...,4},
    minor x tick num = 1, 
    minor y tick num = 1, 
    xticklabel=\empty,
]

\addplot [forget plot] graphics [xmin=-3.50535343675, xmax=4.50338794125, ymin=-1.20535343675, ymax=4.50141013875] {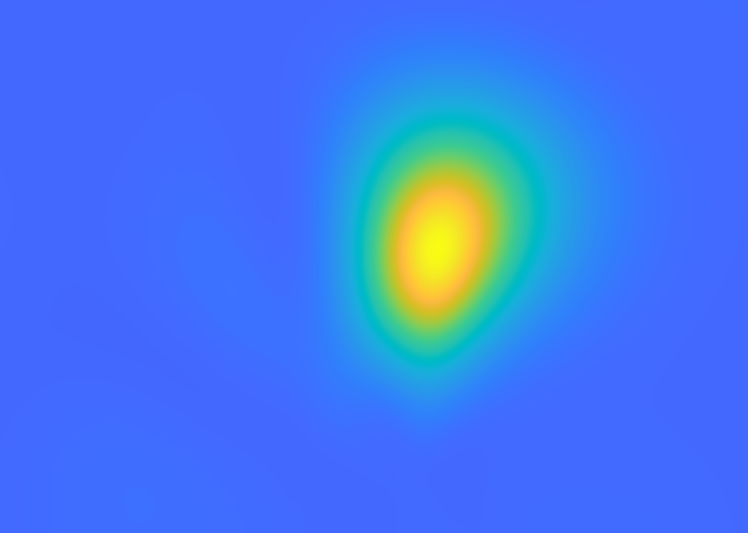};

\input{\datapath/trajectory.tex}

\addplot [color=blue, line width=0.5pt, only marks, mark size=0.9pt, mark=o, mark options={solid, IEEEred}, forget plot]
  table[row sep=crcr]{%
1.18592964824121	1.8675879396984\\
};

\input{\datapath/arrays2D.tex}
\end{axis}

\node[above left, align=right,font={\footnotesize},fill=white,
opacity=0.9,
inner sep=0.75mm, 
xshift=0.65mm, 
yshift=-0.65mm, 
draw, 
line width = 0.8pt,
minimum width = 3mm,
minimum height = 3mm,
] at %
(boundaryLeft.south east)
{a)};


\node[
    below left,
    align=right,
    shape=rectangle,
    draw=black,
    fill=gray!20!white,
    rounded corners=0.5pt,
    inner sep=1pt,
    minimum size=1.0em,
    text=black
] at (boundaryLeft.north east)
{\adjustbox{max height=0.65em}{NC}};

\begin{axis}[%
name=boundaryCenter,  
width=0.878\figurewidth,
height=0.625\figurewidth,
at={(0.9\figurewidth,0\figurewidth)},
axis line style = thick,	
line cap = round,	
line join = round,	
scale only axis,
point meta min=-5,
point meta max=0,
axis on top,
xmin=-3.5,
xmax=4.487327631,
xlabel style={font=\color{white!15!black}},
xlabel=\empty,
yticklabel=\empty,                      
ymin=-1.1,
ymax=4.481919232,
tick align=inside,
xtick = {-3,-2,...,4},
ytick = {-1,0,...,4},
    minor x tick num = 1, 
    minor y tick num = 1, 
    xticklabel=\empty,
]

\addplot [forget plot] graphics [xmin=-3.50535343675, xmax=4.50338794125, ymin=-1.20535343675, ymax=4.50141013875] {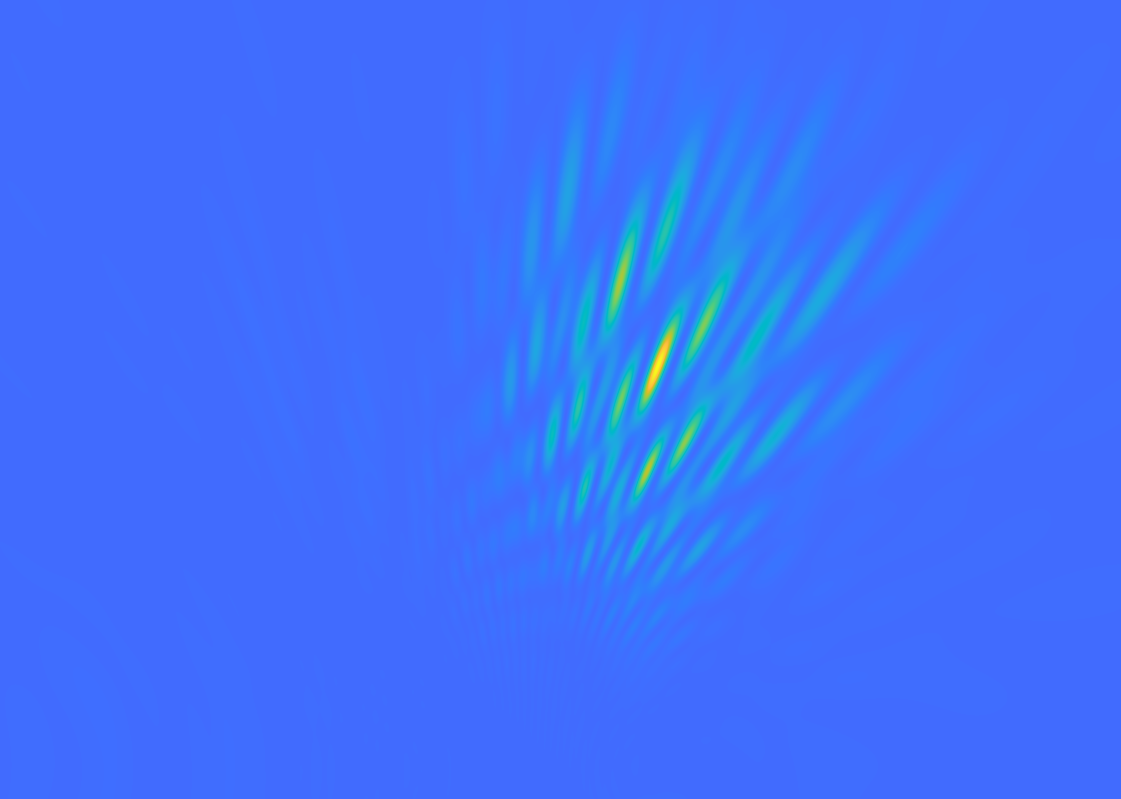};

\input{\datapath/trajectory.tex}

\addplot [color=blue, line width=0.5pt, only marks, mark size=0.9pt, mark=o, mark options={solid, IEEEred}, forget plot]
  table[row sep=crcr]{%
1.18592964824121	1.8675879396984\\
};

\input{\datapath/arrays2D.tex}

\end{axis}

\node[above left, align=right,font={\footnotesize},fill=white,
opacity=0.9,
inner sep=0.75mm, 
xshift=0.65mm, 
yshift=-0.65mm, 
draw, 
line width = 0.8pt,
minimum width = 3mm,
minimum height = 3mm,
] at %
(boundaryCenter.south east)
{b)};


\node[
    below left,
    align=right,
    shape=rectangle,
    draw=black,
    fill=gray!20!white,
    rounded corners=0.5pt,
    inner sep=1pt,
    minimum size=1.0em,
    text=black
] at (boundaryCenter.north east)
{\adjustbox{max height=0.65em}{C}};

\begin{axis}[%
name=boundaryRight,  
width=0.878\figurewidth,
height=0.625\figurewidth,
at={(1.8\figurewidth,0\figurewidth)},
axis line style = thick,	
line cap = round,	
line join = round,	
scale only axis,
point meta min=-7,
point meta max=0,
axis on top,
xmin=-3.5,
xmax=4.487327631,
yticklabel=\empty,                      
ymin=-1.1,
ymax=4.481919232,
tick align=inside,
xtick = {-3,-2,...,4},
ytick = {-1,0,...,4},
    minor x tick num = 1, 
    minor y tick num = 1, 
    xticklabel=\empty,
colormap={mymap}{[1pt] rgb(0pt)=(0.2422,0.1504,0.6603); rgb(1pt)=(0.2444,0.1534,0.6728); rgb(2pt)=(0.2464,0.1569,0.6847); rgb(3pt)=(0.2484,0.1607,0.6961); rgb(4pt)=(0.2503,0.1648,0.7071); rgb(5pt)=(0.2522,0.1689,0.7179); rgb(6pt)=(0.254,0.1732,0.7286); rgb(7pt)=(0.2558,0.1773,0.7393); rgb(8pt)=(0.2576,0.1814,0.7501); rgb(9pt)=(0.2594,0.1854,0.761); rgb(11pt)=(0.2628,0.1932,0.7828); rgb(12pt)=(0.2645,0.1972,0.7937); rgb(13pt)=(0.2661,0.2011,0.8043); rgb(14pt)=(0.2676,0.2052,0.8148); rgb(15pt)=(0.2691,0.2094,0.8249); rgb(16pt)=(0.2704,0.2138,0.8346); rgb(17pt)=(0.2717,0.2184,0.8439); rgb(18pt)=(0.2729,0.2231,0.8528); rgb(19pt)=(0.274,0.228,0.8612); rgb(20pt)=(0.2749,0.233,0.8692); rgb(21pt)=(0.2758,0.2382,0.8767); rgb(22pt)=(0.2766,0.2435,0.884); rgb(23pt)=(0.2774,0.2489,0.8908); rgb(24pt)=(0.2781,0.2543,0.8973); rgb(25pt)=(0.2788,0.2598,0.9035); rgb(26pt)=(0.2794,0.2653,0.9094); rgb(27pt)=(0.2798,0.2708,0.915); rgb(28pt)=(0.2802,0.2764,0.9204); rgb(29pt)=(0.2806,0.2819,0.9255); rgb(30pt)=(0.2809,0.2875,0.9305); rgb(31pt)=(0.2811,0.293,0.9352); rgb(32pt)=(0.2813,0.2985,0.9397); rgb(33pt)=(0.2814,0.304,0.9441); rgb(34pt)=(0.2814,0.3095,0.9483); rgb(35pt)=(0.2813,0.315,0.9524); rgb(36pt)=(0.2811,0.3204,0.9563); rgb(37pt)=(0.2809,0.3259,0.96); rgb(38pt)=(0.2807,0.3313,0.9636); rgb(39pt)=(0.2803,0.3367,0.967); rgb(40pt)=(0.2798,0.3421,0.9702); rgb(41pt)=(0.2791,0.3475,0.9733); rgb(42pt)=(0.2784,0.3529,0.9763); rgb(43pt)=(0.2776,0.3583,0.9791); rgb(44pt)=(0.2766,0.3638,0.9817); rgb(45pt)=(0.2754,0.3693,0.984); rgb(46pt)=(0.2741,0.3748,0.9862); rgb(47pt)=(0.2726,0.3804,0.9881); rgb(48pt)=(0.271,0.386,0.9898); rgb(49pt)=(0.2691,0.3916,0.9912); rgb(50pt)=(0.267,0.3973,0.9924); rgb(51pt)=(0.2647,0.403,0.9935); rgb(52pt)=(0.2621,0.4088,0.9946); rgb(53pt)=(0.2591,0.4145,0.9955); rgb(54pt)=(0.2556,0.4203,0.9965); rgb(55pt)=(0.2517,0.4261,0.9974); rgb(56pt)=(0.2473,0.4319,0.9983); rgb(57pt)=(0.2424,0.4378,0.9991); rgb(58pt)=(0.2369,0.4437,0.9996); rgb(59pt)=(0.2311,0.4497,0.9995); rgb(60pt)=(0.225,0.4559,0.9985); rgb(61pt)=(0.2189,0.462,0.9968); rgb(62pt)=(0.2128,0.4682,0.9948); rgb(63pt)=(0.2066,0.4743,0.9926); rgb(64pt)=(0.2006,0.4803,0.9906); rgb(65pt)=(0.195,0.4861,0.9887); rgb(66pt)=(0.1903,0.4919,0.9867); rgb(67pt)=(0.1869,0.4975,0.9844); rgb(68pt)=(0.1847,0.503,0.9819); rgb(69pt)=(0.1831,0.5084,0.9793); rgb(70pt)=(0.1818,0.5138,0.9766); rgb(71pt)=(0.1806,0.5191,0.9738); rgb(72pt)=(0.1795,0.5244,0.9709); rgb(73pt)=(0.1785,0.5296,0.9677); rgb(74pt)=(0.1778,0.5349,0.9641); rgb(75pt)=(0.1773,0.5401,0.9602); rgb(76pt)=(0.1768,0.5452,0.956); rgb(77pt)=(0.1764,0.5504,0.9516); rgb(78pt)=(0.1755,0.5554,0.9473); rgb(79pt)=(0.174,0.5605,0.9432); rgb(80pt)=(0.1716,0.5655,0.9393); rgb(81pt)=(0.1686,0.5705,0.9357); rgb(82pt)=(0.1649,0.5755,0.9323); rgb(83pt)=(0.161,0.5805,0.9289); rgb(84pt)=(0.1573,0.5854,0.9254); rgb(85pt)=(0.154,0.5902,0.9218); rgb(86pt)=(0.1513,0.595,0.9182); rgb(87pt)=(0.1492,0.5997,0.9147); rgb(88pt)=(0.1475,0.6043,0.9113); rgb(89pt)=(0.1461,0.6089,0.908); rgb(90pt)=(0.1446,0.6135,0.905); rgb(91pt)=(0.1429,0.618,0.9022); rgb(92pt)=(0.1408,0.6226,0.8998); rgb(93pt)=(0.1383,0.6272,0.8975); rgb(94pt)=(0.1354,0.6317,0.8953); rgb(95pt)=(0.1321,0.6363,0.8932); rgb(96pt)=(0.1288,0.6408,0.891); rgb(97pt)=(0.1253,0.6453,0.8887); rgb(98pt)=(0.1219,0.6497,0.8862); rgb(99pt)=(0.1185,0.6541,0.8834); rgb(100pt)=(0.1152,0.6584,0.8804); rgb(101pt)=(0.1119,0.6627,0.877); rgb(102pt)=(0.1085,0.6669,0.8734); rgb(103pt)=(0.1048,0.671,0.8695); rgb(104pt)=(0.1009,0.675,0.8653); rgb(105pt)=(0.0964,0.6789,0.8609); rgb(106pt)=(0.0914,0.6828,0.8562); rgb(107pt)=(0.0855,0.6865,0.8513); rgb(108pt)=(0.0789,0.6902,0.8462); rgb(109pt)=(0.0713,0.6938,0.8409); rgb(110pt)=(0.0628,0.6972,0.8355); rgb(111pt)=(0.0535,0.7006,0.8299); rgb(112pt)=(0.0433,0.7039,0.8242); rgb(113pt)=(0.0328,0.7071,0.8183); rgb(114pt)=(0.0234,0.7103,0.8124); rgb(115pt)=(0.0155,0.7133,0.8064); rgb(116pt)=(0.0091,0.7163,0.8003); rgb(117pt)=(0.0046,0.7192,0.7941); rgb(118pt)=(0.0019,0.722,0.7878); rgb(119pt)=(0.0009,0.7248,0.7815); rgb(120pt)=(0.0018,0.7275,0.7752); rgb(121pt)=(0.0046,0.7301,0.7688); rgb(122pt)=(0.0094,0.7327,0.7623); rgb(123pt)=(0.0162,0.7352,0.7558); rgb(124pt)=(0.0253,0.7376,0.7492); rgb(125pt)=(0.0369,0.74,0.7426); rgb(126pt)=(0.0504,0.7423,0.7359); rgb(127pt)=(0.0638,0.7446,0.7292); rgb(128pt)=(0.077,0.7468,0.7224); rgb(129pt)=(0.0899,0.7489,0.7156); rgb(130pt)=(0.1023,0.751,0.7088); rgb(131pt)=(0.1141,0.7531,0.7019); rgb(132pt)=(0.1252,0.7552,0.695); rgb(133pt)=(0.1354,0.7572,0.6881); rgb(134pt)=(0.1448,0.7593,0.6812); rgb(135pt)=(0.1532,0.7614,0.6741); rgb(136pt)=(0.1609,0.7635,0.6671); rgb(137pt)=(0.1678,0.7656,0.6599); rgb(138pt)=(0.1741,0.7678,0.6527); rgb(139pt)=(0.1799,0.7699,0.6454); rgb(140pt)=(0.1853,0.7721,0.6379); rgb(141pt)=(0.1905,0.7743,0.6303); rgb(142pt)=(0.1954,0.7765,0.6225); rgb(143pt)=(0.2003,0.7787,0.6146); rgb(144pt)=(0.2061,0.7808,0.6065); rgb(145pt)=(0.2118,0.7828,0.5983); rgb(146pt)=(0.2178,0.7849,0.5899); rgb(147pt)=(0.2244,0.7869,0.5813); rgb(148pt)=(0.2318,0.7887,0.5725); rgb(149pt)=(0.2401,0.7905,0.5636); rgb(150pt)=(0.2491,0.7922,0.5546); rgb(151pt)=(0.2589,0.7937,0.5454); rgb(152pt)=(0.2695,0.7951,0.536); rgb(153pt)=(0.2809,0.7964,0.5266); rgb(154pt)=(0.2929,0.7975,0.517); rgb(155pt)=(0.3052,0.7985,0.5074); rgb(156pt)=(0.3176,0.7994,0.4975); rgb(157pt)=(0.3301,0.8002,0.4876); rgb(158pt)=(0.3424,0.8009,0.4774); rgb(159pt)=(0.3548,0.8016,0.4669); rgb(160pt)=(0.3671,0.8021,0.4563); rgb(161pt)=(0.3795,0.8026,0.4454); rgb(162pt)=(0.3921,0.8029,0.4344); rgb(163pt)=(0.405,0.8031,0.4233); rgb(164pt)=(0.4184,0.803,0.4122); rgb(165pt)=(0.4322,0.8028,0.4013); rgb(166pt)=(0.4463,0.8024,0.3904); rgb(167pt)=(0.4608,0.8018,0.3797); rgb(168pt)=(0.4753,0.8011,0.3691); rgb(169pt)=(0.4899,0.8002,0.3586); rgb(170pt)=(0.5044,0.7993,0.348); rgb(171pt)=(0.5187,0.7982,0.3374); rgb(172pt)=(0.5329,0.797,0.3267); rgb(173pt)=(0.547,0.7957,0.3159); rgb(175pt)=(0.5748,0.7929,0.2941); rgb(176pt)=(0.5886,0.7913,0.2833); rgb(177pt)=(0.6024,0.7896,0.2726); rgb(178pt)=(0.6161,0.7878,0.2622); rgb(179pt)=(0.6297,0.7859,0.2521); rgb(180pt)=(0.6433,0.7839,0.2423); rgb(181pt)=(0.6567,0.7818,0.2329); rgb(182pt)=(0.6701,0.7796,0.2239); rgb(183pt)=(0.6833,0.7773,0.2155); rgb(184pt)=(0.6963,0.775,0.2075); rgb(185pt)=(0.7091,0.7727,0.1998); rgb(186pt)=(0.7218,0.7703,0.1924); rgb(187pt)=(0.7344,0.7679,0.1852); rgb(188pt)=(0.7468,0.7654,0.1782); rgb(189pt)=(0.759,0.7629,0.1717); rgb(190pt)=(0.771,0.7604,0.1658); rgb(191pt)=(0.7829,0.7579,0.1608); rgb(192pt)=(0.7945,0.7554,0.157); rgb(193pt)=(0.806,0.7529,0.1546); rgb(194pt)=(0.8172,0.7505,0.1535); rgb(195pt)=(0.8281,0.7481,0.1536); rgb(196pt)=(0.8389,0.7457,0.1546); rgb(197pt)=(0.8495,0.7435,0.1564); rgb(198pt)=(0.86,0.7413,0.1587); rgb(199pt)=(0.8703,0.7392,0.1615); rgb(200pt)=(0.8804,0.7372,0.165); rgb(201pt)=(0.8903,0.7353,0.1695); rgb(202pt)=(0.9,0.7336,0.1749); rgb(203pt)=(0.9093,0.7321,0.1815); rgb(204pt)=(0.9184,0.7308,0.189); rgb(205pt)=(0.9272,0.7298,0.1973); rgb(206pt)=(0.9357,0.729,0.2061); rgb(207pt)=(0.944,0.7285,0.2151); rgb(208pt)=(0.9523,0.7284,0.2237); rgb(209pt)=(0.9606,0.7285,0.2312); rgb(210pt)=(0.9689,0.7292,0.2373); rgb(211pt)=(0.977,0.7304,0.2418); rgb(212pt)=(0.9842,0.733,0.2446); rgb(213pt)=(0.99,0.7365,0.2429); rgb(214pt)=(0.9946,0.7407,0.2394); rgb(215pt)=(0.9966,0.7458,0.2351); rgb(216pt)=(0.9971,0.7513,0.2309); rgb(217pt)=(0.9972,0.7569,0.2267); rgb(218pt)=(0.9971,0.7626,0.2224); rgb(219pt)=(0.9969,0.7683,0.2181); rgb(220pt)=(0.9966,0.774,0.2138); rgb(221pt)=(0.9962,0.7798,0.2095); rgb(222pt)=(0.9957,0.7856,0.2053); rgb(223pt)=(0.9949,0.7915,0.2012); rgb(224pt)=(0.9938,0.7974,0.1974); rgb(225pt)=(0.9923,0.8034,0.1939); rgb(226pt)=(0.9906,0.8095,0.1906); rgb(227pt)=(0.9885,0.8156,0.1875); rgb(228pt)=(0.9861,0.8218,0.1846); rgb(229pt)=(0.9835,0.828,0.1817); rgb(230pt)=(0.9807,0.8342,0.1787); rgb(231pt)=(0.9778,0.8404,0.1757); rgb(232pt)=(0.9748,0.8467,0.1726); rgb(233pt)=(0.972,0.8529,0.1695); rgb(234pt)=(0.9694,0.8591,0.1665); rgb(235pt)=(0.9671,0.8654,0.1636); rgb(236pt)=(0.9651,0.8716,0.1608); rgb(237pt)=(0.9634,0.8778,0.1582); rgb(238pt)=(0.9619,0.884,0.1557); rgb(239pt)=(0.9608,0.8902,0.1532); rgb(240pt)=(0.9601,0.8963,0.1507); rgb(241pt)=(0.9596,0.9023,0.148); rgb(242pt)=(0.9595,0.9084,0.145); rgb(243pt)=(0.9597,0.9143,0.1418); rgb(244pt)=(0.9601,0.9203,0.1382); rgb(245pt)=(0.9608,0.9262,0.1344); rgb(246pt)=(0.9618,0.932,0.1304); rgb(247pt)=(0.9629,0.9379,0.1261); rgb(248pt)=(0.9642,0.9437,0.1216); rgb(249pt)=(0.9657,0.9494,0.1168); rgb(250pt)=(0.9674,0.9552,0.1116); rgb(251pt)=(0.9692,0.9609,0.1061); rgb(252pt)=(0.9711,0.9667,0.1001); rgb(253pt)=(0.973,0.9724,0.0938); rgb(254pt)=(0.9749,0.9782,0.0872); rgb(255pt)=(0.9769,0.9839,0.0805)},
colorbar,
colorbar style={
    ylabel={\scalebox{0.8}{$\log f(\observationn{n}|\pos{n}) - \max\big(\log f(\observationn{n}|\pos{n})\big)$}},
    axis line style = thick,	
    line cap = round,
    line join = round,
    width=0.3cm,%
    tick align=inside, %
    ytick={0,-1,-2,-3,-4,-5,-6},
    yticklabel style={/pgf/number format/fixed},
    title={$\,\times 10^2$},
    xshift = -0.2cm, %
    title style={font=\footnotesize, yshift=-1.5mm},
    },
]

\addplot [forget plot] graphics [xmin=-3.50535343675, xmax=4.50338794125, ymin=-1.20535343675, ymax=4.50141013875] {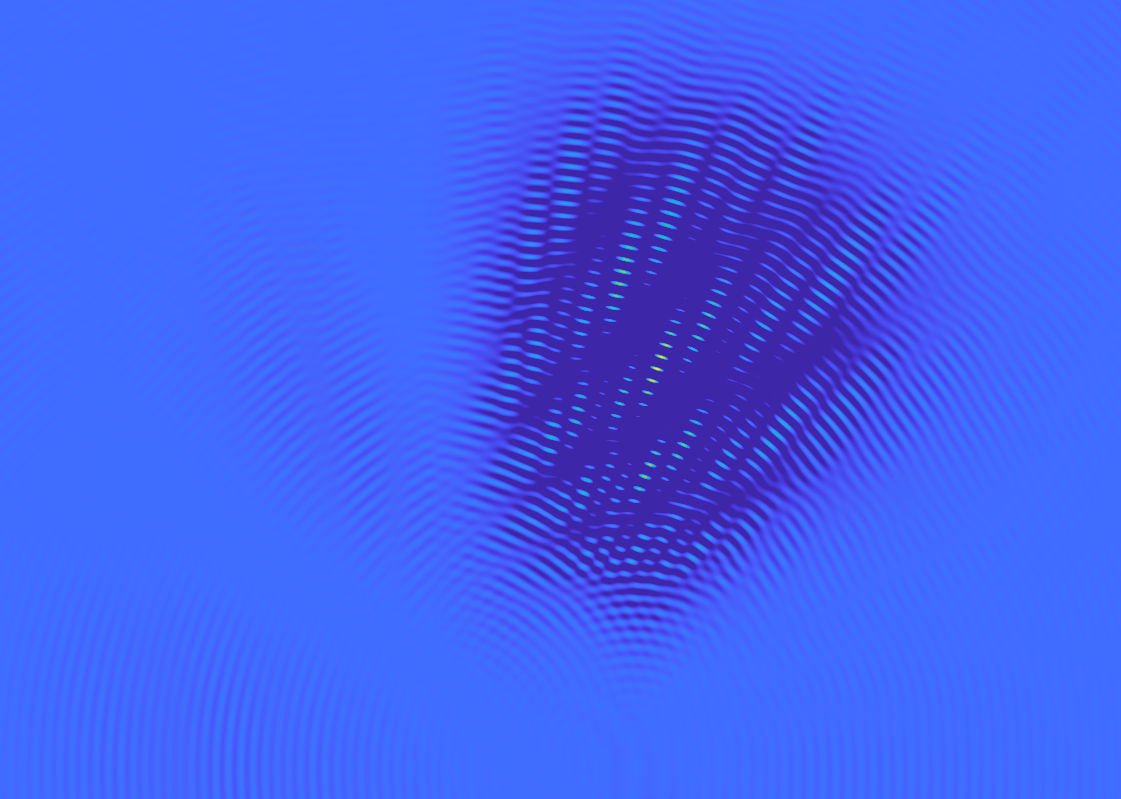};

\input{\datapath/trajectory.tex}

\addplot [color=blue, line width=0.5pt, only marks, mark size=0.9pt, mark=o, mark options={solid, IEEEred}, forget plot]
  table[row sep=crcr]{%
1.18592964824121	1.8675879396984\\
};

\input{\datapath/arrays2D.tex}

\end{axis}

\node[above left, align=right,font={\footnotesize},fill=white,
opacity=0.9,
inner sep=0.75mm, 
xshift=0.65mm, 
yshift=-0.65mm, 
draw, 
line width = 0.8pt,
minimum width = 3mm,
minimum height = 3mm,
] at %
(boundaryRight.south east)
{c)};

\node[
    below left,
    align=right,
    shape=rectangle,
    draw=black,
    fill=gray!20!white,
    rounded corners=0.5pt,
    inner sep=1pt,
    minimum size=1.0em,
    text=black
] at (boundaryRight.north east)
{\adjustbox{max height=0.65em}{CP}};

\begin{axis}[%
name=boundaryLeftBelow,  
width=0.878\figurewidth,
height=0.625\figurewidth,
at={(0\figurewidth,-0.625\figurewidth-1mm)},
axis line style = thick,	
line cap = round,	
line join = round,	
scale only axis,
point meta min=-5,
point meta max=0,
axis on top,
xmin=-3.5,
xmax=4.487327631,
xlabel style={font=\color{white!15!black}},
xlabel={$x$ in \SI{}{\metre}},
ylabel={$y$ in \SI{}{\metre}},
ylabel style={yshift=-1.5mm},
ymin=-1.1,
ymax=4.481919232,
tick align=inside,
xtick = {-3,-2,...,4},
ytick = {-1,0,...,4},
    minor x tick num = 1, 
    minor y tick num = 1, 
]

\addplot [forget plot] graphics [xmin=-3.50535343675, xmax=4.50338794125, ymin=-1.20535343675, ymax=4.50141013875] {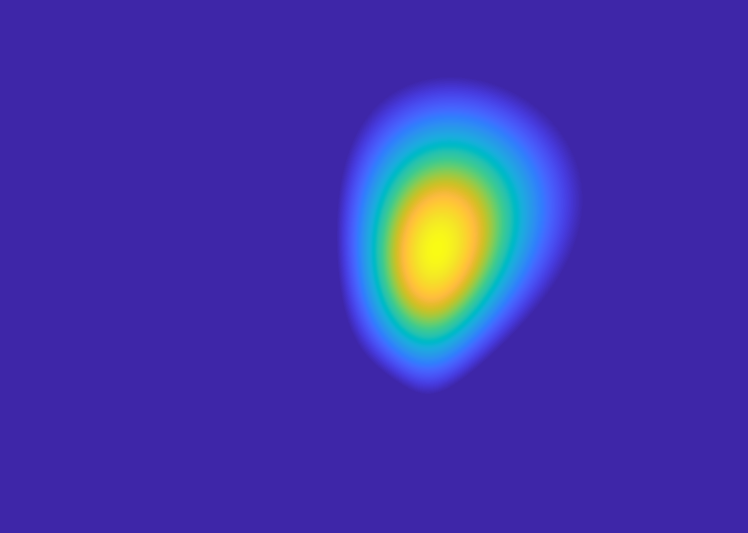};

\input{\datapath/trajectory.tex}

\addplot [color=blue, line width=0.5pt, only marks, mark size=0.9pt, mark=o, mark options={solid, IEEEred}, forget plot]
  table[row sep=crcr]{%
1.18592964824121	1.8675879396984\\
};

\input{\datapath/arrays2D.tex}
\end{axis}

\node[above left, align=right,font={\footnotesize},fill=white,
opacity=0.9,
inner sep=0.75mm, 
xshift=0.65mm, 
yshift=-0.65mm, 
draw, 
line width = 0.8pt,
minimum width = 3mm,
minimum height = 3mm,
] at %
(boundaryLeftBelow.south east)
{d)};

\node[
    below left,
    align=right,
    shape=rectangle,
    draw=black,
    fill=gray!20!white,
    rounded corners=0.5pt,
    inner sep=1pt,
    minimum size=1.0em,
    text=black
] at (boundaryLeftBelow.north east)
{\adjustbox{max height=0.65em}{ZM}};

\begin{axis}[%
name=boundaryCenterBelow,  
width=0.878\figurewidth,
height=0.625\figurewidth,
at={(0.9\figurewidth,-0.625\figurewidth-1mm)},
axis line style = thick,	
line cap = round,	
line join = round,	
scale only axis,
point meta min=-5,
point meta max=0,
axis on top,
xmin=-3.5,
xmax=4.487327631,
xlabel style={font=\color{white!15!black}},
xlabel={$x$ in \SI{}{\metre}},
yticklabel=\empty,                      
ymin=-1.1,
ymax=4.481919232,
tick align=inside,
xtick = {-3,-2,...,4},
ytick = {-1,0,...,4},
    minor x tick num = 1, 
    minor y tick num = 1, 
]

\addplot [forget plot] graphics [xmin=-3.50535343675, xmax=4.50338794125, ymin=-1.20535343675, ymax=4.50141013875] {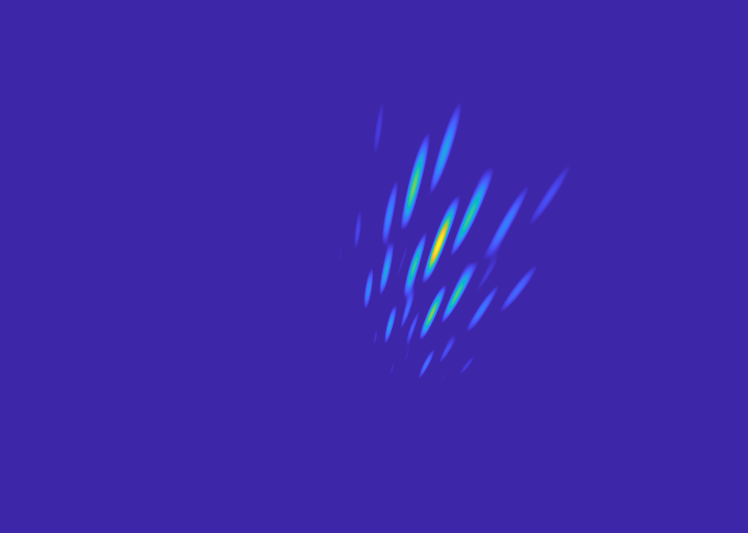};

\input{\datapath/trajectory.tex}

\addplot [color=blue, line width=0.5pt, only marks, mark size=0.9pt, mark=o, mark options={solid, IEEEred}, forget plot]
  table[row sep=crcr]{%
1.18592964824121	1.8675879396984\\
};

\input{\datapath/arrays2D.tex}

\end{axis}

\node[above left, align=right,font={\footnotesize},fill=white,
opacity=0.9,
inner sep=0.75mm, 
xshift=0.65mm, 
yshift=-0.65mm, 
draw, 
line width = 0.8pt,
minimum width = 3mm,
minimum height = 3mm,
] at %
(boundaryCenterBelow.south east)
{e)};

\node[
    below left,
    align=right,
    shape=rectangle,
    draw=black,
    fill=gray!20!white,
    rounded corners=0.5pt,
    inner sep=1pt,
    minimum size=1.0em,
    text=black
] at (boundaryCenterBelow.north east)
{\adjustbox{max height=0.65em}{ZM-S}};

\begin{axis}[%
name=boundaryRightBelow,  
width=0.878\figurewidth,
height=0.625\figurewidth,
at={(1.8\figurewidth,-0.625\figurewidth-1mm)},
axis line style = thick,	
line cap = round,	
line join = round,	
scale only axis,
point meta min=-7,
point meta max=0,
axis on top,
xmin=-3.5,
xmax=4.487327631,
xlabel style={font=\color{white!15!black}},
xlabel={$x$ in \SI{}{\metre}},
yticklabel=\empty,                      
ymin=-1.1,
ymax=4.481919232,
tick align=inside,
xtick = {-3,-2,...,4},
ytick = {-1,0,...,4},
    minor x tick num = 1, 
    minor y tick num = 1, 
colormap={mymap}{[1pt] rgb(0pt)=(0.2422,0.1504,0.6603); rgb(1pt)=(0.2444,0.1534,0.6728); rgb(2pt)=(0.2464,0.1569,0.6847); rgb(3pt)=(0.2484,0.1607,0.6961); rgb(4pt)=(0.2503,0.1648,0.7071); rgb(5pt)=(0.2522,0.1689,0.7179); rgb(6pt)=(0.254,0.1732,0.7286); rgb(7pt)=(0.2558,0.1773,0.7393); rgb(8pt)=(0.2576,0.1814,0.7501); rgb(9pt)=(0.2594,0.1854,0.761); rgb(11pt)=(0.2628,0.1932,0.7828); rgb(12pt)=(0.2645,0.1972,0.7937); rgb(13pt)=(0.2661,0.2011,0.8043); rgb(14pt)=(0.2676,0.2052,0.8148); rgb(15pt)=(0.2691,0.2094,0.8249); rgb(16pt)=(0.2704,0.2138,0.8346); rgb(17pt)=(0.2717,0.2184,0.8439); rgb(18pt)=(0.2729,0.2231,0.8528); rgb(19pt)=(0.274,0.228,0.8612); rgb(20pt)=(0.2749,0.233,0.8692); rgb(21pt)=(0.2758,0.2382,0.8767); rgb(22pt)=(0.2766,0.2435,0.884); rgb(23pt)=(0.2774,0.2489,0.8908); rgb(24pt)=(0.2781,0.2543,0.8973); rgb(25pt)=(0.2788,0.2598,0.9035); rgb(26pt)=(0.2794,0.2653,0.9094); rgb(27pt)=(0.2798,0.2708,0.915); rgb(28pt)=(0.2802,0.2764,0.9204); rgb(29pt)=(0.2806,0.2819,0.9255); rgb(30pt)=(0.2809,0.2875,0.9305); rgb(31pt)=(0.2811,0.293,0.9352); rgb(32pt)=(0.2813,0.2985,0.9397); rgb(33pt)=(0.2814,0.304,0.9441); rgb(34pt)=(0.2814,0.3095,0.9483); rgb(35pt)=(0.2813,0.315,0.9524); rgb(36pt)=(0.2811,0.3204,0.9563); rgb(37pt)=(0.2809,0.3259,0.96); rgb(38pt)=(0.2807,0.3313,0.9636); rgb(39pt)=(0.2803,0.3367,0.967); rgb(40pt)=(0.2798,0.3421,0.9702); rgb(41pt)=(0.2791,0.3475,0.9733); rgb(42pt)=(0.2784,0.3529,0.9763); rgb(43pt)=(0.2776,0.3583,0.9791); rgb(44pt)=(0.2766,0.3638,0.9817); rgb(45pt)=(0.2754,0.3693,0.984); rgb(46pt)=(0.2741,0.3748,0.9862); rgb(47pt)=(0.2726,0.3804,0.9881); rgb(48pt)=(0.271,0.386,0.9898); rgb(49pt)=(0.2691,0.3916,0.9912); rgb(50pt)=(0.267,0.3973,0.9924); rgb(51pt)=(0.2647,0.403,0.9935); rgb(52pt)=(0.2621,0.4088,0.9946); rgb(53pt)=(0.2591,0.4145,0.9955); rgb(54pt)=(0.2556,0.4203,0.9965); rgb(55pt)=(0.2517,0.4261,0.9974); rgb(56pt)=(0.2473,0.4319,0.9983); rgb(57pt)=(0.2424,0.4378,0.9991); rgb(58pt)=(0.2369,0.4437,0.9996); rgb(59pt)=(0.2311,0.4497,0.9995); rgb(60pt)=(0.225,0.4559,0.9985); rgb(61pt)=(0.2189,0.462,0.9968); rgb(62pt)=(0.2128,0.4682,0.9948); rgb(63pt)=(0.2066,0.4743,0.9926); rgb(64pt)=(0.2006,0.4803,0.9906); rgb(65pt)=(0.195,0.4861,0.9887); rgb(66pt)=(0.1903,0.4919,0.9867); rgb(67pt)=(0.1869,0.4975,0.9844); rgb(68pt)=(0.1847,0.503,0.9819); rgb(69pt)=(0.1831,0.5084,0.9793); rgb(70pt)=(0.1818,0.5138,0.9766); rgb(71pt)=(0.1806,0.5191,0.9738); rgb(72pt)=(0.1795,0.5244,0.9709); rgb(73pt)=(0.1785,0.5296,0.9677); rgb(74pt)=(0.1778,0.5349,0.9641); rgb(75pt)=(0.1773,0.5401,0.9602); rgb(76pt)=(0.1768,0.5452,0.956); rgb(77pt)=(0.1764,0.5504,0.9516); rgb(78pt)=(0.1755,0.5554,0.9473); rgb(79pt)=(0.174,0.5605,0.9432); rgb(80pt)=(0.1716,0.5655,0.9393); rgb(81pt)=(0.1686,0.5705,0.9357); rgb(82pt)=(0.1649,0.5755,0.9323); rgb(83pt)=(0.161,0.5805,0.9289); rgb(84pt)=(0.1573,0.5854,0.9254); rgb(85pt)=(0.154,0.5902,0.9218); rgb(86pt)=(0.1513,0.595,0.9182); rgb(87pt)=(0.1492,0.5997,0.9147); rgb(88pt)=(0.1475,0.6043,0.9113); rgb(89pt)=(0.1461,0.6089,0.908); rgb(90pt)=(0.1446,0.6135,0.905); rgb(91pt)=(0.1429,0.618,0.9022); rgb(92pt)=(0.1408,0.6226,0.8998); rgb(93pt)=(0.1383,0.6272,0.8975); rgb(94pt)=(0.1354,0.6317,0.8953); rgb(95pt)=(0.1321,0.6363,0.8932); rgb(96pt)=(0.1288,0.6408,0.891); rgb(97pt)=(0.1253,0.6453,0.8887); rgb(98pt)=(0.1219,0.6497,0.8862); rgb(99pt)=(0.1185,0.6541,0.8834); rgb(100pt)=(0.1152,0.6584,0.8804); rgb(101pt)=(0.1119,0.6627,0.877); rgb(102pt)=(0.1085,0.6669,0.8734); rgb(103pt)=(0.1048,0.671,0.8695); rgb(104pt)=(0.1009,0.675,0.8653); rgb(105pt)=(0.0964,0.6789,0.8609); rgb(106pt)=(0.0914,0.6828,0.8562); rgb(107pt)=(0.0855,0.6865,0.8513); rgb(108pt)=(0.0789,0.6902,0.8462); rgb(109pt)=(0.0713,0.6938,0.8409); rgb(110pt)=(0.0628,0.6972,0.8355); rgb(111pt)=(0.0535,0.7006,0.8299); rgb(112pt)=(0.0433,0.7039,0.8242); rgb(113pt)=(0.0328,0.7071,0.8183); rgb(114pt)=(0.0234,0.7103,0.8124); rgb(115pt)=(0.0155,0.7133,0.8064); rgb(116pt)=(0.0091,0.7163,0.8003); rgb(117pt)=(0.0046,0.7192,0.7941); rgb(118pt)=(0.0019,0.722,0.7878); rgb(119pt)=(0.0009,0.7248,0.7815); rgb(120pt)=(0.0018,0.7275,0.7752); rgb(121pt)=(0.0046,0.7301,0.7688); rgb(122pt)=(0.0094,0.7327,0.7623); rgb(123pt)=(0.0162,0.7352,0.7558); rgb(124pt)=(0.0253,0.7376,0.7492); rgb(125pt)=(0.0369,0.74,0.7426); rgb(126pt)=(0.0504,0.7423,0.7359); rgb(127pt)=(0.0638,0.7446,0.7292); rgb(128pt)=(0.077,0.7468,0.7224); rgb(129pt)=(0.0899,0.7489,0.7156); rgb(130pt)=(0.1023,0.751,0.7088); rgb(131pt)=(0.1141,0.7531,0.7019); rgb(132pt)=(0.1252,0.7552,0.695); rgb(133pt)=(0.1354,0.7572,0.6881); rgb(134pt)=(0.1448,0.7593,0.6812); rgb(135pt)=(0.1532,0.7614,0.6741); rgb(136pt)=(0.1609,0.7635,0.6671); rgb(137pt)=(0.1678,0.7656,0.6599); rgb(138pt)=(0.1741,0.7678,0.6527); rgb(139pt)=(0.1799,0.7699,0.6454); rgb(140pt)=(0.1853,0.7721,0.6379); rgb(141pt)=(0.1905,0.7743,0.6303); rgb(142pt)=(0.1954,0.7765,0.6225); rgb(143pt)=(0.2003,0.7787,0.6146); rgb(144pt)=(0.2061,0.7808,0.6065); rgb(145pt)=(0.2118,0.7828,0.5983); rgb(146pt)=(0.2178,0.7849,0.5899); rgb(147pt)=(0.2244,0.7869,0.5813); rgb(148pt)=(0.2318,0.7887,0.5725); rgb(149pt)=(0.2401,0.7905,0.5636); rgb(150pt)=(0.2491,0.7922,0.5546); rgb(151pt)=(0.2589,0.7937,0.5454); rgb(152pt)=(0.2695,0.7951,0.536); rgb(153pt)=(0.2809,0.7964,0.5266); rgb(154pt)=(0.2929,0.7975,0.517); rgb(155pt)=(0.3052,0.7985,0.5074); rgb(156pt)=(0.3176,0.7994,0.4975); rgb(157pt)=(0.3301,0.8002,0.4876); rgb(158pt)=(0.3424,0.8009,0.4774); rgb(159pt)=(0.3548,0.8016,0.4669); rgb(160pt)=(0.3671,0.8021,0.4563); rgb(161pt)=(0.3795,0.8026,0.4454); rgb(162pt)=(0.3921,0.8029,0.4344); rgb(163pt)=(0.405,0.8031,0.4233); rgb(164pt)=(0.4184,0.803,0.4122); rgb(165pt)=(0.4322,0.8028,0.4013); rgb(166pt)=(0.4463,0.8024,0.3904); rgb(167pt)=(0.4608,0.8018,0.3797); rgb(168pt)=(0.4753,0.8011,0.3691); rgb(169pt)=(0.4899,0.8002,0.3586); rgb(170pt)=(0.5044,0.7993,0.348); rgb(171pt)=(0.5187,0.7982,0.3374); rgb(172pt)=(0.5329,0.797,0.3267); rgb(173pt)=(0.547,0.7957,0.3159); rgb(175pt)=(0.5748,0.7929,0.2941); rgb(176pt)=(0.5886,0.7913,0.2833); rgb(177pt)=(0.6024,0.7896,0.2726); rgb(178pt)=(0.6161,0.7878,0.2622); rgb(179pt)=(0.6297,0.7859,0.2521); rgb(180pt)=(0.6433,0.7839,0.2423); rgb(181pt)=(0.6567,0.7818,0.2329); rgb(182pt)=(0.6701,0.7796,0.2239); rgb(183pt)=(0.6833,0.7773,0.2155); rgb(184pt)=(0.6963,0.775,0.2075); rgb(185pt)=(0.7091,0.7727,0.1998); rgb(186pt)=(0.7218,0.7703,0.1924); rgb(187pt)=(0.7344,0.7679,0.1852); rgb(188pt)=(0.7468,0.7654,0.1782); rgb(189pt)=(0.759,0.7629,0.1717); rgb(190pt)=(0.771,0.7604,0.1658); rgb(191pt)=(0.7829,0.7579,0.1608); rgb(192pt)=(0.7945,0.7554,0.157); rgb(193pt)=(0.806,0.7529,0.1546); rgb(194pt)=(0.8172,0.7505,0.1535); rgb(195pt)=(0.8281,0.7481,0.1536); rgb(196pt)=(0.8389,0.7457,0.1546); rgb(197pt)=(0.8495,0.7435,0.1564); rgb(198pt)=(0.86,0.7413,0.1587); rgb(199pt)=(0.8703,0.7392,0.1615); rgb(200pt)=(0.8804,0.7372,0.165); rgb(201pt)=(0.8903,0.7353,0.1695); rgb(202pt)=(0.9,0.7336,0.1749); rgb(203pt)=(0.9093,0.7321,0.1815); rgb(204pt)=(0.9184,0.7308,0.189); rgb(205pt)=(0.9272,0.7298,0.1973); rgb(206pt)=(0.9357,0.729,0.2061); rgb(207pt)=(0.944,0.7285,0.2151); rgb(208pt)=(0.9523,0.7284,0.2237); rgb(209pt)=(0.9606,0.7285,0.2312); rgb(210pt)=(0.9689,0.7292,0.2373); rgb(211pt)=(0.977,0.7304,0.2418); rgb(212pt)=(0.9842,0.733,0.2446); rgb(213pt)=(0.99,0.7365,0.2429); rgb(214pt)=(0.9946,0.7407,0.2394); rgb(215pt)=(0.9966,0.7458,0.2351); rgb(216pt)=(0.9971,0.7513,0.2309); rgb(217pt)=(0.9972,0.7569,0.2267); rgb(218pt)=(0.9971,0.7626,0.2224); rgb(219pt)=(0.9969,0.7683,0.2181); rgb(220pt)=(0.9966,0.774,0.2138); rgb(221pt)=(0.9962,0.7798,0.2095); rgb(222pt)=(0.9957,0.7856,0.2053); rgb(223pt)=(0.9949,0.7915,0.2012); rgb(224pt)=(0.9938,0.7974,0.1974); rgb(225pt)=(0.9923,0.8034,0.1939); rgb(226pt)=(0.9906,0.8095,0.1906); rgb(227pt)=(0.9885,0.8156,0.1875); rgb(228pt)=(0.9861,0.8218,0.1846); rgb(229pt)=(0.9835,0.828,0.1817); rgb(230pt)=(0.9807,0.8342,0.1787); rgb(231pt)=(0.9778,0.8404,0.1757); rgb(232pt)=(0.9748,0.8467,0.1726); rgb(233pt)=(0.972,0.8529,0.1695); rgb(234pt)=(0.9694,0.8591,0.1665); rgb(235pt)=(0.9671,0.8654,0.1636); rgb(236pt)=(0.9651,0.8716,0.1608); rgb(237pt)=(0.9634,0.8778,0.1582); rgb(238pt)=(0.9619,0.884,0.1557); rgb(239pt)=(0.9608,0.8902,0.1532); rgb(240pt)=(0.9601,0.8963,0.1507); rgb(241pt)=(0.9596,0.9023,0.148); rgb(242pt)=(0.9595,0.9084,0.145); rgb(243pt)=(0.9597,0.9143,0.1418); rgb(244pt)=(0.9601,0.9203,0.1382); rgb(245pt)=(0.9608,0.9262,0.1344); rgb(246pt)=(0.9618,0.932,0.1304); rgb(247pt)=(0.9629,0.9379,0.1261); rgb(248pt)=(0.9642,0.9437,0.1216); rgb(249pt)=(0.9657,0.9494,0.1168); rgb(250pt)=(0.9674,0.9552,0.1116); rgb(251pt)=(0.9692,0.9609,0.1061); rgb(252pt)=(0.9711,0.9667,0.1001); rgb(253pt)=(0.973,0.9724,0.0938); rgb(254pt)=(0.9749,0.9782,0.0872); rgb(255pt)=(0.9769,0.9839,0.0805)},
colorbar,
colorbar style={
    ylabel={\scalebox{0.8}{$\log f(\observationn{n}|\pos{n}) - \max\big(\log f(\observationn{n}|\pos{n})\big)$}},
    axis line style = thick,	
    line cap = round,
    line join = round,
    width=0.3cm,%
    tick align=inside, %
    ytick={0,-1,-2,-3,-4,-5,-6,-7},
    yticklabel style={/pgf/number format/fixed},
    xshift = -0.2cm, %
    title style={font=\footnotesize, yshift=-1.5mm},
    },
]

\addplot [forget plot] graphics [xmin=-3.50535343675, xmax=4.50338794125, ymin=-1.20535343675, ymax=4.50141013875] {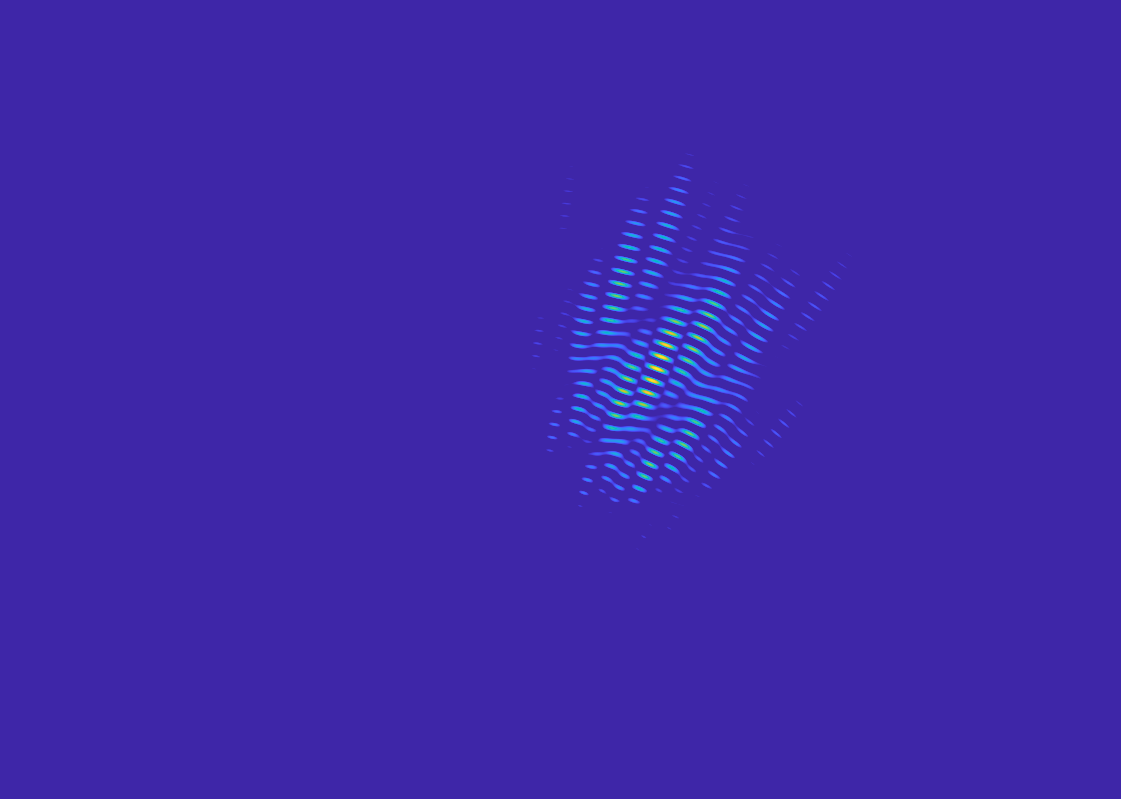};

\input{\datapath/trajectory.tex}

\addplot [color=blue, line width=0.5pt, only marks, mark size=0.9pt, mark=o, mark options={solid, IEEEred}, forget plot]
  table[row sep=crcr]{%
1.18592964824121	1.8675879396984\\
};

\input{\datapath/arrays2D.tex}

\end{axis}

\node[above left, align=right,font={\footnotesize},fill=white,
opacity=0.9,
inner sep=0.75mm, 
xshift=0.65mm, 
yshift=-0.65mm, 
draw, 
line width = 0.8pt,
minimum width = 3mm,
minimum height = 3mm,
] at %
(boundaryRightBelow.south east)
{f)};

\node[
    below left,
    align=right,
    shape=rectangle,
    draw=black,
    fill=gray!20!white,
    rounded corners=0.5pt,
    inner sep=1pt,
    minimum size=1.0em,
    text=black
] at (boundaryRightBelow.north east)
{\adjustbox{max height=0.65em}{NZM}};

\end{tikzpicture}%

%% file: figures/results/PEB.tex
%
%

\pgfplotsset{every axis/.append style={
  label style={font=\footnotesize},
  legend style={font=\scriptsize},
  tick label style={font=\footnotesize},
}}

\newcommand{\lblPOSest}[1]{RMSE REF}
\newcommand{\lblPOScrlb}[1]{noncoherent PEB}
\newcommand{\lblPOSestc}[1]{RMSE PROP}
\newcommand{\lblPOScrlbc}[1]{coherent PEB}
\newcommand{\lblPOSuncertainty}[1]{$\interval{50\%}$}

\def\datapath{./figures/data}   

\begin{tikzpicture}

\begin{axis}[%
name=boundaryLeft,  
axis line style = thick,	
width=0.65\figurewidth,
height=\figureheight,
at={(0\figurewidth,0\figureheight)},
scale only axis,
xmin=1,
xmax=200,
xlabel={Step $n$},                     
ymin=0,
ymax=4,
tick align=inside,
ytick = {0,1,...,6},
grid=both, 
    minor grid style={ultra thin, lightgray!20}, 
    minor x tick num = 3, 
    minor y tick num = 9, 
    major grid style={very thin, lightgray!50},
ylabel={
$\lVert \pos{n}  - \posHat{n}\rVert$ in \SI{}{\centi\metre}},
ylabel style={yshift=-0.3mm},
]
\pgfplotsset{every major tick/.append style={-Triangle Cap,thick},thick}

\addplot [color=\colorSnapshot, 
          line cap = round, 
          line join=round, 
          line width=\LWestimates,
          mark=\markSnapshot, 
          mark repeat=\markRepeat, 
          mark size=0.7\markSize,
          mark options={solid, line width=0.3pt,fill=\colorSnapshot},
          ]%
    table{\datapath/position-errors/RMSE-snapshot.dat};%
    \label{pgf:RMSE-snapshot}

\addplot[area legend, draw=none, fill=\colorSnapshot, fill opacity=0.1]%
    table{\datapath/position-errors/interval-snapshot.dat};%
    \label{pgf:interval-snapshot}

\addplot [color=\colorNC, 
          line cap = round, 
          line join=round, 
          line width=\LWestimates,
          mark=\markNC, 
          mark repeat=\markRepeat, 
          mark size=0.7\markSize,
          mark options={solid, line width=0.3pt,fill=\colorNC}]%
    table{\datapath/position-errors/RMSE-Type-I-coh0.dat};%
    \label{pgf:RMSE-Type-I-coh0}

\addplot[area legend, draw=none, fill=\colorNC, fill opacity=0.1]%
    table{\datapath/position-errors/interval-Type-I-coh0.dat};%
    \label{pgf:interval-Type-I-coh0}

\addplot [color=\colorC, 
          line cap = round, 
          line join=round, 
          line width=\LWestimates,
          mark=\markC, 
          mark repeat=\markRepeat, 
          mark size=\markSize,
          mark options={solid, line width=0.3pt,fill=\colorC}]%
    table{\datapath/position-errors/RMSE-Type-I-coh2.dat};%
    \label{pgf:RMSE-Type-I-coh2}

\addplot[area legend, draw=none, fill=\colorC, fill opacity=0.1]%
    table{\datapath/position-errors/interval-Type-I-coh2.dat};%
    \label{pgf:interval-Type-I-coh2}
    
\addplot [color=\colorCP, 
          line cap = round, 
          line join=round, 
          line width=\LWestimates,
          mark=\markCP, 
          mark repeat=\markRepeat, 
          mark size=\markSize,
          mark options={solid, line width=0.3pt,fill=\colorCP}]%
    table{\datapath/position-errors/RMSE-Type-I-coh1.dat};%
    \label{pgf:RMSE-Type-I-coh1}

\addplot[area legend, draw=none, fill=\colorCP, fill opacity=0.1]%
    table{\datapath/position-errors/interval-Type-I-coh1.dat};%
    \label{pgf:interval-Type-I-coh1}

\addplot [color=\colorNCtypeII, 
          line cap = round, 
          line join=round, 
          line width=\LWestimates,
          mark=\markNCtypeII, 
          mark repeat=\markRepeat, 
          mark size=\markSize,
          mark options={solid, line width=0.3pt,fill=\colorNCtypeII},
          dash pattern=on 3pt off 1pt]%
    table{\datapath/position-errors/RMSE-Type-II-coh0.dat};%
    \label{pgf:RMSE-Type-II-coh0}

\addplot[area legend, draw=none, fill=\colorNCtypeII, fill opacity=0.1]%
    table{\datapath/position-errors/interval-Type-II-coh0.dat};%
    \label{pgf:interval-Type-II-coh0}

\addplot [color=\colorNCtypeIIS, 
          line cap = round, 
          line join=round, 
          line width=\LWestimates,
          mark=\markNCtypeIIS, 
          mark repeat=\markRepeat, 
          mark size=\markSize,
          mark options={solid, line width=0.3pt,fill=\colorNCtypeIIS},
          dash pattern=on 3pt off 1pt]%
    table{\datapath/position-errors/RMSE-Type-II-coh0-PA-stacked.dat};%
    \label{pgf:RMSE-Type-II-coh0-S}

\addplot[area legend, draw=none, fill=\colorNCtypeIIS, fill opacity=0.1]%
    table{\datapath/position-errors/interval-Type-II-coh0-PA-stacked.dat};%
    \label{pgf:interval-Type-II-coh0-S}

\addplot [color=\colorCtypeII, 
          line cap = round, 
          line join=round, 
          line width=\LWestimates,
          mark=\markCtypeII, 
          mark repeat=\markRepeat, 
          mark size=\markSize,
          mark options={solid, line width=0.3pt,fill=\colorCtypeII},
          dash pattern=on 3pt off 1pt]%
    table{\datapath/position-errors/RMSE-Type-II-coh1.dat};%
    \label{pgf:RMSE-Type-II-coh1}

\addplot[area legend, draw=none, fill=\colorCtypeII, fill opacity=0.1]%
    table{\datapath/position-errors/interval-Type-II-coh1.dat};%
    \label{pgf:interval-Type-II-coh1}

\addplot [LineCohNC,
          dash pattern=on 3pt off 1pt]%
    table{\datapath/position-errors/RMSE-Type-II-coh1-nc-data.dat};%
    \label{pgf:RMSE-Type-II-coh1-nc-data}


\addplot [color=\colorNC, 
          line cap = round, 
          line join=round, 
          line width=\LWbound, 
          draw opacity = 0.5]%
    table{\datapath/position-errors/PEB-nc.dat};%
    \label{pgf:PEB-nc}

\addplot [color=\colorC, line cap = round, line join=round, line width=\LWbound, draw opacity = 0.5]%
    table{\datapath/position-errors/PEB-c.dat};%
    \label{pgf:PEB-c}

\addplot [color=\colorCP, line cap = round, line join=round, line width=\LWbound, draw opacity = 0.5]%
    table{\datapath/position-errors/PEB-cp.dat};%
    \label{pgf:PEB-cp}

\addplot [color=\colorSnapshot, 
          line cap = round, 
          line join=round, 
          dashed, 
          line width=\LWbound, 
          draw opacity = 0.5,
          dash pattern=on 3pt off 2pt on 0.1pt off 2pt]%
    table{\datapath/position-errors/PEB-classic.dat};%
    \label{pgf:PEB-classic}

\ifthenelse{\equal{\IEEEversion}{true}}%
{%
}%
{%
\addplot [color=red, 
          only marks,
          mark=*, 
          mark repeat=\markRepeat, 
          mark size=1.3\markSize,
          mark options={solid, line width=0.3pt,fill=red}]%
    table[row sep=crcr]{%
    100	2.881\\
};%
}

\addplot [
  color=black,
  line width=0.25mm,
  line cap = round,
  forget plot,
  -{Stealth[inset=0pt, scale=1.15, angle'=20]}
]
table[row sep=crcr]{%
80	2.4943801561213\\
80	0.918959851613962\\
}
node[
  midway,
  anchor=west,
  xshift=1mm,
  yshift=2mm,
  align=left
] {\footnotesize(i) Prior information from \\[-2pt]\footnotesize~\hspace{1.65mm} Bayesian state filtering};

\addplot [
  color=black,
  line width=0.25mm,
  line cap = round,
  forget plot,
  -{Stealth[inset=0pt, scale=1.15, angle'=20]}
]
table[row sep=crcr]{%
150	1.35363576863761\\
150	0.606038721516405\\
}
node[
  midway,
  anchor=east,
  xshift=-1mm,
  yshift=-1mm,
  align=right
] {\footnotesize(ii) Aperture gain};

\addplot [
  color=black,
  line width=0.25mm,
  line cap = round,
  forget plot,
  -{Stealth[inset=0pt, scale=1.15, angle'=20]}
]
table[row sep=crcr]{%
110	0.384437063433451\\
110	0.103714731840868\\
}
node[
  midway,
  anchor=west,
  xshift=1mm,
  yshift=0.5mm,
  align=left
] {\footnotesize(iii) Carrier-phase information};

\addplot [
  color=black,
  line width=0.25mm,
  line cap = round,
  forget plot
]
table[row sep=crcr]{%
167	1.175\\
165	1.175\\
163	1.1\\
163	0.51\\
};

\node[
  right,
  align = left, 
  xshift=-1mm,
  yshift=0mm,
  align=right
] at (axis cs:168.5,1.2) {\footnotesize Soft coherence};

\addplot [
  color=black,
  line width=0.25mm,
  line cap = round,
  forget plot
]
table[row sep=crcr]{%
167	1.225\\
165	1.225\\
163	1.3\\
163	1.73\\
};

\end{axis}

\begin{axis}[%
name=boundaryRight,  
axis line style = thick,	
width=0.295\figurewidth,
height=\figureheight,
at={(0.7\figurewidth,0\figureheight)},
scale only axis,
xmin=0,
xmax=3.99,
xlabel={$\lVert \pos{n}  - \posHat{n}\rVert$ in \SI{}{\centi\metre}},
ymin=0,
ymax=1,
tick align=inside,
ytick = {0,0.1,...,1},
yticklabels={,0.1,0.2,0.3,0.4,0.5,0.6,0.7,0.8,0.9,1},
yticklabel style={yshift=-0.5mm},
grid=both, 
    minor grid style={ultra thin, lightgray!20}, 
    minor x tick num = 9, 
    minor y tick num = 4, 
    major grid style={very thin, lightgray!50},
ylabel={Cumulative frequency},
ylabel style={yshift=-0.3mm},
]
\pgfplotsset{every major tick/.append style={-Triangle Cap,thick},thick}

\addplot [color=\colorSnapshot,
          line cap = round, 
          line join=round, 
          line width=\LWestimates,
          mark=\markSnapshot, 
          mark repeat=\markRepeat, 
          mark size=0.7\markSize,
          mark options={solid, line width=0.3pt,fill=\colorSnapshot}]%
    table{\datapath/position-errors/CDF-snapshot.dat};%

\addplot [color=\colorNC,
          line cap = round, 
          line join=round, 
          line width=\LWestimates,
          mark=\markNC, 
          mark repeat=\markRepeat, 
          mark size=0.7\markSize,
          mark options={solid, line width=0.3pt,fill=\colorNC}]%
    table{\datapath/position-errors/CDF-Type-I-coh0.dat};%

\addplot [color=\colorC, 
          line cap = round, 
          line join=round, 
          line width=\LWestimates,
          mark=\markC, 
          mark repeat=\markRepeat, 
          mark size=\markSize,
          mark options={solid, line width=0.3pt,fill=\colorC}]%
    table{\datapath/position-errors/CDF-Type-I-coh2.dat};%

\addplot [color=\colorCP,
          line cap = round, 
          line join=round, 
          line width=\LWestimates,
          mark=\markCP, 
          mark repeat=\markRepeat, 
          mark size=0.7\markSize,
          mark options={solid, line width=0.3pt,fill=\colorCP}]%
    table{\datapath/position-errors/CDF-Type-I-coh1.dat};%

\addplot [color=\colorNCtypeII, 
          line cap = round, 
          line join=round, 
          line width=\LWestimates,
          mark=\markNCtypeII, 
          mark repeat=\markRepeat, 
          mark size=\markSize,
          mark options={solid, line width=0.3pt,fill=\colorNCtypeII},
          dash pattern=on 3pt off 1pt]%
    table{\datapath/position-errors/CDF-Type-II-coh0.dat};%

\addplot [color=\colorNCtypeIIS, 
          line cap = round, 
          line join=round, 
          line width=\LWestimates,
          mark=\markNCtypeIIS, 
          mark repeat=\markRepeat, 
          mark size=\markSize,
          mark options={solid, line width=0.3pt,fill=\colorNCtypeIIS},
          dash pattern=on 3pt off 1pt]%
    table{\datapath/position-errors/CDF-Type-II-coh0-PA-stacked.dat};%

\addplot [color=\colorCtypeII, 
          line cap = round, 
          line join=round, 
          line width=\LWestimates,
          mark=\markCtypeII, 
          mark repeat=\markRepeat, 
          mark size=\markSize,
          mark options={solid, line width=0.3pt,fill=\colorCtypeII},
          dash pattern=on 3pt off 1pt]%
    table{\datapath/position-errors/CDF-Type-II-coh1.dat};%

\addplot [LineCohNC,
          dash pattern=on 3pt off 1pt]%
    table{\datapath/position-errors/CDF-Type-II-coh1-nc-data.dat};%

\end{axis}

\node[below,draw,thick,fill=white,inner sep=0pt,above right=0.0em,line cap = round, fill opacity=0.8, text opacity = 1, draw opacity = 1 , yshift = 0.75cm, minimum height=1.35cm, fit = (boundaryLeft.north west) (boundaryRight.north east)] 
{~\vspace{-4.3mm}

\footnotesize
\setlength{\tabcolsep}{4.8pt} 
\begin{tabular}{l|cccc|cccc|} 
    \multirow{2}{*}{Estimator} & \multicolumn{4}{c|}{Type-I} & \multicolumn{4}{c|}{Type-II} \\
    & \estLabel{S} 
    & \estLabel{NC}
    & \estLabel{C} 
    & \estLabel{CP} 
    & \estLabel{ZM} 
    & \estLabel{ZM-S} 
    & \estLabel{NZM}  
    & \estLabel{NZM-NC}  
     \\
     \cline{2-9}
    RMSE 
     & \ref{pgf:RMSE-snapshot}  
     & \ref{pgf:RMSE-Type-I-coh0}
     & \ref{pgf:RMSE-Type-I-coh2}
     & \ref{pgf:RMSE-Type-I-coh1}
     & \ref{pgf:RMSE-Type-II-coh0}
     & \ref{pgf:RMSE-Type-II-coh0-S}
     & \ref{pgf:RMSE-Type-II-coh1}
     & \ref{pgf:RMSE-Type-II-coh1-nc-data}
     \\[-1pt]
    $\interval{50\%}$
    & \ref{pgf:interval-snapshot} 
    & \ref{pgf:interval-Type-I-coh0} 
    & \ref{pgf:interval-Type-I-coh2}
    & \ref{pgf:interval-Type-I-coh1}
    & \ref{pgf:interval-Type-II-coh0}
    & \ref{pgf:interval-Type-II-coh0-S}
    & \ref{pgf:interval-Type-II-coh1}
    & --
\end{tabular}%
\hspace{0.5mm}%
\begin{tabular}{|c|ccc} 
    \multicolumn{4}{|c}{Position Error Bound}
    \\ 
     \, \hspace{2mm}
     Classic & \multicolumn{3}{c}{Bayesian} \\[0.475pt] 
     \cline{1-4}\addlinespace[0.3pt]
     Noncoherent
     & Noncoherent
     & Coherent
     & Carrier-phase-based 
     \\
     \, \hspace{2mm}
     \ref{pgf:PEB-classic} 
     & \ref{pgf:PEB-nc} 
     & \ref{pgf:PEB-c}
     & \ref{pgf:PEB-cp}
\end{tabular}};

\node[above left, align=right,font={\footnotesize},fill=white,
opacity=0.9,
inner sep=0.75mm, 
xshift=0.65mm, 
yshift=-0.65mm, 
draw, 
line width = 0.8pt,
minimum width = 3mm,
minimum height = 3mm,
] at %
(boundaryLeft.south east)
{a)};

\node[above left, align=right,font={\footnotesize},fill=white,
opacity=0.9,
inner sep=0.75mm, 
xshift=0.65mm, 
yshift=-0.65mm, 
draw, 
line width = 0.8pt,
minimum width = 3mm,
minimum height = 3mm,
] at %
(boundaryRight.south east)
{b)};

\end{tikzpicture}%